\documentclass[prx,aps,twocolumn,
preprintnumbers,superscriptaddress,longbibliography]{revtex4} 
\usepackage{amsmath,amssymb,amsfonts,bm,amsthm,bbm,accents}
\usepackage{mathrsfs}
\usepackage{mathalfa}
\usepackage{braket}
\usepackage{graphicx}
\usepackage{subcaption}
\usepackage{mathtools}
\usepackage{color}
\usepackage{caption}
\usepackage{xcolor}
\usepackage[colorlinks,linkcolor=blue,citecolor=blue]{hyperref}
\usepackage{comment}
\usepackage{physics}
\hypersetup{
colorlinks=true, 
linkcolor=violet,  
urlcolor=blue,    
anchorcolor=green, 
citecolor=violet   
}
\usepackage{enumitem}
\usepackage{shuffle}
\usepackage{tabularx}
\usepackage{caption}

\newcommand{\as}{\alpha}
\newcommand{\ri}{\RN{1}}
\newcommand{\rii}{\RN{2}}
\newcommand{\ii}[3]{\int_{#1}^{#2}\mathrm{d}#3\,}

\newcommand{\id}{\mathbbm{1}}
\newcommand{\cl}{\mathrm{cl}}
\newcommand{\cc}{\mathrm{c}}
\newcommand{\s}{\mathrm{s}}
\newcommand{\rr}{\mathrm{r}}

\newcommand{\oZ}{\accentset{\circ}{Z}}

\newcommand{\eq}[1]{Eq.\,(\ref{#1})}

\newcommand{\eqtwo}[2]{Eqs.\,(\ref{#1}) and (\ref{#2})}

\newcommand{\term}[1]{{\Large\textcircled{\small $#1$}}}
\newcommand{\terms}[2]{\raisebox{-0.18ex}{\Large\textcircled{\small $#1$}}\,\text{in \eq{#2}}
}

\newtheorem{theorem}{Theorem}
\newtheorem{lemma}{Lemma}
\newtheorem{proposition}{Proposition}

\newtheorem{definition}{Definition}
\newtheorem{corollary}{Corollary}

\newtheorem{remark}{Remark}

\usepackage[inline,marginal,final]{showlabels}

\newcommand{\RN}[1]{%
\textup{\uppercase\expandafter{\romannumeral#1}}%
}

\begin{document}
\title{Toward a Complexity Classification of High-Temperature Bosons: Computational Tractability and Power-Law Clustering}

\author{Xin-Hai Tong*}
\affiliation{Department of Physics, The University of Tokyo, 5-1-5 Kashiwanoha, Kashiwa-shi, Chiba 277-8574, Japan}
\email[]{xinhai@iis.u-tokyo.ac.jp}
\author{Tomotaka Kuwahara$^{\dagger}$}
\affiliation{Analytical Quantum Complexity RIKEN Hakubi Research Team,
RIKEN Center for Quantum Computing (RQC), 2-1 Hirosawa, Wako, Saitama 351-0198, Japan}
\affiliation{RIKEN Pioneering Research Institute (PRI), Wako, Saitama 351-0198, Japan}
\affiliation{PRESTO, Japan Science and Technology (JST), Kawaguchi, Saitama 332-0012, Japan}
\email[]{tomotaka.kuwahara@riken.jp}


\begin{abstract}
Determining when quantum many-body systems admit simple, efficiently simulable structure is a central problem. High-temperature thermal states are a natural candidate for such simplicity, yet for bosons, the unbounded local Hilbert space and energy invalidate the usual expectation that large $T$ guarantees tractability. Here we investigate the resulting complexity boundary for interacting lattice bosons and show that the repulsive Bose--Hubbard class lies on the ``simple'' side. 
For a family with long-range hopping decaying as $r^{-\alpha}$, we prove convergence of a controlled cluster expansion, which implies (above an explicit temperature threshold) an efficient classical algorithm to approximate the partition function and a rigorous power-law clustering bound for connected correlations. More broadly, our results provide a first step toward charting complexity boundaries for high-temperature bosons and suggest the repulsive Bose--Hubbard class as a natural candidate cusp.
\end{abstract}

\maketitle

\paragraph*{Introduction.---}
A central question in Hamiltonian complexity is whether one can efficiently simulate quantum many-body systems---for instance, whether expectation values of physical observables can be computed in time polynomial in the system size and the desired precision \cite{Kitaev2002,Gharibian2015,Osborne2012}. 
One prominent class of states believed to admit such efficient descriptions is given by thermal equilibrium states away from criticality. 
For spin and fermionic lattice systems, the ``noncritical'' regime is most transparently realized in one dimension, and it is also well understood in higher dimensions at sufficiently high temperatures above a threshold. 
Indeed, a large body of work has developed efficient methods in these settings, ranging from rapidly mixing sampling schemes \cite{Ding2025,Anshu2021,Kastoryano2016,Temme2011,Rouze2024} to deterministic algorithms based on convergent cluster expansions \cite{mann2024algorithmic,sanchez2025high,molnar2015approximating,kliesch2014locality,Haah2022}. 
These results have contributed to the widespread expectation that, at least for finite-dimensional local Hilbert spaces, sufficiently high temperature tends to place thermal states in a computationally tractable regime.

Most strikingly, recent work has identified explicit loopholes to ``high-$T$ simplicity" in lattice systems without a local energy bound: bosonic number fluctuations can stabilize entropically ordered phases even as $T\to\infty$ (i.e., $\beta \to 0$ with $\beta$ being the inverse temperature) \cite{Han2025} (see also \cite{yb7d-6tvc}). These developments motivate a re-examination of the conventional expectation that ``high temperature implies simplicity" beyond the models characterized by finite-dimensional local Hilbert spaces.

\begin{figure}[tt]
\centering
\includegraphics[width=\linewidth]{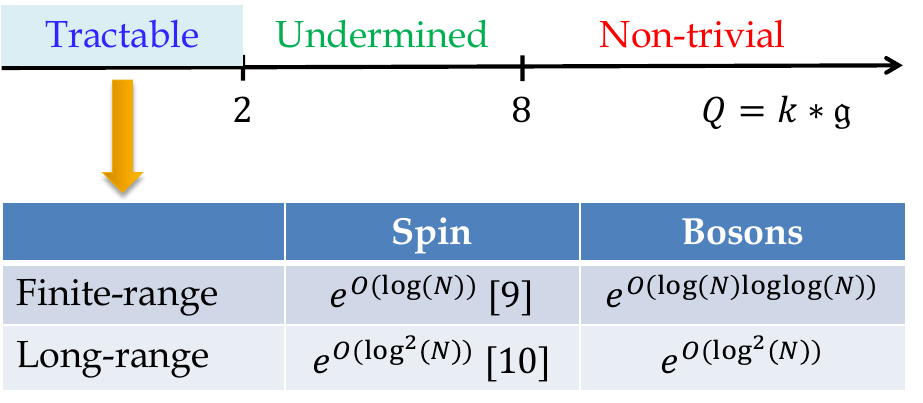}
\captionsetup{justification=raggedright,singlelinecheck=false}
\caption{Computational complexity of high-temperature partition-function approximation to error $1/\mathrm{poly}(N)$.
We introduce the structural parameter $Q:=k\mathfrak{g}$ (interaction locality $\times$ bosonic degree) governing the growth of cluster weights.
We prove controllability and efficient classical approximation for $Q\le2$, while $Q>2$ can already invalidate cluster-expansion--based control.
Known high-temperature nontrivial bosonic phases (entropic order~\cite{Han2025}) occur only for $Q\ge8$, leaving an open intermediate window.
Here $k$ is the maximal support size of interaction terms, and $\mathfrak{g}$ is the maximal on-site bosonic degree in off-site monomials (number of creation/annihilation operators acting on one site).
Theorems~\ref{theorem_boson_lr_complexity} and~\ref{theorem_boson_fr_complexity} state our long- and finite-range bosonic results; the table contrasts them with state-of-the-art spin-system bounds.
}
\label{tab_1}
\end{figure}

The goal of this work is to determine to what extent genuinely bosonic systems can still exhibit trivial behaviors at high temperature, and to establish a principled criterion that separates tractable regimes from pathological ones (cf.\,\cite{Han2025}). 
We focus on a broad class of Bose-Hubbard-type models,
which is of both fundamental theoretical importance \cite{bloch2008many,cazalilla2011one,kollath2007quench,fisher1989boson,freericks1996strong} and direct experimental relevance \cite{jaksch1998cold,greiner2002quantum,stoferle2004transition,bakr2009quantum,fallani2007ultracold,bakr2010probing}. The underlying local Hilbert space is infinite-dimensional, and the Hamiltonian involves unbounded operators. To capture the physics of modern experimental platforms, we explicitly incorporate long-range hopping whose strength decays as $r^{-\alpha}$ with respect to distance $r$.
Such form of couplings are intrinsic to systems like dipolar molecules~\cite{recati2023supersolidity,klaus2022observation,bigagli2024observation} and trapped ions~\cite{foss2024progress,defenu2023long,monroe2021programmable}, and lead to distinct thermodynamic behaviors absent in short-range models \cite{campa2014physics,campa2009statistical,defenu2023long,defenu2024out}, including ensemble inequivalence~\cite{kuwahara2020gaussian,barre2001inequivalence,kastner2010nonequivalence,russomanno2021quantum} and the emergence of quasi-stationary states~\cite{russomanno2021quantum,latora2002dynamical,campa2017quasilinear,chavanis2006quasi}.

For long-range bosonic settings, a key obstacle is whether one can establish a \emph{convergent} cluster expansion: when arbitrarily many bosons can occupy a single site, naive combinatorial bounds on cluster weights typically diverge, preventing the standard high-temperature expansion from being controlled. Our convergence criterion appears optimal for the repulsive Bose–Hubbard model, while several seemingly mild extensions [e.g., density-assisted tunneling, $a_i(n_i+n_j)a_j^\dagger+\text{h.c.}$] fall outside its scope and may allow genuinely nontrivial behavior even at high temperature. 

In this Letter, we establish two main results for the repulsive Bose-Hubbard model in the high-temperature regime. 
First, we construct an efficient simulation algorithm based on a \emph{convergent} cluster expansion. 
Second, we prove a clustering theorem that quantifies the decay of correlations. 
Our algorithm runs in quasi-polynomial time: for short-range hopping, the cost scales as $\exp(\widetilde{\mathcal{O}}(\log(N)))$, while for long-range hopping it scales as $\exp(\mathcal{O}(\log^2(N)))$, where $N$ is the system size, and the precision is set to be $1/{\rm poly}(N)$.
A key difference from spin systems~\cite{mann2024algorithmic,sanchez2025high,molnar2015approximating,kliesch2014locality,achutha2025provably} is that bosonic models have neither a bounded local Hilbert space nor an a priori energy cutoff, requiring a different analysis to control approximation errors (see Fig.~\ref{tab_1} for a comparison).
Nevertheless, in the long-range setting, our runtime scaling matches the best-known results for long-range spin systems~\cite{sanchez2025high,achutha2025provably}.


A central technical challenge is the appearance of high powers of bosonic operators in the expansion (e.g., $a_i^m$). 
Naively bounding such terms yields norms that grow like $(m\beta)^m$ at $m$th order, leading to a divergence of the cluster series as $m\to\infty$ even at small $\beta$. Our key contribution is to show that, in the repulsive Bose--Hubbard setting, this potentially fatal growth can be controlled by exploiting the structure of the on-site interaction, thereby restoring convergence (see Supplemental Material \cite{sm} for the detailed estimates). 
This criterion also clarifies why certain extensions---such as extended Hubbard interactions, density-assisted tunneling terms, or quartic inter-site couplings of the form $n_i^2 n_j^2$ as in Ref.~\cite{Han2025}---may fall outside the reach of cluster-expansion-based methods and can support nontrivial behavior at high temperatures.

Our results can be viewed as a first step toward a long-term goal: identifying the precise boundary between trivial and nontrivial behavior in high-temperature bosonic matter. 
By delineating when cluster-expansion techniques remain controllable---and when they fundamentally break down due to bosonic unboundedness---we provide a concrete framework for understanding which microscopic ingredients can obstruct ``high-temperature simplicity'' and thereby enable genuinely nontrivial phenomena. We expect this perspective to serve as a useful guide for exploring and classifying high-temperature phases and dynamical regimes in a broad range of interacting bosonic systems.

\paragraph*{Setup and notations.—}
We consider a system of bosons on a finite $D$-dimensional lattice, denoted by $V$. At each site $i \in V$, we equip it with an infinite-dimensional Hilbert space $\mathcal{H}_i$. Then the total Hilbert space is given by $\mathcal{H} = \bigotimes_{i \in V} \mathcal{H}_i$. Though our methodology is applicable to a much wider scope and cases, for readers’ convenience, let us focus on the long-range canonical Bose-Hubbard model,
\begin{equation}\label{lr_bose_hubbard}
H = -\sum_{i \neq j \in V} J_{i,j} (a_i^\dagger a_j + \text{h.c.}) + \sum_{i \in V} \left[ \frac{U_i}{2} n_i(n_i - 1) - \mu_i n_i \right].
\end{equation}
Here, $a_i$ and $a_i^\dagger$ are the bosonic annihilation and creation operators at site $i$, satisfying the canonical commutation relations $[a_i, a_j^\dagger] = \delta_{ij}$, and $n_i := a_i^\dagger a_i$ is the local number operator. The hopping amplitudes $J_{i,j}$ exhibit a power-law decay with distance $d_{i,j}$,
\begin{equation}\label{lr_condition}
|J_{i,j}| \leq \frac{g}{(1 + d_{i,j})^\alpha},
\end{equation}
where $g > 0$ is the hopping strength and $\alpha$ is the decay exponent. The distance $d_{i,j}$ between two sites $i,j \in V$ is defined as the length of the shortest path connecting them. We remark that it is not hard to see that one can obtain the finite-range Bose-Hubbard model by replacing the long-range condition with the following finite-range counterpart:
\begin{equation}\label{fr_condition}
\begin{aligned}[b]
|J_{i,j}|\leq g\Theta(d_{c}-d_{i,j})
\end{aligned}
\end{equation}
with $\Theta(\bullet)$ being the Heaviside step function and $d_{c}=\mathcal{O}(1)$ being the cutoff hopping range.

The well-definedness of the thermal state requires uniform bounds on the on-site interaction strength $U_i$ and the chemical potential $\mu_i$:
\begin{equation}\label{eq:param_bounds}
|\mu_i| \leq \mu < \infty, \quad 0 < U_{\min} \leq U_i \leq U_{\max} < \infty,
\end{equation}
for all sites $i \in V$. The on-site potential can be written as $W = \sum_{i \in V} W_i(n_{i})$, where $W_i(n_{i}) = U_i n_i^2 / 2 - \mu_{0,i} n_i$ with the effective chemical potential $\mu_{0,i} := \mu_i + U_i/2$. As a standard constraint, we also require $\alpha>D$ to ensure the thermodynamical stability \cite{kim2025thermal,sanchez2025high}.

Then we define the quantum Gibbs state $\rho_{\beta}$ at inverse temperature $\beta$ as $\rho_{\beta}\coloneq e^{-\beta H}/\mathcal{Z}_{\beta}$, with $\mathcal{Z}_{\beta}\coloneq \operatorname{Tr}e^{-\beta H}$ being the partition function. 

We investigate two main objectives.
First, we show that, in the high-temperature regime, the partition function $\mathcal{Z}_{\beta}$ can be simulated efficiently.
At the same time, we discuss which complexity class captures the complexity boundary for bosonic systems.
Second, we prove two-point clustering for long-range interacting systems: the correlation function
\begin{equation}\label{cor}
\begin{aligned}[b]
C_{\beta}(O_{X},O_{Y})\coloneq \operatorname{Tr}(\rho_{\beta}O_{X}O_{Y})
-\operatorname{Tr}(\rho_{\beta}O_{X})\cdot \operatorname{Tr}(\rho_{\beta}O_{Y})
\end{aligned}
\end{equation}
decays rapidly with the distance between the supports of any two local operators $O_X$ and $O_Y$.

\paragraph*{Computational complexity for long-range bosons.---}
We show that even long-range bosonic systems can be simulated efficiently on a classical computer. This finding, which stands in contrast to the expected difficulty of such problems, is rigorously established in our first main result.
\begin{theorem}\label{theorem_boson_lr_complexity}
For the long-range Bose-Hubbard model defined by Eqs.~\eqref{lr_bose_hubbard} and \eqref{lr_condition} with decay exponent $\alpha > D$ on a finite lattice $V$. Denote $N=|V|$, then there exists a classical algorithm that computes an approximation $f_{\beta}$ satisfying
$|\log \mathcal{Z}_{\beta}-f_{\beta}|\leq 1/\operatorname{poly}(N)$ 
at any fixed temperature above a threshold, $\beta\leq \beta_{c}(\as, D, g, \mu, U_{\min},U_{\max})=\mathcal{O}(1)$. The runtime of the algorithm is quasi-polynomial in the system size, given by $$\exp(\mathcal{O}(\log^{2} (N/\epsilon))).$$
\end{theorem}

The key to this algorithm is the insight that the partition function can be mapped to the partition function of a polymer model, whose convergence is controlled by an abstract cluster expansion~\cite{mann2024algorithmic,kotecky1986cluster}. Our main technical innovation is a method to satisfy the stringent Koteck\'{y}-Preiss convergence condition \cite{kotecky1986cluster} for this model. We achieve this by deploying the recently developed interaction-picture cluster expansion~\cite{tong2024boson}, combined with new techniques for bounding summations over desired clusters. A detailed proof sketch is provided at the end of this Letter, with the full derivation available in the Supplemental Material~\cite{sm}.

It is instructive to compare this result with recent progress on long-range spin systems~\cite{sanchez2025high}, where a similar quasi-polynomial complexity for computing the partition function was also established.
Despite this apparent equivalence in computational complexity, our findings highlight a sharp distinction between bosonic systems and their spin counterparts.
Specifically, Theorem~\ref{theorem_boson_lr_complexity} reveals that the total error in our approximation contains an additional, non-vanishing term that decays polynomially with the system size.
Consequently, in stark contrast to the spin case, the total error cannot be made arbitrarily small.
As we detail in the appendix and the Supplementary Material~\cite{sm}, this additional error term is a direct consequence of truncating the infinite-dimensional on-site bosonic Hilbert space to a finite dimension.
This truncation is a necessary step to render the problem computationally tractable, converting the trace of an operator on an infinite-dimensional space into a finite matrix computation. A rigorous low-density inequality guarantees the validity of our truncation (see discussion below).
Therefore, this error is an inevitable feature of the approximation scheme for such bosonic systems and
fundamentally reflects the absence of a Pauli exclusion principle. 

Simulating bosonic systems is challenging because their Hamiltonians are represented by infinite-dimensional matrices, rendering exact calculations intractable. A standard approach is to truncate the local Hilbert space by imposing a cutoff on the maximum occupation number. This is motivated by the physical intuition that high-energy states are negligibly populated at finite temperatures. However, a rigorous justification for this truncation, which guarantees accuracy, has been lacking. Our results provide this justification, establishing a rigorous framework for truncating bosonic systems that ensures both computational efficiency and controlled accuracy. 

Despite the aforementioned difference in the error bounds, it is remarkable that the time complexity with respect to the system size presented in Theorem~\ref{theorem_boson_lr_complexity} coincides with that of its spin counterpart.
This efficiency stems from a key physical principle: effective locality.

\paragraph*{Complexity boundary.---}
A fundamental source of complexity in lattice bosonic systems is the absence of the Pauli exclusion principle, which allows for arbitrary particle accumulation on a single site. This motivates us to examine the structure of local terms in the bosonic Hamiltonian, considering the general $k$-local form
\begin{equation}\label{eq:hamiltonian_form}
H_{0}=\sum_{Z\subset V: 1<|Z|\leq k}h_{Z}+\sum_{i\in V}W_{i}(n_{i}).
\end{equation}
Here, $h_{Z}$ represents the off-site term acting on the subspace $\bigotimes_{i\in Z}\mathcal{H}_{i}$, and $W_{i}(n_{i})$ denotes the on-site potential.
To quantitatively gauge the potential pathology of the system, we introduce a structural index $\mathfrak{g}$ based on the bosonic degree of the interactions. 
For each constituent monomial of the off-site term $h_Z$, we define its \emph{local bosonic degree} as the maximum number of bosonic operators (creation plus annihilation) acting on any single site within its support. 
For instance, for distinct sites $i, j$, the monomials $a_{i}^{\dagger}a_{j}$, $n_{i}n_{j}$, and $a_{i}n_{i}a_{j}^{\dagger}$ have local bosonic degrees of $1$, $2$, and $3$, respectively. We then define the \emph{global bosonic degree}, denoted by $\mathfrak{g}$, as the maximum local bosonic degree across all monomials in the off-site terms.

As detailed in the Supplemental Material~\cite{sm} (specifically Proposition 4, Lemma 16, and Remark 12), our analysis establishes that the convergence of the cluster expansion requires the condition:
\begin{equation}\label{condition_cc}
k\mathfrak{g}\leq 2.
\end{equation}
This rigorous condition reveals why the canonical Bose-Hubbard model represents a boundary case within the tractable regime. For the canonical Bose-Hubbard model, the condition is satisfied with $(k,\mathfrak{g})=(2,1)$. In contrast, for interactions involving higher powers of number operators, such as $n_i^{2} n_j^{2}$ with \cite{Han2025}, we have $k=2$ but $\mathfrak{g}=4$. This clearly violates the condition ($k\mathfrak{g} \ge 8$), consistent with the reported nontrivial phenomena even at high temperatures.

When condition \eqref{condition_cc} is violated (i.e., $k\mathfrak{g}>2$), the convergence breaks down. As explicitly demonstrated by a counter-example constructed in Remark 12 of the Supplemental Material~\cite{sm}, the combinatorial weight of clusters scales super-exponentially with cluster size. This divergence suggests that the Bose-Hubbard model represents a ``cusp'' of computational tractability, beyond which genuine complexity and nontrivial high-temperature phenomena may emerge.

\paragraph*{Short-range cases.---}
To understand the role of coupling range and its effect on locality, it is physically insightful to contrast the long-range results with their finite-range counterparts. This comparison reveals how the transition from strictly local to non-local couplings affects the system's computational complexity.

Our first point of comparison is the classical simulation of the partition function. For the finite-range Boson-Hubbard model, we establish the following result.

\begin{theorem}\label{theorem_boson_fr_complexity}
For the finite-range Bose-Hubbard model defined by Eqs.~\eqref{lr_bose_hubbard} and \eqref{fr_condition} on a finite lattice $V$. Denote $N=|V|$, then there exists a classical algorithm that computes an approximation $f_{\beta}$ satisfying
$|\log \mathcal{Z}_{\beta}-f_{\beta}|\leq 1/\operatorname{poly}(N)$ 
at any fixed temperature above a threshold, $\beta\leq \beta_{c}(d_{c}, D, g, \mu, U_{\min},U_{\max})=\mathcal{O}(1)$. The runtime of the algorithm is almost polynomial in the system size, given by $$\exp(\mathcal{O}(\log((N/\epsilon)\log\log (N/\epsilon)))).$$
\end{theorem}

This runtime for the finite-range Bose-Hubbard model is a significant improvement over that of its long-range counterpart. Indeed, this constitutes a quasi-polynomial runtime. This is because the $\log\log(N/\epsilon)$ factor grows exceedingly slowly with the system size $N$, rendering the algorithm's complexity nearly polynomial for all practical purposes. This improvement originates from the different topological structures of the models' interaction graphs. To see this, we represent the system as a simple undirected graph $(V, \mathcal{E})$, where the vertices $V$ are the lattice sites and the edges $\mathcal{E}$ correspond to non-zero couplings. For the finite-range model, an edge connects sites $i$ and $j$ only if their distance $d_{i,j}$ is less than a constant cutoff $d_c$. In contrast, the long-range model is represented by a highly connected graph where most sites are coupled. The crucial difference lies in the graph's maximum degree, which is defined as the largest number of edges connected to any single site. In the finite-range case, this is an $\mathcal{O}(1)$ constant (e.g., equals $2D$ for a $D$-dimensional square lattice with nearest-neighbor hopping), whereas for the long-range case, it grows with the system size $N$. This topological property directly governs the complexity of enumerating polymers in the cluster expansion at the core of our algorithm, thus leading to the substantially faster runtime in the finite-range case.

It is instructive to contrast the almost polynomial runtime scaling in Theorem~\ref{theorem_boson_fr_complexity} with the rigorously polynomial scaling for finite-range spin systems, as established in Theorem 21 of Ref.\,\cite{mann2024algorithmic}. As will be detailed in Supplementary Material \cite{sm}, the extra factor $\log\log (N/\epsilon)$ in our result for bosons is a direct consequence of the infinite-dimensional local Hilbert space. To overcome this infinity, our algorithm must truncate the local boson number to a maximum $q \propto \log N$. This choice of $q$ is sufficient to bound the truncation error to be smaller than $\mathrm{poly}(N)^{-1}$. In contrast, for a quantum spin system, the local Hilbert space is finite-dimensional, corresponding to an effective cutoff $q=\mathcal{O}(1)$. This obviates the need for a system-size-dependent truncation, thereby eliminating the extra factor and ensuring a polynomial runtime.

\paragraph*{Clustering theorem.---}
Although the system interacts at long range, the condition $\alpha > D$ ensures that correlations still decay sufficiently rapidly.
In our second main result, we provide a rigorous quantification of this physical picture (see Figure~\ref{fig_1} for an intuitive illustration).

\begin{figure}[tt]
\centering
\includegraphics[width=0.9\linewidth]{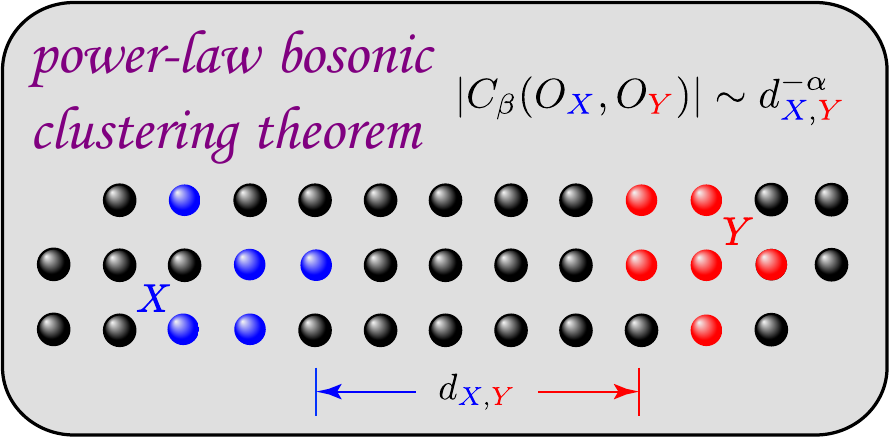}
\captionsetup{justification=raggedright,singlelinecheck=false}
\caption{Illustration of the clustering theorem for long-range bosonic systems. Summarized as Theorem \ref{theorem_boson_lr_clustering}, the magnitude of the correlation function $C_{\beta}(O_{X},O_{Y})$ exhibits a power-law decay with distance $d_{X,Y}$ at high temperatures.}
\label{fig_1}
\end{figure}

\begin{theorem}\label{theorem_boson_lr_clustering}
For the long-range Bose-Hubbard model defined by Eqs.~\eqref{lr_bose_hubbard} and \eqref{lr_condition} with $\alpha > D$ on a finite lattice $V$. Let $O_X$ and $O_Y$ be operators formed by finite products of creation and annihilation operators supported on disjoint finite subsets $X, Y \subset V$. Then, for any temperature above a threshold, $\beta \leq \beta_{c}(\as, D, g, \mu, U_{\min},U_{\max})= \mathcal{O}(1)$, the correlation function is bounded by:
\begin{equation}\label{re_cor_lr_boson}
|C_{\beta}(O_{X},O_{Y})| \leq C^{|X|+|Y|} \frac{\Phi_{X,Y}(\beta)}{(1+d_{X,Y})^{\alpha}},
\end{equation}
where $C$ is an $\mathcal{O}(1)$ constant, $d_{X,Y} = \min_{i \in X, j \in Y} d_{i,j}$, and $|L|$ is the number of sites in region $L\subset V$. The function $\Phi_{X,Y}(\beta)$ is defined as $\Phi_{X,Y}(\beta) := [\mathsf{N}(O_{X})!]^{1/2} [\mathsf{N}(O_{Y})!]^{1/2}\beta^{-[\mathsf{N}(O_{X})+\mathsf{N}(O_{Y})]/4}$, where $\mathsf{N}(O)$ is the total number of creation and annihilation operators constituting the operator $O$.
\end{theorem}

Notably, the temperature scaling of the bound in Eq.\,\eqref{re_cor_lr_boson} is optimal. This can be readily demonstrated in a simple special case: the on-site Bose-Hubbard model, which is recovered by setting all hopping terms $J_{i,j}=0$ in Eq.\,\eqref{lr_bose_hubbard}. In this regime, a straightforward estimation by integration shows that the local particle number fluctuations, quantified by the variance $C_{\beta}(n_{i},n_{i})=\langle n_{i}^{2} \rangle_{\beta W}-\langle n_{i}\rangle^{2}_{\beta W}$, scale as $\beta^{-1}$ at high temperatures. Noticing that $\mathsf{N}(n_{i}^{2})=2\mathsf{N}(n_{i})=4$, this result saturates the bound presented in Theorem~\ref{theorem_boson_lr_clustering} for the zero-distance case ($d_{i,i}=0$), thus confirming its optimality with respect to $\beta$ at high temperatures.
We emphasize that this theorem extends to operators $\widetilde{O}_X$ that are general polynomials in $\{a_i^\dagger, a_i\}_{i \in X}$. Such an operator can be decomposed into a sum of monomials, $\widetilde{O}_X = \sum_{\ell} O_{X,\ell}$, where each monomial satisfies the conditions of Theorem~\ref{theorem_boson_lr_clustering}. The correlation function is then bounded via the triangle inequality.

The established clustering of correlations has profound implications, serving as a cornerstone for proving fundamental properties of the system~\cite{kim2025thermal,brandao2013area,kuwahara2020eigenstate,molnar2015approximating}. 
First, it rigorously implies a \emph{thermal area law} for mutual information. 
Specifically, for any bipartition $V=A\sqcup B$, we establish that $\mathcal{I}(A:B)\leq C \beta^{1/2} |\partial_{AB}|$ (see Supplemental Material~\cite{sm} for the proof). 
This result extends the information-theoretic understanding of quantum correlations---previously initiated for spins~\cite{wolf2008area,kim2025thermal} and short-range bosons~\cite{lemm2023thermal,tong2024boson}---to the challenging long-range bosonic regime.

Second, the interaction-picture formalism developed for our proof yields a rigorous \emph{low-boson-density inequality}. 
We prove in Supplemental Material~\cite{sm} that the local moments satisfy the optimal scaling $\operatorname{Tr}(n_{i}^{\ell}\rho_{\beta})\leq C^{\ell}\ell! \beta^{-\ell/2}$. 
While generalizing the finite-range result of Ref.\,~\cite{tong2024boson}, this estimate provides a crucial technical ingredient for broader studies of bosonic systems, including Lieb-Robinson bounds~\cite{kuwahara2021lieb,kuwahara2024effective} and the existence of Gibbs states in the thermodynamic limit~\cite{deuchert2025dynamics}. The temperature scaling $\beta^{-\ell/2}$ here is optimal, as can be easily verified in the simple on-site Bose-Hubbard model by considering $ \langle n_{i}^{\ell} \rangle_{\beta W}$.

\paragraph*{Summary and outlook.---}
High-temperature thermal states are often viewed as the prototypical ``easy'' regime of many-body physics; for bosons, however, the lack of both a local Hilbert-space bound and an energy cutoff makes this intuition far from obvious.
Here we show that long-range Bose--Hubbard systems nonetheless exhibit a robust form of high-temperature simplicity: above an explicit threshold, correlations cluster with optimal scaling, and the partition function admits an efficient classical approximation.

Technically, this is enabled by an interaction-picture cluster expansion tailored to unbounded bosonic operators, which restores convergence precisely where conventional methods diverge.
The resulting framework yields a rigorous complexity analysis for thermal equilibrium of both finite- and long-range bosons and provides a quasi-polynomial (nearly polynomial for finite-range hopping) algorithm for approximating $\log \mathcal{Z}_\beta$ to inverse-polynomial accuracy.
It also reveals an intrinsic obstruction: once the structural parameter $Q=k\mathfrak{g}$ exceeds $2$, cluster weights grow combinatorially, suggesting the possibility of genuinely nontrivial high-temperature bosonic phases and placing the canonical long-range Bose--Hubbard model at the cusp between tractability and pathology.

More broadly, our results point to a \emph{complexity boundary} for thermal bosonic matter.
Pinpointing the optimal threshold $Q^\ast$ and classifying models in its vicinity constitute natural open problems of direct complexity-theoretic significance.

Looking forward, it will be important to connect our algorithmic guarantees to practical representations of thermal states, such as tensor-network approximations, and to extend the framework toward dynamical simulation and nonequilibrium physics.
This includes constructing matrix-product operator approximations of thermal states~\cite{molnar2015approximating,kliesch2014locality} and advancing methods for simulating the time-evolution of quantum many-body systems~\cite{wild2023classical,mizuta2025trotterization}.
On the experimental front, long-range bosonic platforms---most notably dipolar Bose--Einstein condensates~\cite{bigagli2024observation,klaus2022observation}---offer an exciting opportunity to directly test our predictions for correlation clustering and the thermal area law~\cite{recati2023supersolidity}.
More broadly, we hope that the perspective developed here provides a systematic route to identifying the true boundary between trivial and nontrivial behavior in high-temperature bosonic matter.

\paragraph*{Acknowledgements.---} 
X.-H.T. thanks Prof.\,Naomichi Hatano, Prof.\,Zongping Gong, and Ao Yuan for
fruitful discussions. X.-H.T. was supported by the FoPM,
WINGS Program, the University of Tokyo. T. K. acknowledges the Hakubi projects of RIKEN. T. K. was supported by JST PRESTO (Grant No.
JPMJPR2116), ERATO (Grant No. JPMJER2302), and JSPS Grants-in-Aid for Scientific Research (No. JP23H01099, JP24H00071), Japan.

\nocite{Bibtexkey}
\bibliographystyle{unsrt}
\bibliography{ref}

\clearpage
\begin{center}
\textbf{\large End Matter}
\end{center}
We demonstrate the sketch and idea of the proof for Theorems \ref{theorem_boson_lr_complexity} and \ref{theorem_boson_lr_clustering} for readers' convenience. For detailed discussion on the setup, notations and full proof, see Supplementary Material \cite{sm}. 

\paragraph*{Proof sketch of Theorem \ref{theorem_boson_lr_complexity}.---}
We present the construction of the efficient classical algorithm described in Theorem \ref{theorem_boson_lr_complexity}. The algorithm proceeds by recasting the partition function as that of an abstract polymer model.

We first briefly introduce some basic concepts, and the detailed version is covered in Sec.\,S.V. of Supplementary Material \cite{sm}. The abstract polymer model is defined as a collection of three key elements: a countable set $\mathfrak{C}$  whose elements are called polymers and denoted by $\gamma$, a function $w$ called weight maps the polymer to a complex number, and a symmetric compatibility relation $\sim$ such that $\gamma\nsim \gamma$. The sets of 
pairwise compatible polymers are called admissible sets. The collection of all admissible sets is denoted by $\mathscr{G}$. These two concepts are essential for the definition of abstract polymer partition function,
\begin{equation}\label{abs_paritition_function_em}
\mathcal{Z}(\mathfrak{C},w)\coloneqq \sum_{\Gamma\in \mathscr{G}}\prod_{\gamma\in \Gamma}w_{\gamma}.
\end{equation}
For any ordered sequence of polymers denoted by $a$ we can define the associated incompatibility graph, which is a simple graph with vertex set being distinct polymers in $a$ and edges between any two polymers $\gamma_{i},\gamma_{j}\in a$ if and only if $\gamma_{i}\nsim \gamma_{j}$. Let $\mathscr{G}_{\mathrm{c}}$ denote the collection of all finite, ordered sequences of polymers whose incompatibility graph is connected, and then the logathrithm partition function can be written as 
\begin{equation}\label{log_abstract_polymer_em}
\log \mathcal{Z}(\mathfrak{C},w) = \sum_{a\in \mathscr{G}_{\mathrm{c}}} \varphi(\mathsf{G}_{a}) \prod_{\gamma\in a} w_{\gamma},
\end{equation}
where $\varphi(\mathsf{G})$ is the Ursell function for a graph $\mathsf{G}$. 

The exact computation of the partition function $\mathcal{Z}_{\beta}$ is intractable, as the bosonic Hamiltonian acts on an infinite-dimensional Hilbert space. We therefore introduce a truncation by restricting the on-site boson number to a maximum of $q$. The corresponding truncated partition function is defined as $\mathcal{Z}_{\beta}^{(q)} := \operatorname{Tr}\qty(\Pi_{V,q}e^{-\beta H})$.
Here, $\Pi_{i,q}$ is the projection operator onto the space such that the boson number at the site $i$ is smaller than $q$ while we denote $\Pi_{L,q}=\bigotimes_{i\in L}\Pi_{i,q}$ for any region $L\subseteq V$. 
According to Corollary \ref{corollary_truncation_error} of the Supplementary Material \cite{sm}, we can choose $q\propto \log |V|$ such that $|\log \mathcal{Z}_{\beta}-\log \mathcal{Z}^{(q)}_{\beta}|\leq |V|^{-\theta}$. This truncation estimation is non-trivial. One should first establish the low-boson-density inequality for any site $i\in V$, i.e., $\operatorname{Tr}(n_{i}^{\ell}\rho_{\beta})\leq C^{\ell}\cdot \ell!\cdot \beta^{-\ell/2}$ with $C=\mathcal{O}(1)$ independent from $i$ and $\ell$ but dependent on $\alpha, D, g,\mu,U_{\min}$ and $U_{\max}$, as summarized in Theorem \ref{stheorem_low_density} of the Supplementary Material \cite{sm}. This upper bound for the moments of local particle number has a significant implication called the bosonic concentration bound, summarized in Corollary \ref{coro_boson_concen}, which
implies an exponential suppression of states with high local boson number. Such a guarantee is fundamentally important, as it rigorously justifies the use of finite boson-number truncation schemes to approximate the partition function with controllable error. 

The next step is to connect the truncated partition function $\mathcal{Z}_{\beta}$ to the abstract polymer model. We first identify the polymer here as connected multisets, whose elements are drawn from $E$. Then we first define  $\mathcal{Z}^{(q)}_{W}\coloneq \operatorname{Tr}\qty(\Pi_{V,q}e^{-\beta W})$ and after some calculations we find the quantity $\mathcal{Z}_{\beta}^{(q)}/\mathcal{Z}^{(q)}_{W}$ can be formulated in the same form as Eq.\,\eqref{abs_paritition_function_em} with the weight given by 
\begin{equation}\label{boson_fr_weight_em}
\begin{aligned}[b]
w_{\gamma}=(-1)^{|\gamma|} \sum_{T \subseteq \gamma}(-1)^{|T|} \frac{\operatorname{Tr}_{V_{T}}\qty(\Pi_{V_{T},q}e^{-\beta H_{V_{T}}})}{\operatorname{Tr}_{V_{T}}\qty(\Pi_{V_{T},q}e^{-\beta W_{V_{T}}})}.
\end{aligned}
\end{equation}
Then we can also expand $\log [\mathcal{Z}_{\beta}^{(q)}/\mathcal{Z}^{(q)}_{W}]$ via Eq.\,\eqref{log_abstract_polymer_em}, which can be approximated by the truncated series 
\begin{equation}\label{t_m_em}
\begin{aligned}[b]
T_{m}\coloneq \sum_{a\in \mathscr{G}_{\mathrm{c}}: |a|\leq m} \varphi(\mathsf{G}_{a}) \prod_{\gamma\in a} w_{\gamma}.
\end{aligned}
\end{equation}
The absolute convergence of the cluster expansion in Eq.\,\eqref{log_abstract_polymer_em} is guaranteed by the Kotecký--Preiss criterion \cite{kotecky1986cluster}. This criterion, along with an analysis of the truncation error from Eq.\,\eqref{t_m_em}, is addressed in Proposition \ref{pro_truncation} of the Supplementary Material \cite{sm}, which establishes the bound $\qty| \log[\mathcal{Z}_{\beta}^{(q)}/\mathcal{Z}^{(q)}_{W}] - T_{m} | \leq |V| e^{-m}$.
Furthermore, Proposition \ref{pro_T_m} provides an estimate for the runtime required to compute $T_{m}$, given by $|V|^{\mathcal{O}(km)} \times \mathcal{O}((4keq)^{3km}) \times e^{\mathcal{O}(m)}$. This bound is derived by combining the complexities of cluster enumeration (Lemma \ref{lemma_cluster_list}), computation of the polymer weights $w_{\gamma}$ (Lemma \ref{lemma_weight}), and evaluation of the Ursell functions (Lemma \ref{lemma_ursell}).
By choosing the truncation order $m=\mathcal{O}(\log(N/\epsilon))$, we ensure that the approximation $f_{\beta} := \log\mathcal{Z}^{(q)}_{W} + T_{m}$ is computed with a total error of at most $\epsilon+|V|^{-\theta}$. Substituting this choice of $m$ into the runtime estimate, we find the overall computational complexity to be $e^{\mathcal{O}(\log^{2} (N/\epsilon))}$. This completes the proof once we choose $\epsilon=1/\operatorname{poly}(N)$.

\paragraph*{Proof sketch of Theorem \ref{theorem_boson_lr_clustering}.---}
We then present a sketch of the proof for the power-law clustering theorem for a long-range bosonic system. To employ the interaction-picture cluster expansion technique, we first reformulate the correlation function [cf.\,Eq.\,(\ref{cor})] in a more compact form, which is achieved by introducing the notations $O^{(+)}\coloneq O\otimes \id+\id \otimes O, O^{(0)}\coloneq O\otimes \id $ and $ O^{(1)}\coloneq\id \otimes O-O\otimes \id$ with $\id$ being the identity and write
\begin{equation}\label{cor_Stilde_em}
\begin{aligned}[b]
C_{\beta}(O_X,O_Y)=\frac{1}{\mathcal{Z}_{\beta}^{2}}\operatorname{Tr}\qty[e^{-\beta H^{(+)}}O_X^{(0)}O_Y^{(1)}].
\end{aligned}
\end{equation}
Conventional cluster expansion techniques \cite{kim2025thermal,wild2023classical} are not directly applicable to bosonic systems for two primary reasons. First, the Taylor series of the Boltzmann factor $e^{-\beta H}$ does not converge absolutely for Hamiltonians involving unbounded operators. Second, the final results in clustering theorems for spins/fermions (e.g., Theorem 2 in Ref.\,\cite{kim2025thermal} and Theorem 2 in Ref.\,\cite{kliesch2014locality}) rely on finite operator norms of local observables. These norms are typically infinite in the bosonic case. For instance, the local particle number operator effectively satisfies $\|n_{i}\|=\infty$.
To overcome these difficulties, we employ an interaction-picture cluster expansion. The approach begins by decomposing the Hamiltonian. We define the on-site part as $W \coloneq \sum_{i\in V} W_{i}(n_{i})$, with $W_{i}(n_{i}) = U_{i}n_{i}(n_{i}-1)/2 - \mu_{i}n_{i}$. The remaining off-site terms are then contained in the operator $I \coloneq H-W = \sum_{Z\in E}h_{Z}$. The method proceeds by applying the Dyson series to expand the Boltzmann factor 
\begin{equation}\label{inter_picture_em}
\begin{aligned}[b]
e^{-\beta H^{(+)}}=e^{-\beta W^{(+)}}\widetilde{S}(\beta),
\end{aligned}
\end{equation}
with
\begin{equation}\label{stil_em}
\begin{aligned}[b]
&\widetilde{S}(\beta)=\sum_{m=0}^{\infty}
\\&\sum_{Z_{1},Z_{2},...,Z_{m}
\in E}\ii{0}{\beta}{\tau_{1}}\ii{0}{\tau_{1}}{\tau_{2}}...\ii{0}{\tau_{m-1}}{\tau_{m}}\prod_{l=1}^{m}h_{Z_{l}}(\tau_{l})^{(+)}
\end{aligned}
\end{equation}
and $\bullet(\tau)\coloneq e^{\tau W}\bullet e^{-\tau W}$.
The absolute convergence of such an expansion is discussed in Proposition \ref{pro_abs_boson} of Supplementary Material \cite{sm}. We regard the collection of the regions drawn from $E$ as clusters. For a more compact presentation, one may regard the summation in Eq.\,\eqref{stil_em} as the summation over all ordered sequences with elements drawn from $E$ and denote it by $w$. Then by defining the function 
\begin{equation}\label{}
\begin{aligned}[b]
\widetilde{F}&\colon Z_{1}Z_{2}\cdots Z_{m} 
\mapsto \\&\ii{0}{\beta}{\tau_{1}}\ii{0}{\tau_{1}}{\tau_{2}}...\ii{0}{\tau_{|m|-1}}{\tau_{|m|}}\, \prod_{l=1}^{m}h_{Z_{l}}(\tau_{l})^{(+)},
\end{aligned}
\end{equation}
Eq.\,\eqref{stil_em} can be rewritten as $\widetilde{S}(\beta)=\sum_{w}\widetilde{F}(w)$. 

we put Eq.\,\eqref{inter_picture_em} into Eq.\,\eqref{cor_Stilde_em} to expand the correlation function.
Further analysis reveals that not all clusters contribute meaningfully to the correlation function in Eq.\,\eqref{cor_Stilde_em}. In fact, only those clusters that connect regions $X$ and $Y$ yield non-zero values, see Lemma \ref{lemma_disconnect_cluster} in Supplementary Material \cite{sm}. We usually use $G$ to denote a cluster, and the set of all clusters connecting $X$ and $Y$ is denoted by $\mathcal{G}_{\cl}(E)$ for convenience. Here, connectivity is defined at the collection level: we say regions $X$ and $Y$ are connected by a collection if the latter contains a sequence of pairwise overlapping regions where the first region overlaps with $X$ and the last overlaps with $Y$. This connectivity constraint significantly reduces the number of relevant clusters and provides the foundation for establishing the relationship between correlation functions and the separation distance between the two regions. 

After some arrangement, one will arrive at the following estimation
\begin{equation}\label{cor_beta_stech_em}
\begin{aligned}[b]
&|C_{\beta}(O_{X},O_{Y})|
\leq \sum_{m=0}^{\infty}\sum_{G\in \mathcal{G}_{\cl}(E): |G|=m}\qty(\frac{1}{\operatorname{Tr}_{\widetilde{V}_{G}}e^{-\beta W_{\widetilde{V}_{G}}}})^{2}
\\&\quad \times \sum_{w\in \mathcal{S}(G)}\norm{e^{-\beta W^{(+)}_{\widetilde{V}_{G}}}\widetilde{F}(w)O_{X}^{(0)}O_{Y}^{(1)}}_{1}.
\end{aligned}
\end{equation}
Here, we use $\mathcal{S}(G)$ to denote the collection of the sequences generated by the cluster $G$ and $\|\bullet\|_{1}$ is the trace norm. As an example, given $G=\{Z_1,Z_1,Z_2\}$, we have $\mathcal{S}(G)=\{Z_1 Z_1 Z_2,Z_1 Z_2 Z_1,Z_2 Z_1 Z_1\}$. In Eq.\,\eqref{cor_beta_stech_em} we use $V_{G}\coloneq\bigcup_{Z\in G}Z$ to denote the region corresponding to the cluster $G$ and $\widetilde{V}_{G}\coloneq V_{G}\cup X\cup Y$. We also write $W_{L}\coloneq \sum_{i\in L}W_{i}$ for any region $L\subset V$. A key step in bounding the expression in Eq.\,\eqref{cor_beta_stech_em} is to establish an upper bound on the trace norm of its right-hand side (RHS). This bound is provided in Lemma \ref{lemma_trace_Ftil_w} of the Supplementary Material \cite{sm}. Applying this result effectively reduces the physical problem of bounding the correlation function to the combinatorial task of estimating a weighted sum over clusters.

After careful analysis, we reorganize the summations over all clusters connecting $X$ and $Y$ into two distinct components. The first component yields the factor $(1+d_{X,Y})^{-\alpha}$, under the help of Lemma \ref{lemma_repro}. The second, more challenging component involves summations over clusters that constitute connected components with a fixed region. By ``connected components", we mean collections of regions where any pair of regions is connected by the collection itself. The weight for each cluster involves the quantity $\prod_{l=1}^{|w|}N(w_{l}|w)$ with the self-intersection index $N(w_{l}|w)=|\{Z\in w: V_Z\cap V_{w_{l}}\neq \emptyset\}|$ counting the number of elements in $w$ that geometrically overlap with $w_l$. 
Estimating sums that contain such combinatorial weights is a recurring challenge in the theory of cluster expansions \cite{kuwahara2020gaussian}. In Proposition \ref{pro_cluster_sum}, we establish a bound for these sums, providing a key technical tool for our main results. 

By combining both parts of the summations, we finally arrive at the following estimation
\begin{equation}\label{}
\begin{aligned}[b]
&|C_{\beta}(O_{X},O_{Y})|
\leq C_{4}^{|X|+|Y|}\Phi(\beta)|X||Y|\frac{e^{c_{1}(2t_{1}+|X|+|Y|)/2}}{(1+d_{X,Y})^{\as}}
\\&\quad \times  \sum_{t_{1}=0}^{\infty}\sum_{t_{2}=0}^{\infty}(4C_{3}C_{5}^{2}c_{2}g\nu\sqrt{\beta})^{t_{1}+t_{2}}(2^{\as}\nu)^{t_{1}-1}
\end{aligned}
\end{equation}
with $C_{3},C_{4},C_{5},\nu,c_{1}$ and $c_{2}$ being $\mathcal{O}(1)$ constants depending on $\alpha, D, g,\mu,U_{\min}$ and $U_{\max}$. Then by choosing a sufficiently small $\beta$, the summation above clearly converges, and the proof for the clustering theorem is completed.

\clearpage

\clearpage

\renewcommand\thefootnote{*\arabic{footnote}}

\addtocounter{section}{0}

\setcounter{equation}{0}
\setcounter{theorem}{0}
\setcounter{definition}{0}
\setcounter{lemma}{0}
\setcounter{proposition}{0}
\setcounter{assumption}{0}
\setcounter{corollary}{0}

\renewcommand{\theequation}{S.\arabic{equation}}

\renewcommand{\thesection}{S.\Roman{section}}

\renewcommand{\theHequation}{S.\arabic{equation}}

\renewcommand{\theHtheorem}{S.\arabic{theorem}}

\renewcommand{\theHdefinition}{S.\arabic{definition}}

\renewcommand{\theHlemma}{S.\arabic{lemma}}

\renewcommand{\theHproposition}{S.\arabic{proposition}}

\renewcommand{\theHassumption}{S.\arabic{assumption}}

\renewcommand{\theHcorollary}{S.\arabic{corollary}}
\begin{widetext}

\begin{center}
{\large \bf Supplemental Material for  ``Long-Range Bosonic Systems at Thermal Equilibrium: Computational Complexity and Clustering of Correlations''}\\
\vspace*{0.3cm}
Xin-Hai Tong$^{1}$ and Tomotaka Kuwahara$^{2,3,4}$ \\
\vspace*{0.1cm}
$^{1}${\small \it Department of Physics, The University of Tokyo, 5-1-5 Kashiwanoha, Kashiwa-shi, Chiba 277-8574, Japan} \\
$^{2}${\small \it Analytical Quantum Complexity RIKEN Hakubi Research Team, RIKEN Center for Quantum Computing (RQC), Wako, Saitama 351-0198, Japan} \\
$^{2}${\small \it RIKEN Cluster for Pioneering Research (CPR), Wako, Saitama 351-0198, Japan}\\
$^{3}${\small \it PRESTO, Japan Science and Technology (JST), Kawaguchi, Saitama 332-0012, Japan}\\
\end{center}

\tableofcontents

\section{Notations and Preliminaries}
Though we primarily focus on the long-range Bose-Hubbard model, we first introduce the general notation and then subsequently specialize to the primary interest.

Let $V$ be a finite set of sites embedding in the $D$-dimensional (not necessarily Euclidean) space, where each site $i\in V$ is equipped with an infinite-dimensional Hilbert space $\mathcal{H}_{i}$. Usually, we use the letters $i,j$ (and $x,y$ when necessary) to denote the sites. We refer to $V$ as the vertex set, and define the set $E \coloneqq \{Z \subset V \colon |Z| \leq k\}$ as the collection of all subsets of $V$ with cardinality at most $k\in \mathbb{N}_{\geq 1}$. For each $Z\in E$, we assign a not necessarily bounded Hermitian operator $h_{Z}$ acting on the local Hilbert space $\bigotimes_{i\in Z}\mathcal{H}_{i}$, and define the total $k$-local Hamiltonian as $H=\sum_{Z\in E}h_{Z}$. The total Hamiltonian acts on the tensor product space $\mathcal{H}\coloneq \bigotimes_{i\in V}\mathcal{H}_{i}$. 
Each term $h_Z$ in the sum is understood as an operator on the full Hilbert space $\mathcal{H}$ through the identification $h_Z \equiv h_Z \otimes \id_{V \setminus Z}$, where $\id_{V \setminus Z}=\bigotimes_{i\in V\setminus Z}\id_{i}$ denotes the identity on the subspace $\mathcal{H}_{V \setminus Z} \coloneqq \bigotimes_{i \in V \setminus Z} \mathcal{H}_i$ with $\id_{i}$ being the identity (not normalized) on $\mathcal{H}_{i}$.

To formulate the long-range condition, we first assume that each local term $h_Z$ admits a decomposition of the form:
\begin{equation}\label{sh_Z}
\begin{aligned}[b]
h_{Z}=\sum_{s}J_{Z}^{(s)}h_{Z}^{(s)}.
\end{aligned}
\end{equation}
Here, the coefficients $J_{Z}^{(s)} \in \mathbb{C}$ are coupling constants that may depend on system parameters. The operators $h_{Z}^{(s)}$, by contrast, are simply the product of the creation and annihilation operators $\{a_i^\dagger, a_i\}_{i \in Z}$ and are independent of such parameters.
The range index $s$ depends on a given region $Z$, and in principle, the summation $\sum_{s}$ should be explicitly presented as $\sum_{s\in \mathbb{N}_{Z}}$ with $\mathbb{N}_{Z}$ being some subset of $\mathbb{N}$. However, for simplification, we just omit this notation whenever there is no potential ambiguity. 

Next, to introduce a notion of distance, we equip the vertex set $V$ with an edge set $\mathcal{E}$, thereby defining a simple, connected, and undirected graph $G \coloneqq (V, \mathcal{E})$. The distance between two sites $i,j \in V$ is then defined as their graph distance, i.e., the length of the shortest path connecting them in the graph.
This graph-theoretic distance should be distinguished from any physical distance arising from a specific geometric embedding of the lattice. Note that $d_{i,j}=0$ if and only if $i=j$. With this distance metric, we now define the long-range nature of the couplings by imposing the following condition  for all pairs of sites $i\neq j$,
\begin{equation}\label{slr_condition}
\begin{aligned}[b]
\sum_{Z\in E: Z\ni \{i,j\}}\sum_{s}|J_{Z}^{(s)}| \eqcolon \sum_{Z\in E: Z\ni \{i,j\}}J_{Z} \leq \frac{g}{(1+d_{i,j})^{\as}}, \quad g=\mathcal{O}(1).
\end{aligned}
\end{equation}
Note that the condition \eqref{slr_condition} only restricts the off-site terms (with $i\neq j$). For on-site terms
\begin{equation}\label{}
\begin{aligned}[b]
\sum_{i}h_{i}=\sum_{i}\sum_{s}J_{i}^{(s)}h_{i}^{(s)}\eqcolon H_{0},
\end{aligned}
\end{equation}
we require that its coefficients are always bounded
\begin{equation}\label{on_site_unbound_unifo}
\begin{aligned}[b]
\sum_{s}|J_{i}^{(s)}|\leq g_{0}, \quad g_{0}=\mathcal{O}(1).
\end{aligned}
\end{equation}
Furthermore, we assume that for each site $i \in V$, the local term $h_i$ is non-trivial and contains a component $f(n_i)$, where $f$ is a real-valued function of the local number operator $n_i \coloneqq a_i^\dagger a_i$. This component must be bounded below by a polynomial in $n_i$ with a positive leading coefficient.

Note that one of the primary focuses in this work is the correlation function defined on the Gibbs state $\rho_{\beta}\coloneq e^{-\beta H}/\operatorname{Tr}e^{-\beta H}$: 
\begin{equation}\label{scor}
\begin{aligned}[b]
C_{\beta}(O_{X},O_{Y})\coloneq \operatorname{Tr}(\rho_{\beta}O_{X}O_{Y})-\operatorname{Tr}(\rho_{\beta}O_{X})\cdot \operatorname{Tr}(\rho_{\beta}O_{Y})
\end{aligned}
\end{equation}
where $O_{X}$ and $O_{Y}$ are operators acting on $\bigotimes_{i\in X}\mathcal{H}_{i}$ and $\bigotimes_{i\in Y}\mathcal{H}_{i}$ respectively, with $X,Y\subset V$ being vertex subsets (or referred as regions). We will use $X^{\cc}\coloneq V\backslash X$ to denote the complement of $X$.
For convenience, we first express the correlation function in a more compact form by introducing the following notation:
\begin{equation}\label{splus_0_1}
\begin{aligned}[b]
O^{(+)}\coloneq O\otimes \id+\id \otimes O,\quad O^{(0)}\coloneq O\otimes \id ,\quad O^{(1)}\coloneq\id \otimes O-O\otimes \id
\end{aligned}
\end{equation}
where $\id=\bigotimes_{i\in V}\id_{i}$ denotes the identity operator on $\mathcal{H}$. This allows us to reformulate Eq.\,\eqref{scor} in the doubled total Hilbert space $\mathcal{H}\otimes \mathcal{H}$ as:
\begin{equation}\label{cor_doubled_Hilbert}
\begin{aligned}[b]
C_{\beta}(O_{X},O_{Y})=\frac{1}{\mathcal{Z}^{2}}\operatorname{Tr}\qty[e^{-\beta H^{(+)}}O_{X}^{(0)}O_{Y}^{(1)}],
\end{aligned}
\end{equation}
where $\mathcal{Z}\coloneq \operatorname{Tr}e^{-\beta H}$ is the partition function.

In this supplementary material, we use $C, C_1, C_{2}, c, c_1, c_{2}, ...=\mathcal{O}(1)$, to denote generic constants that are independent of the system size $|V|$. Unless specified otherwise, their values depend only on the model parameters ($\alpha, D, g,\mu,U_{\min}$ and $U_{\max}$). The value of such constants may change from one statement to another. The threshold temperature $\beta_c=\mathcal{O}(1)$ is a constant that is independent of the system size $|V|$. In the event that various results presented herein require distinct threshold values, say $\{\beta_{c,i}\}$, we can, without loss of generality, select a single common threshold $\beta_c := \min_i\beta_{c,i}$ that remains valid for all statements simultaneously. The partition function $\operatorname{Tr}e^{-\beta H}$ is denoted by $\mathcal{Z}_{\beta}$. We may suppress the subscript $\beta$ when its dependence on the inverse temperature $\beta$ is clear from the context or not immediately relevant. See Fig.\,\ref{mindmap_sm} for the illustration of the proof structure of this work.

\begin{figure*}[h]
\centering
\includegraphics[width=\textwidth]{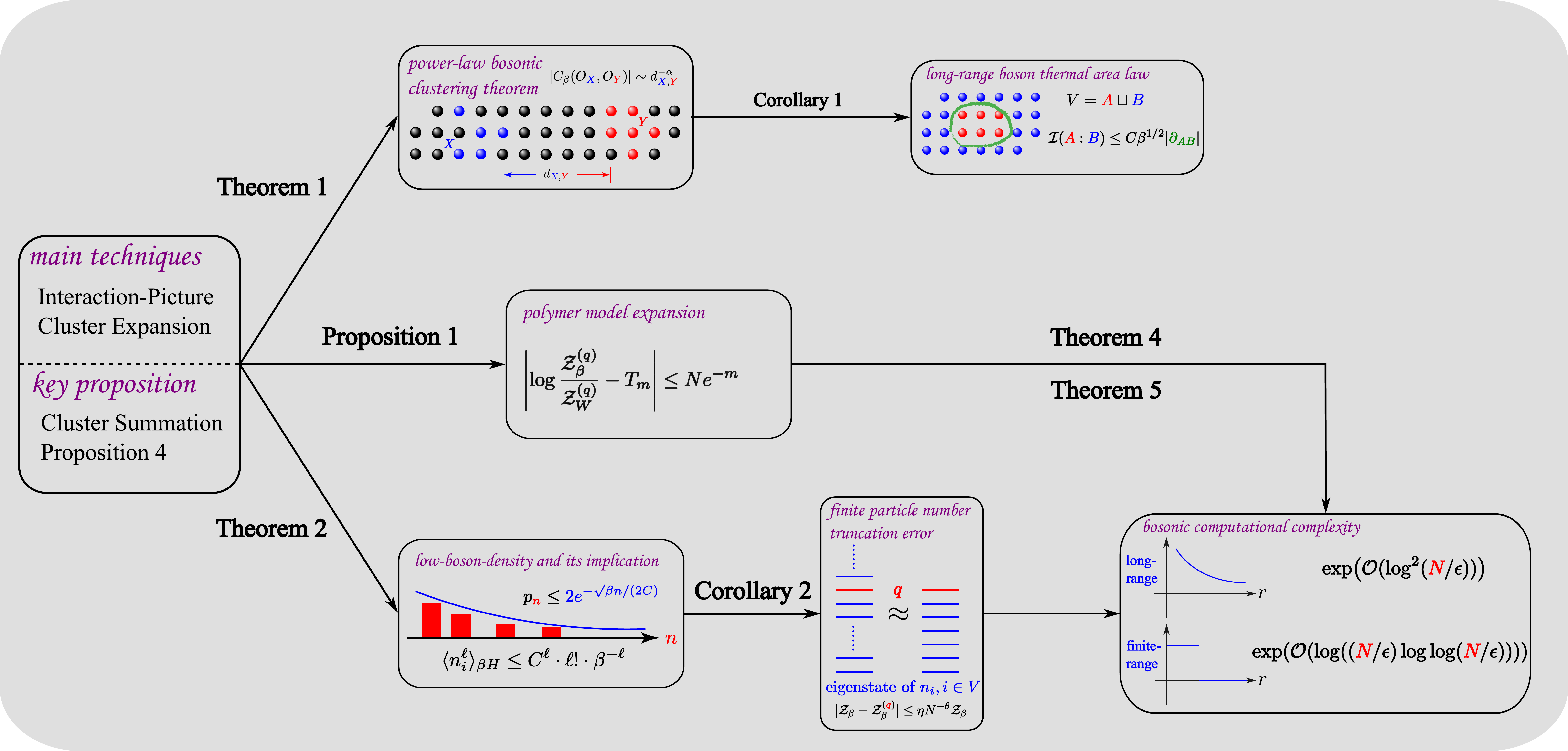}
\captionsetup{justification=raggedright,singlelinecheck=false}
\caption{This figure provides a schematic overview of the main results, methodological framework, and structure of this supplementary material, which concerns the study of high-temperature thermal states in long-range bosonic systems. Our analysis approaches the problem from the perspectives of both complexity theory and statistical mechanics. The system size here is denoted by $N \coloneq |V|$.}
\label{mindmap_sm}
\end{figure*}

\section{Interaction-picture Cluster Expansion Formalism}
One of the main techniques to prove the long-range clustering theorem, as well as the efficient algorithm, is the interaction-picture cluster expansion. This formalism, initiated in Ref.\,\cite{tong2024boson}, is particularly designed for lattice bosons and has promising applications. This section is devoted to presenting the general formulation of such a novel kind of cluster expansion, for later convenience. The following sections will focus on the application aspect of the results in this section.

To highlight the motivation, we start from the correlation function. Before that, we denote $W$ as the summation of all the 
on-site interaction terms, and the rest of the terms are expressed by $W-H\eqcolon I$. Here, the original on-site interaction term are set to be zero in $I=\sum_{Z\in E}h_{Z}=\sum_{Z\in E}\sum_{s}J_{Z}^{(s)}h_{Z}^{(s)}$ and we absorb the minus sign into the parameter $J_{Z}^{(s)}$.
Then the correlation function is given by [cf.\,Eq.\,(\ref{cor_doubled_Hilbert})]
\begin{equation}\label{cor_Stilde}
\begin{aligned}[b]
C_{\beta}(O_X,O_Y)=\frac{1}{\mathcal{Z}^{2}}\operatorname{Tr}\qty[e^{-\beta W^{(+)}}\widetilde{S}(\beta)O_X^{(0)}O_Y^{(1)}]
\end{aligned}
\end{equation}
with the Dyson series defined over the doubled Hilbert space given by
\begin{equation}\label{Stilde}
\begin{aligned}[b]
\widetilde{S}(\beta)=\sum_{m=0}^{\infty}\sum_{Z_{1},Z_{2},...,Z_{m}\in E}\ii{0}{\beta}{\tau_{1}}\ii{0}{\tau_{1}}{\tau_{2}}...\ii{0}{\tau_{m-1}}{\tau_{m}} h_{Z_{1}}(\tau_{1})^{(+)}h_{Z_{2}}(\tau_{2})^{(+)}...h_{Z_{m}}(\tau_{m})^{(+)}, \quad \bullet(\tau)\coloneq e^{\tau W}\bullet e^{-\tau W}
\end{aligned}
\end{equation}

For a more compact and elegant presentation of the proof and related discussions, we introduce several definitions.
Let $G$ denote a multiset with elements being vertex subsets, represented as $G=\{G_{1},G_{2},...,G_{|G|}\}$ where $G_{i}\subseteq V$ for all $i\in [|G|]$ (For a positive integer $n$, we let $[n]$ denote the set $\{1,2,…,n\}$) represents the elements in $G$. A multiset is an unordered collection that may contain repeated elements. It should be noted that we use $\cup$ only for the standard union between sets, not between multisets. For multisets $G$ and $G'$, we define their union (or unordered concatenation) as the multiset $G\oplus  G'\coloneq\{G_{1},G_{2},...,G_{|G|},G'_{1},G'_{2},...,G'_{|G'|}\}$. Obviously, we still have the associative property $G\oplus G' \oplus G''\coloneq (G\oplus G') \oplus G''= G\oplus (G' \oplus G'')$. We denote the set of all multisets whose elements are taken from $E_{0}\subseteq E$ as $\mathcal{M}(E_{0})$.
In our analysis, we denote the set of all ordered sequences generated by a fixed multiset $G$ as $\mathcal{S}(G)$, defined by:
\begin{equation}
\begin{aligned}[b]
\mathcal{S}(G)\coloneq \{G_{P(1)}G_{P(2)}...G_{P(|G|)} \colon  P \in \mathcal{P}(G)\}.
\end{aligned}
\end{equation} 
Here, $\mathcal{P}(G)$ represents the set of all multiset permutations over $G$.
For each element $Z \in G$, let $\mu(Z)$ denote its multiplicity in $G$. If we define $\mathcal{D}(G)$ as the set obtained from $G$ by keeping exactly one copy of each distinct element, then the number of multiset permutations of $G$ is given by:
\begin{equation}\label{multiset_permutation}
|\mathcal{P}(G)| = \frac{|G|!}{\prod_{Z \in \mathcal{D}(G)} \mu(Z)!} = |\mathcal{S}(G)|.
\end{equation}
For example, given $G=\{Z_1,Z_1,Z_2\}$, we have $\mathcal{S}(G)=\{Z_1 Z_1 Z_2,Z_1 Z_2 Z_1,Z_2 Z_1 Z_1\}$, $\mathcal{D}(G)=\{Z_{1},Z_{2}\}$ and $|\mathcal{P}(G)|=3!/(2!1!)=3$. For the empty multiset $G=\emptyset$, we simply have $\mathcal{S}(\emptyset)=\emptyset$, and we also use $\emptyset$ to denote an empty sequence when no ambiguity arises. Then Eq.\,\eqref{Stilde} can be recast as
\begin{equation}\label{sexp_cluster}
\begin{aligned}[b]
\widetilde{S}(\beta) = \sum_{G\in \mathcal{M}(E)}\sum_{w\in \mathcal{S}(G)}\widetilde{F}(w),
\end{aligned}
\end{equation}
where we have defined the function $F$ that maps an ordered sequence of elements in $E$ to an intergal over a product of operators:
\begin{equation}\label{sF_function}
\begin{aligned}[b]
\widetilde{F}\colon Z_{1}Z_{2}\cdots Z_{m} \mapsto \ii{0}{\beta}{\tau_{1}}\ii{0}{\tau_{1}}{\tau_{2}}...\ii{0}{\tau_{|m|-1}}{\tau_{|m|}}\, h_{Z_{1}}(\tau_{1})^{(+)}h_{Z_{2}}(\tau_{2})^{(+)}...h_{Z_{m}}(\tau_{m})^{(+)}
\end{aligned}
\end{equation}
with $m\in \mathbb{N} $ and we define $F(\emptyset)=\id$ for an empty sequence. Note that the absolute convergence can be established in the long-range Bose-Hubbard model as detailed in Proposition \ref{pro_abs_boson}.
 
According to the following lemma, not every term on the RHS (right-hand side) of Eq.\,\eqref{sexp_cluster} contributes nonzero values to the correlation function in Eq.\,\eqref{cor_Stilde}. To address this lemma, we first introduce the concept of connectivity, 
\begin{definition}[Connect/Connected Multiset]
We say two vertex subsets $X,Y\subset L$ are connected by a multiset $G$, if there exists a series of elements $Z_{1}, Z_{2}, \ldots, Z_{n} \in G/w$ ($n\leq |G|$) such that $Z_{1} \cap X \neq \emptyset$, $Z_{n} \cap Y \neq \emptyset$, and $Z_{j} \cap Z_{j+1} \neq \emptyset$ for all $j \in \{1, 2, \ldots, n-1\}$.
We define a multiset to be connected if any pair of its elements is connected by itself. These definitions apply analogously to sequences.
\end{definition}
It should be emphasized that this definition of connectivity is on the level of multisets rather than geometry. We note that $(V_{G},\mathcal{E}_{G})$ is not necessarily a connected subgraph of $(V,\mathcal{E})$, even if $G$ is a connected multiset. Here, we have denoted the corresponding vertex and edge subset as $V_{G} \coloneqq \{i \in V \colon \exists Z \in G \text{ such that } Z \ni i\}$ and $\mathcal{E}_{G} \coloneqq \{e \in \mathcal{E} \colon e = (i,j) \text{ for some } i,j \in V_G\}$, respectively. Then we have 
\begin{lemma}\label{lemma_disconnect_cluster}
Let $G$ be a multiset that does not connect $X$ and $Y$. Define the function $\widetilde{F}$ as in Eq.\,\eqref{sF_function}. Then for any sequence $w\in \mathcal{S}(G)$, we have 
\begin{equation}
\begin{aligned}[b]
\operatorname{Tr}\qty[e^{-\beta W^{(+)}}\widetilde{F}(w)O_{X}^{(0)}O^{(1)}_{Y}]=0.
\end{aligned}
\end{equation}
Here, $O_{X}$ and $O_{Y}$ are operators supported on $X$ and $Y$, respectively.
\end{lemma}
\begin{proof}
We denote $|w|=m$ as the length of the sequence. By assumption, the sequence $w=Z_{1}Z_{2}...Z_{m}$ cannot connect the regions $X$ and $Y$ i.e.,
\begin{equation*}      
(X \cup Z_{i_1} \cup Z_{i_2} \cup \dots \cup Z_{i_s}) \cap (Y \cup Z_{i_{s+1}} \cup Z_{i_{s+2}} \cup \dots \cup Z_{i_m}) = \emptyset
\end{equation*}
for some $\{i_1, i_2, \dots, i_m\} = \{1, 2, \dots, m\}$. Without loss of generality, let $\{i_1, i_2, \dots, i_s\} = \{1, 2, \dots, s\}$ and $\{1, 2, \dots, s\} \setminus \{i_1, i_2, \dots, i_m\} = \{s+1, s+2, \dots, m\}$. Then we denote
\begin{equation}\label{}
\begin{aligned}[b]
h_{Z_{1}}(\tau_{1})^{(+)}h_{Z_{2}}(\tau_{2})^{(+)}...h_{Z_{s}}(\tau_{s})^{(+)}=\mathcal{W}_{1}, \quad h_{Z_{s+1}}(\tau_{s+1})^{(+)}h_{Z_{s+2}}(\tau_{s+2})^{(+)}...h_{Z_{m}}(\tau_{m})^{(+)}=\mathcal{W}_{2}
\end{aligned}
\end{equation}
and recognize
\begin{equation}
\text{Supp}(X \cup \mathcal{W}_1) \cap \text{Supp}(Y \cup \mathcal{W}_2) = \emptyset. 
\end{equation}
Obviously, we have
\allowdisplaybreaks[4]
\begin{align*}\label{}
&\operatorname{Tr}\qty[e^{-\beta W^{(+)}}\widetilde{F}(w)O_{X}^{(0)}O^{(1)}_{Y}]
\\=&\operatorname{Tr}\qty[e^{-\beta W^{(+)}}\ii{0}{\beta}{\tau_{1}}\ii{0}{\tau_{1}}{\tau_{2}}...\ii{0}{\tau_{|m|-1}}{\tau_{|m|}}\,h_{Z_{1}}(\tau_{1})^{(+)}h_{Z_{2}}(\tau_{2})^{(+)}...h_{Z_{m}}(\tau_{m})^{(+)}O_{X}^{(0)}O^{(1)}_{Y}]
\\=&\ii{0}{\beta}{\tau_{1}}\ii{0}{\tau_{1}}{\tau_{2}}...\ii{0}{\tau_{|m|-1}}{\tau_{|m|}}\,\operatorname{Tr}\qty[ e^{-\beta W^{(+)}} h_{Z_{1}}(\tau_{1})^{(+)}h_{Z_{2}}(\tau_{2})^{(+)}...h_{Z_{m}}(\tau_{m})^{(+)}O_{X}^{(0)}O^{(1)}_{Y}]
\\=&\ii{0}{\beta}{\tau_{1}}\ii{0}{\tau_{1}}{\tau_{2}}...\ii{0}{\tau_{|m|-1}}{\tau_{|m|}}\,\operatorname{Tr}_{X}\qty[e^{-\beta W_{X}^{(+)}}\mathcal{W}_{1} O_{X}^{(0)}] \cdot \operatorname{Tr}_{X^{\cc}}\qty[e^{-\beta W_{X^{\cc}}^{(+)}}\mathcal{W}_{2}O_{Y}^{(1)}].
\refstepcounter{equation}\tag{\theequation}
\end{align*}
Here, the quantity $W_{X}\coloneq \sum_{i\in X}W_{i}$ collects the on-site potentials in the region $X$ and $\operatorname{Tr}_{X}$ denotes the partial trace over $\bigotimes_{i\in X}\mathcal{H}_{i}$ (similar notation has clear meaning for $X^{\cc}$). 

Then we observe that, the operators $O_{Z_{i_k}}^{(+)} = O_{Z_{i_k}} \otimes \id + \id \otimes O_{Z_{i_k}}$ with $k\in [m]$ and $e^{-\beta W_{X^{\cc}}^{(+)}}=e^{-\beta W_{X^{\cc}}}\otimes  e^{-\beta W_{X^{\cc}}}  $ are symmetric for the exchange of the two Hilbert spaces while the operator $O_Y^{(1)} = O_Y \otimes \id - \id \otimes O_Y$ is antisymmetric, which yields $\operatorname{Tr}_{X^{\cc}}\left[ e^{-\beta W_{X^{\cc}}^{(+)}} \mathcal{W}_2 O_Y^{(1)} \right] = 0$. We thus finish the proof of this lemma. 

\end{proof}

{~}

\hrulefill{\bf [ End of Proof of Lemma~\ref{lemma_disconnect_cluster}]}

{~}

From the lemma above, only those multisets that connect $X$ and $Y$ contribute to the correlation function in Eq.\,\eqref{cor_Stilde}, therefore we write
\begin{equation}\label{cor_rho_cl}
\begin{aligned}[b]
C_{\beta}(X,Y)=\operatorname{Tr}\qty[e^{-\beta W^{(+)}}\widetilde{\rho}_{\cl}O_X^{(0)}O_Y^{(1)}]
\end{aligned}
\end{equation}
with
\begin{equation}\label{srho_cl}
\begin{aligned}[b]
\widetilde{\rho}_{\cl}=\frac{1}{\mathcal{Z}^{2}}\sum_{G\in \mathcal{G}(E)}\sum_{w\in \mathcal{S}(G)}\widetilde{F}(w).
\end{aligned}
\end{equation}
Here, we have already defined
\begin{equation}\label{}
\begin{aligned}[b]
\mathcal{G}(E)\coloneq \{G\in \mathcal{M}(E) \colon X \text{ and } Y \text{ are connected by } G\}
\end{aligned}
\end{equation}
to collect all the desirable multisets. 

To expand Eq.\,\eqref{srho_cl} and convert it into a more convenient form, we need the following observation. Even if a multiset $G \in \mathcal{G}(E)$ connects $X$ and $Y$, their union $G\oplus X\oplus Y$ is not necessarily a connected multiset, as illustrated in Fig.~\ref{w_cl_c_example}.
\begin{figure}[h]
\centering
\includegraphics[width=0.6\linewidth]{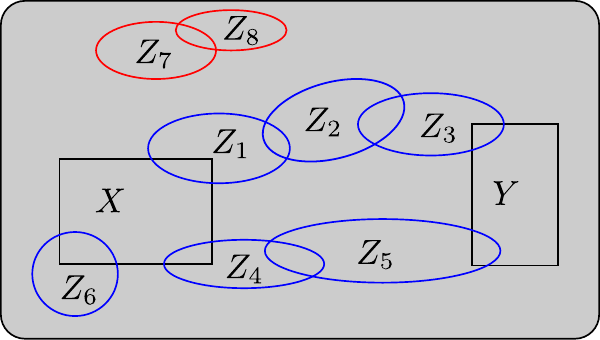}
\caption{An example where the multiset $G = \{Z_1, Z_2, \ldots, Z_8\}$ connects $X$ and $Y$, but $G\oplus X\oplus Y$ is not connected.}
\label{w_cl_c_example}
\end{figure}
This observation motivates us to further refine the enumeration of all desirable multisets in Eq.~\eqref{srho_cl}. 
If $G$ connects $X$ and $Y$, then there always exists a subset $G_0 \subseteq G$ such that $G_0\oplus X\oplus Y$ is connected. This allows us to decompose any $G \in \mathcal{G}(E)$ as follows:
\begin{equation}\label{sdecompos}
G = w_{\text{cl}}(G)\oplus w_{\text{cl}}^{\text{c}}(G)
\end{equation}
where $w_{\text{cl}}^{\text{c}}(G) \coloneqq G \setminus w_{\text{cl}}(G)$, and $w_{\text{cl}}(G)$ is the largest subset (with respect to cardinality) of $G$ such that $\{w_{\text{cl}}(G), X, Y\}$ forms a connected multiset:
\begin{equation}\label{w_cl}
w_{\text{cl}}(G) \coloneqq \arg\max_{G_0 \subseteq G : G_0\oplus X\oplus Y \in \mathcal{C}} |G_0|
\end{equation}
Here, $\mathcal{C}$ denotes the set of all connected multisets whose elements are the vertex subsets. For clarity, we define $\mathcal{C}(E) \coloneqq \{G \in \mathcal{M}(E) \colon G \text{ is connected}\}$ as the set of all connected multisets whose elements are drawn from $E$. For any multiset $G' \notin \mathcal{G}(E)$, we adopt the convention that $w_{\text{cl}}(G') = \emptyset$, ensuring that decomposition \eqref{sdecompos} remains well-defined for all multisets in $\mathcal{C}$.
The decomposition \eqref{sdecompos} provides an alternative approach to evaluating the sum over $\mathcal{G}(E)$ in Eq.~\eqref{srho_cl}. Specifically, we can first identify a multiset $G_1$ such that $G_1\oplus X\oplus Y \in \mathcal{C}$, which can be considered as $w_{\text{cl}}(G)$ for some $G \in \mathcal{M}(E)$. Then, we find another multiset $G_2$ such that $G = G_{1}\oplus G_2$ and $G_{1}=w_{\cl}(G)$. This reverse process allows us to systematically enumerate the elements of $\mathcal{G}(E)$.

To address this process in a formal way, we first denote 
\begin{equation}\label{def_g_cl_E}
\mathcal{G}_{\mathrm{cl}}(E) \coloneqq \{G \in \mathcal{M}(E) \colon G\oplus X\oplus Y \in \mathcal{C}\}
\end{equation}
as the set of all possible elements that can be considered as the image of the map $w_{\mathrm{cl}}(\bullet)$.

Note that for a given $G_{\mathrm{cl}} \in \mathcal{G}_{\mathrm{cl}}(E)$, the choices of $G_{\mathrm{cl}}^{\mathrm{c}}$ such that $G_{\mathrm{cl}}\oplus G_{\mathrm{cl}}^{\mathrm{c}}$ correspond to summands in Eq.~\eqref{srho_cl} are highly constrained by this fixed $G_{\mathrm{cl}}$. We refer the reader to Fig.\,\ref{w_cl_c_example} for examples. If we take $G_{\mathrm{cl}} = \{Z_1, Z_2, Z_3\}$, then $G_0 = \{Z_4, Z_5\}$ is not an acceptable $G_{\mathrm{cl}}^{\mathrm{c}}$ since, by definition, $G_{\mathrm{cl}}\oplus G_0 = \{Z_1, \ldots, Z_5\}$ and $w_{\mathrm{cl}}(G_{\mathrm{cl}}^{(0)}\oplus G_0) = \{Z_1, \ldots, Z_5\} \neq G_{\mathrm{cl}}$. By the same reasoning, we know that $\{Z_6\}$ and $\{Z_4, Z_5, Z_6\}$ are also not eligible, while $\{Z_7, Z_8\}$ is a candidate for $G_{\mathrm{cl}}^{\mathrm{c}}$.
Therefore, the set collecting all possible candidates for $w_{\mathrm{cl}}^{\mathrm{c}}$ should depend on the previously fixed $w_{\mathrm{cl}}^{(0)} \in \mathcal{G}_{\mathrm{cl}}(E)$, leading us to introduce the following definition:
\begin{equation}\label{def_g_cl_c_G}
\mathcal{G}_{\mathrm{cl}}^{\mathrm{c}}(G) \coloneqq \{G_0 \in \mathcal{M}(E) \colon w_{\mathrm{cl}}(G\oplus G_0) = G\}.
\end{equation}

Then the main goal of this section is to justify the following identity:
\begin{equation}\label{cluster_decom}
\begin{aligned}[b]
\widetilde{\rho}_{\cl}=\frac{1}{\mathcal{Z}^{2}}\sum_{G^{\ri}\in \mathcal{G}_{\cl}(E)}\sum_{w^{\ri}\in \mathcal{S}(G^{\ri})}\widetilde{F}(w^{\ri})\cdot \sum_{G^{\rii}\in \mathcal{G}^{\cc}_{\cl}(G^{\ri})}\sum_{w^{\rii}\in \mathcal{S}(G^{\rii})}\widetilde{F}(w^{\rii}).
\end{aligned}
\end{equation}
This will be achieved by the following lemma

\begin{lemma}\label{lemma_shuffle_F}
For a multiset  $G\in \mathcal{M}(E)$, let $G=G^{\ri}\oplus G^{\rii}$
be a decomposition such that $\mathcal{D}(G_{1})\cap \mathcal{D}(G_{2})=\emptyset$ and $f_{Z_{1}}(x_{1})f_{Z_{2}}(x_{2})=f_{Z_{2}}(x_{2})f_{Z_{1}}(x_{1})$ for $Z_{1}\in G^{\ri}, Z_{2}\in G^{\rii}$ and $x_{1},x_{2}\in \mathbb{R}_{\geq 0}$. Here $f_{Z}(x)$ denotes an operator parametrized by $Z\in E$ and $x\in \mathbb{R}_{\geq 0}$. For a sequence $w$, define  
\begin{equation}\label{def_F_w_abstract}
\begin{aligned}[b]
F(w)\coloneq\ii{0}{\beta}{\tau_{1}}\ii{0}{\tau_{1}}{\tau_{2}}...\ii{0}{\tau_{|w|-1}}{\tau_{|w|}} f_{w_{1}}(\tau_{1})f_{w_{2}}(\tau_{2})...f_{w_{|w|}}(\tau_{|w|}),
\end{aligned}
\end{equation}
then we have
\begin{equation}\label{}
\begin{aligned}[b]
\sum_{w\in \mathcal{S}(G)}F(w)=\sum_{w^{\ri}\in \mathcal{S}(G^{\ri})}F(w^{\ri})\cdot \sum_{w^{\rii}\in \mathcal{S}(G^{\rii})}F(w^{\rii}).
\end{aligned}
\end{equation}
\end{lemma}
\begin{proof}
We first note the following simple fact:
\begin{equation}\label{}
\begin{aligned}[b]
\sum_{w\in \mathcal{S}(G)}w=\sum_{w^{\ri}\in \mathcal{S}(G^{\ri})}\sum_{w^{\rii}\in \mathcal{S}(G^{\rii})}w^{\ri}\shuffle w^{\rii}.
\end{aligned}
\end{equation}
Here, the symbol $\shuffle$ denotes the shuffle product for two sequences. For convenience we define $F(a+b)\coloneq F(a)+F(b)$ for two sequences $a$ and $b$. If we can prove 
\begin{equation}\label{shuffle_F}
\begin{aligned}[b]
F(w^{\ri}\shuffle w^{\rii})=F(w^{\ri})F(w^{\rii}),
\end{aligned}
\end{equation}
with $w^{\ri/\rii}\in \mathcal{S}(G^{\ri/\rii})$, then we immediately arrive at
\begin{equation}\label{}
\begin{aligned}[b]
\sum_{w\in \mathcal{S}(G)}F(w)&=F(\sum_{w\in \mathcal{S}(G)} w)=F(\sum_{w^{\ri}\in \mathcal{S}(G^{\ri})}\sum_{w^{\rii}\in \mathcal{S}(G^{\rii})}w^{\ri}\shuffle w^{\rii})=\sum_{w^{\ri}\in \mathcal{S}(G^{\ri})}\sum_{w^{\rii}\in \mathcal{S}(G^{\rii})}F(w^{\ri}\shuffle w^{\rii})
\\&=\sum_{w^{\ri}\in \mathcal{S}(G^{\ri})}\sum_{w^{\rii}\in \mathcal{S}(G^{\rii})}F(w^{\ri})F(w^{\rii}),
\end{aligned}
\end{equation}
which finishes the proof. 

Next, we present the proof for Eq.\,\eqref{shuffle_F}. For notational convenience, we replace $w^{\ri},w^{\rii}$ by $a,b$ and simplify $f_{Z}(\tau)$ to $Z(\tau)$. Note that $a$ and $b$ do not share any common elements, and we have
\begin{equation}\label{commu_relation}
\begin{aligned}[b]
Z(\tau)Z'(\tau')=Z'(\tau')Z(\tau), \quad \forall Z\in a, Z'\in b, \tau,\tau'\in \mathbb{R}^{+}.
\end{aligned}
\end{equation}
We will also replace $\ii{0}{\beta}{\tau_{1}}\ii{0}{\tau_{1}}{\tau_{2}}...\ii{0}{\tau_{m-1}}{\tau_{m}}$ by $\int [\tau_{1}\tau_{2}...\tau_{m}]$ for shorter notation.
Then we have (denoting $|a|=m, |b|=n$ and we use $w=a\oplus b \coloneq a_{1}a_{2}...a_{m}b_{1}b_{2}...b_{n}$ to represent their ordered concatenation)
\begin{equation}\label{}
\begin{aligned}[b]
F(a\shuffle b)=\int [\tau_{1}\tau_{2}...\tau_{m+n}]\sum_{\sigma\in \operatorname{OP}(a,b)}w_{\sigma(1)}(\tau_{1})w_{\sigma(2)}(\tau_{2})...w_{\sigma(m+n)}(\tau_{m+n}),
\end{aligned}
\end{equation} 
where $\operatorname{OP}(a,b)$ is the set of all permutations over the elements in $a\oplus b$ with preserving the orders in $a$ and $b$. By switching the argument $\tau_{i}\rightarrow \tau_{\sigma(i)}$ we obtain
\begin{equation}\label{F_a_shu_b}
\begin{aligned}[b]
F(a\shuffle b)=\sum_{\sigma\in \operatorname{OP}(a,b)}\int [\tau_{\sigma(1)}\tau_{\sigma(2)}...\tau_{\sigma(m+n)}]w_{\sigma(1)}(\tau_{\sigma(1)})w_{\sigma(2)}(\tau_{\sigma(2)})...w_{\sigma(m+n)}(\tau_{\sigma(m+n)}).
\end{aligned}
\end{equation}
We know from the property \eqref{commu_relation} that, for any $\sigma$ the integrand on RHS (right hand side) of Eq.\,\eqref{F_a_shu_b} equals $a_{1}(\tau_{1})a_{2}(\tau_{2})...a_{m}(\tau_{m})b_{1}(\tau_{m+1})b_{2}(\tau_{m+2})...b_{n}(\tau_{m+n})$, implying
\begin{equation}\label{F_a_shu_b_2}
\begin{aligned}[b]
F(a\shuffle b)&=\sum_{\sigma\in \operatorname{OP}(a,b)}\int [\tau_{\sigma(1)}\tau_{\sigma(2)}...\tau_{\sigma(m+n)}]a_{1}(\tau_{1})a_{2}(\tau_{2})...a_{m}(\tau_{m})b_{1}(\tau_{m+1})b_{2}(\tau_{m+2})...b_{n}(\tau_{m+n})
\\&=\int [\tau_{1}\tau_{2}...\tau_{m}] a_{1}(\tau_{1})a_{2}(\tau_{2})...a_{m}(\tau_{m})\cdot \int [\tau_{m+1}\tau_{m+2}...\tau_{m+n}]b_{1}(\tau_{m+1})b_{2}(\tau_{m+2})...b_{n}(\tau_{m+n})
\\&=F(a)F(b)
\end{aligned}
\end{equation}
to finish the proof. Here, to obtain the second line, we used the fact that
\begin{equation}\label{}
\begin{aligned}[b]
&\bigcup_{\sigma\in \operatorname{OP}(a,b)}\{\vec{\tau}\in \mathbb{R}^{m+n}| \beta \geq \tau_{\sigma(1)}\geq \tau_{\sigma(2)}\geq ...\geq \tau_{\sigma(m+n)}\geq 0\}
\\=& \{\vec{\tau}\in \mathbb{R}^{m+n}|\beta\geq \tau_{1}\geq \tau_{2}\geq ...\geq \tau_{m}\geq 0, \beta\geq \tau_{m+1}\geq \tau_{m+2}\geq ...\geq \tau_{m+n}\geq 0  \}
\\=& \{\vec{\tau}\in \mathbb{R}^{m+n}|\beta\geq \tau_{1}\geq \tau_{2}\geq ...\geq \tau_{m}\geq 0 \}\bigcup\{\vec{\tau}\in \mathbb{R}^{m+n}|\beta\geq \tau_{m+1}\geq \tau_{m+2}\geq ...\geq \tau_{m+n}\geq 0  \} 
\end{aligned}
\end{equation}
\end{proof}
{~}

\hrulefill{\bf [ End of Proof of Lemma~\ref{lemma_shuffle_F}]}

{~}

By taking $f_{Z}(\tau)=h_{Z}(\tau)^{(+)}$, we clearly justify the Eq.\,\eqref{cluster_decom}. 

The further simplification for Eq.\,\eqref{cluster_decom} requires refined analysis of $\mathcal{G}_{\cl}(E)$ and $\mathcal{G}_{\cl}^{\cc}(G)$. 
Below, we discuss several properties of these newly defined quantities, which are summarized in the following lemma.
\begin{lemma}[Properties of $\mathcal{G}_{\mathrm{cl}}(E)$ and $\mathcal{G}_{\mathrm{cl}}^{\mathrm{c}}(G)$]\label{lemma_properties}
Let $\mathcal{G}_{\mathrm{cl}}(E)$ and $\mathcal{G}_{\mathrm{cl}}^{\mathrm{c}}(G)$ be defined as in Eqs.~\eqref{def_g_cl_E} and \eqref{def_g_cl_c_G}. For any multiset $G$, we denote by $V_{G} \coloneqq \{i \in V \colon \exists Z \in G \text{ such that } Z \ni i\}$ its corresponding vertex subset. Then:
\begin{enumerate}
\item Given any $G \in \mathcal{G}_{\mathrm{cl}}(E)$, if $G' \in \mathcal{G}_{\mathrm{cl}}^{\mathrm{c}}(G)$, then $G'$ does not overlap with $X$, $Y$, and $G$, i.e., $V_{G'} \cap (V_{G} \cup X \cup Y) = \emptyset$.
\item Define $E^{\mathrm{c}}(G,X,Y) \coloneqq \{Z \in E \colon Z \cap (V_G \cup X \cup Y) = \emptyset\}$ to be the collection of all elements in $E$ that do not overlap with $X$, $Y$, and $G$. Then the set of all multisets generated by $E^{\mathrm{c}}(G,X,Y)$ equals $\mathcal{G}_{\mathrm{cl}}^{\mathrm{c}}(G)$, i.e.,
\begin{equation}\label{g_cl_c_G_2}
\mathcal{G}^{\mathrm{c}}_{\mathrm{cl}}(G) = \mathcal{M}(E^{\mathrm{c}}(G,X,Y)).
\end{equation}
\end{enumerate}
\end{lemma}

\begin{proof}
(1) We prove this by contradiction. First, suppose that $V_{G'} \cap X \neq \emptyset$. By the definitions of $\mathcal{G}_{\mathrm{cl}}(E)$ and $\mathcal{G}_{\mathrm{cl}}^{\mathrm{c}}(G)$, we know that there exists $\overline{G} \in \mathcal{M}(E)$ such that $G = w_{\mathrm{cl}}(\overline{G})$ and $\{G', G\} = \overline{G}$. Given the assumption $V_{G'} \cap X \neq \emptyset$, there exists $Z \in G'$ such that $Z \cap X \neq \emptyset$. Therefore, $(G\oplus Z)\oplus X\oplus Y$ is connected. Since $|G\oplus Z| > |G|$, we see that $G$ is not the largest subset of $\overline{G}$ such that it constitutes a connected multiset with $X$ and $Y$, which implies $w_{\mathrm{cl}}(\overline{G}) \neq G$. This contradiction proves that $V_{G'} \cap X = \emptyset$. Similar arguments show that $V_{G'} \cap Y = \emptyset$ and $V_{G'} \cap V_{G} = \emptyset$. Hence, $V_{G'} \cap (V_{G} \cup X \cup Y) = \emptyset$.

(2) From claim (1), we know that every element in $\mathcal{G}_{\mathrm{cl}}^{\mathrm{c}}(G)$ does not overlap with $V_{G}$, $X$, and $Y$, which implies $\mathcal{G}_{\mathrm{cl}}^{\mathrm{c}}(G) \subseteq \mathcal{M}(E^{\mathrm{c}}(G,X,Y))$. Conversely, it is straightforward to verify that every multiset $G' \in \mathcal{M}(E)$ that does not overlap with $V_{G}$, $X$, and $Y$ also belongs to $\mathcal{G}_{\mathrm{cl}}^{\mathrm{c}}(G)$, i.e., $\mathcal{M}(E^{\mathrm{c}}(G,X,Y)) \subseteq \mathcal{G}_{\mathrm{cl}}^{\mathrm{c}}(G)$. This confirms the second claim.
\end{proof}

{~}

\hrulefill{\bf [ End of Proof of Lemma~\ref{lemma_properties}]}

{~}

We proceed with Eq.\,\eqref{cluster_decom} to obtain 
\begin{equation}\label{cluster_decom_interacting}
\begin{aligned}[b]
e^{-\beta W^{(+)}}\widetilde{\rho}_{\cl}&=\frac{1}{\mathcal{Z}^{2}}\sum_{G^{\ri}\in \mathcal{G}_{\cl}(E)}\exp\qty(-\beta W^{(+)}_{\widetilde{V}_{G^{\ri}}})\exp\qty(-\beta W^{(+)}_{\widetilde{V}^{\cc}_{G^{\ri}}})\sum_{w^{\ri}\in \mathcal{S}(G^{\ri})}\widetilde{F}(w^{\ri})\cdot \sum_{G^{\rii}\in \mathcal{G}^{\cc}_{\cl}(G^{\ri})}\sum_{w^{\rii}\in \mathcal{S}(G^{\rii})}\widetilde{F}(w^{\rii})
\\&=\frac{1}{\mathcal{Z}^{2}}\sum_{G\in \mathcal{G}_{\cl}(E)}\exp\qty(-\beta W^{(+)}_{\widetilde{V}_{G}})\sum_{w\in \mathcal{S}(G)}\widetilde{F}(w)\cdot \exp\qty(-\beta H^{(+)}_{\widetilde{V}^{\cc}_{G}}).
\end{aligned}
\end{equation}
Here, $\widetilde{V}_{G}=V_{G}\cup X \cup Y$. Equation \,\eqref{cluster_decom_interacting} will be the starting point to deploy the interaction-picture cluster expansion technique for the long-range bosonic system.

\section{Long-Range Bosonic Clustering Theorem and Implication}
Now, let us restrict ourselves to the long-range Bose-Hubbard model (with the squeezing term), 
\begin{equation}\label{slr_bose_hubbard}
H = -\sum_{i \neq j \in V} \qty[ J_{i,j} (a_i^\dagger a_j + \text{h.c.}) +\widetilde{J}_{i,j} (a_i^\dagger a^{\dagger}_j + \text{h.c.})] + \sum_{i \in V} \left[ \frac{U_i}{2} n_i(n_i - 1) - \mu_i n_i \right],
\end{equation}
with $\max\{|J_{i,j}|,|\widetilde{J}_{i,j}|\}\leq g (1+d_{i,j})^{-\as}$, $|\mu_i| \leq \mu < \infty$, $ 0 < U_{\min} \leq U_i \leq U_{\max} < \infty$ and $\alpha>D$. It can be readily verified that the model \eqref{slr_bose_hubbard} is also covered by condition \eqref{slr_condition}.
Note that in Eq.\,\eqref{slr_bose_hubbard}, we simply have $k=2$ and $E=\{\{i,j\}: i,j\in V\}$. We also denote $(J^{(1)}_{i,j},J^{(2)}_{i,j})=(J_{i,j},\widetilde{J}_{i,j})$ to keep the notation aligned with the convention in Eq.\,\eqref{sh_Z}. If we set all $U_{i}=0$, then Eq.\,\eqref{slr_bose_hubbard} reduces to the long-range free bosons.
Before we move on to the proof for the clustering theorem, we introduce
\begin{definition}[Boson Operator Numbers]\label{definition_boson_operator_number}
For an operator $A$ being the product of creation ($\{a_{i}\}_{i\in V}$) and annihilation operators ($\{a_{i}\}_{i\in V}$), we denote $\mathsf{N}_{i}(A)$ as the total numbers of operators at site $i$ offered by $A$. We also define $\mathsf{N}(A)=\sum_{i\in V_{A}}\mathsf{N}_{i}(A)$ to present the total numbers of operators offered by $A$. 
\end{definition}
\begin{remark}
For example, $\mathsf{N}_{i}(a_{i})=\mathsf{N}(a_{i})=1$, $\mathsf{N}_{i}(n_{i}^{l})=\mathsf{N}(n_{i}^{l})=2l$, $\mathsf{N}_{i}(a_{i}^{\dagger}a_{j})=1$ and $\mathsf{N}(a_{i}^{\dagger}a_{j})=2$. 
\end{remark} 
Then we have
\begin{theorem}[Boson Power-Law Clustering Theorem]\label{stheorem_boson_lr_clustering}
In the long-range Bose-Hubbard model ($\alpha > D$) defined in Eq.\,\eqref{slr_bose_hubbard} over a finite lattice $V$, let operators $O_{X}$ and $O_{Y}$ are products of $\{a_{i}^{\dagger},a_{i}\}_{i\in X}$ and $\{a_{i}^{\dagger},a_{i}\}_{i\in Y}$, respectively. Define the function $\Phi(\beta)\coloneq [\mathsf{N}(O_{X})!]^{1/2} [\mathsf{N}(O_{Y})!]^{1/2}\beta^{-[\mathsf{N}(O_{X})+\mathsf{N}(O_{Y})]/4}$. Then above a threshold temperature $\beta \leq \beta_{c}=\mathcal{O}(1)$,
the following inequality holds for correlation function between $O_{X}$ and $O_{Y}$:
\begin{equation}
|C_{\beta}(O_{X},O_{Y})| \leq C^{|X|+|Y|}\cdot   \frac{\Phi(\beta)}{(1+d_{X,Y})^{\alpha}}.
\end{equation} 
Here, the constant $C=\mathcal{O}(1)$ depends only on the parameter of models (i.e., $U_{\min},U_{\max},\mu,\alpha,g,D$) and the threshold temperature $\beta_{c}$.
\end{theorem}
\begin{remark}
If $\widetilde{O}_{X}$ and $\widetilde{O}_{Y}$ are polynomial functions of $\{a_{i}^{\dagger},a_{i}\}_{i\in X}$ and $\{a_{i}^{\dagger},a_{i}\}_{i\in Y}$ respectively, we can  always decompose $\widetilde{O}_{X}=\sum_{\ell_{1}}A_{\ell_{1}}\widetilde{O}_{X,\ell_{1}}$ and $\widetilde{O}_{Y}=\sum_{\ell_{2}}A_{\ell_{2}}\widetilde{O}_{Y,\ell_{2}}$ such that $\{\widetilde{O}_{X,\ell_{1}}\}$ and $\{\widetilde{O}_{Y,\ell_{2}}\}$ satisfy the precondition of Theorem~\ref{theorem_boson_lr_clustering} and obtain
\begin{equation}\label{}
\begin{aligned}[b]
|C_{\beta}(\widetilde{O}_{X},\widetilde{O}_{Y})|\leq \sum_{\ell_{1},\ell_{2}}|A_{\ell_{1}}A_{\ell_{2}}||C_{\beta }(\widetilde{O}_{X,\ell_{1}},\widetilde{O}_{Y,\ell_{2}})|
\end{aligned}
\end{equation}
to utilize the Theorem~\ref{theorem_boson_lr_clustering}.
\end{remark}
\begin{proof}[Proof of Theorem \ref{stheorem_boson_lr_clustering}]
We start from Eqs.\,\eqref{cluster_decom_interacting} and \eqref{cor_rho_cl} to write
\allowdisplaybreaks[4]
\begin{align*}\label{cor_x_y}
C_{\beta}(O_X,O_Y)&=\frac{1}{\mathcal{Z}^{2}}\sum_{G\in \mathcal{G}_{\cl}(E)}\sum_{w\in \mathcal{S}(G)}\operatorname{Tr}\qty[e^{-\beta W^{(+)}_{\widetilde{V}_{G}}}\widetilde{F}(w)e^{-\beta H^{(+)}_{\widetilde{V}^{\cc}_{G}}}O_{X}^{(0)}O_{Y}^{(1)}]
\\&=\frac{1}{\mathcal{Z}^{2}}\sum_{G\in \mathcal{G}_{\cl}(E)}\sum_{w\in \mathcal{S}(G)}\operatorname{Tr}_{\widetilde{V}_{G}}\qty[e^{-\beta W^{(+)}_{\widetilde{V}_{G}}}\widetilde{F}(w)O_{X}^{(0)}O_{Y}^{(1)}]\cdot \operatorname{Tr}_{\widetilde{V}_{G}^{\cc}}e^{-\beta H^{(+)}_{\widetilde{V}^{\cc}_{G}}}
\\&\leq \sum_{G\in \mathcal{G}_{\cl}(E)}\qty(\frac{1}{\operatorname{Tr}_{\widetilde{V}_{G}}e^{-\beta W_{\widetilde{V}_{G}}}})^{2}\sum_{w\in \mathcal{S}(G)}\qty|\operatorname{Tr}_{\widetilde{V}_{G}}\qty[e^{-\beta W^{(+)}_{\widetilde{V}_{G}}}\widetilde{F}(w)O_{X}^{(0)}O_{Y}^{(1)}]|\\&\leq \sum_{m=0}^{\infty}\sum_{G\in \mathcal{G}_{\cl}(E): |G|=m}\qty(\frac{1}{\operatorname{Tr}_{\widetilde{V}_{G}}e^{-\beta W_{\widetilde{V}_{G}}}})^{2}\sum_{w\in \mathcal{S}(G)}\qty|\operatorname{Tr}_{\widetilde{V}_{G}}\qty[e^{-\beta W^{(+)}_{\widetilde{V}_{G}}}\widetilde{F}(w)O_{X}^{(0)}O_{Y}^{(1)}]|.
\refstepcounter{equation}\tag{\theequation}
\end{align*}
Here, to obtain the third line, we used Proposition \ref{pro_subsystem_partition}. 
Then we use Lemma~\ref{lemma_trace_Ftil_w} to further obtain
\begin{equation}\label{cor_X_Y_proof_2}
\begin{aligned}[b]
&|C_{\beta}(O_{X},O_{Y})|
\\\leq& \sum_{m=0}^{\infty}\sum_{G\in \mathcal{G}_{\cl}(E): |G|=m}\qty(\frac{1}{\operatorname{Tr}_{\widetilde{V}_{G}}e^{-\beta W_{\widetilde{V}_{G}}}})^{2}\sum_{w\in \mathcal{S}(G)}\norm{e^{-\beta W^{(+)}_{\widetilde{V}_{G}}}\widetilde{F}(w)O_{X}^{(0)}O_{Y}^{(1)}}_{1}
\\\leq& \Phi(\beta) \sum_{m=0}^{\infty}\sum_{G\in \mathcal{G}_{\cl}(E): |G|=m}\qty(\frac{1}{\operatorname{Tr}_{\widetilde{V}_{G}}e^{-\beta W_{\widetilde{V}_{G}}}})^{2}\sum_{w\in \mathcal{S}(G)}\qty(\frac{C_{1}}{\sqrt{\beta}})^{2|\widetilde{V}_{G}|}C_{2}^{[\mathsf{N}(O_{X})+\mathsf{N}(O_{Y})]/2}
\frac{\qty(C_{3}\sqrt{\beta})^{m}}{m!} \prod_{x\in V_{w}}[m_{x}(w,\vec{s})!]^{1/2} \\&\quad \times\prod_{l=1}^{m}\qty(\sum_{s_{l}}|J_{w_{l}}^{(s_{l})}|)
\\\leq & C_{4}^{|X|+|Y|}\Phi(\beta)\sum_{m=0}^{\infty}\sum_{G\in \mathcal{G}_{\cl}(E): |G|=m}\frac{\qty(C_{3}C_{5}^{k}\sqrt{\beta})^{m}}{m!}\sum_{w\in \mathcal{S}(G)} \prod_{x\in V_{w}}[m_{x}(w,\vec{s})!]^{1/2}\prod_{l=1}^{m}\qty(\sum_{s_{l}}|J_{w_{l}}^{(s_{l})}|).
\end{aligned}
\end{equation}
Here, we used for high temperatures $\beta\leq \beta_{c}$ that  $\qty(\operatorname{Tr}_{\widetilde{V}_{G}}e^{-\beta W_{\widetilde{V}_{G}}})^{-1}\leq (C_{0,1}\beta^{1/2})^{|\widetilde{V}_{G}|}$ and used $(C_{0,1}^{-1}C_{1})^{|\widetilde{V}_{G}|}\leq (\max\{1,C_{0,1}^{-1}C_{1}\})^{km+|X|+|Y|}\eqcolon C_{5}^{km}C_{5}^{|X|+|Y|}$.  In fact we also used $\mathsf{N}(O_{X})+\mathsf{N}(O_{Y})\leq \max_{i\in X\cup Y }\mathsf{N}_{i}(O_{X}O_{Y}) (|X|+|Y|)$ and denote $C_{4}\coloneq \{1,C^{1/2}_{2}\}^{\max_{i\in X\cup Y }\mathsf{N}_{i}(O_{X}O_{Y})}C_{5}$. To proceed, we need to investigate the summation in the following form:
\begin{equation}\label{}
\begin{aligned}[b]
\sum_{G\in \mathcal{G}_{\cl}(E)}\frac{C^{|G|}}{|G|!}\sum_{w\in \mathcal{S}(G)} f(w), \quad \text{with }  f(w)\coloneq \prod_{x\in V_{w}}[m_{x}(w,\vec{s})!]^{1/2}\prod_{l=1}^{m}\qty(\sum_{s_{l}}|J_{w_{l}}^{(s_{l})}|) \text{ and } C\coloneq C_{3}C_{5}^{k}\sqrt{\beta}.
\end{aligned}
\end{equation}

To approach this estimation, we employ a strategy similar to the one used for Eq.\,\eqref{srho_cl}. Our methodology consists of two main steps: 
First, we decompose the multisets in Eq.\,\eqref{cor_X_Y_proof_2} according to specific criteria. Then, we analyze the reverse process of this decomposition to derive our estimate.
This approach enables us to systematically address complex summation structures of this kind.

We first note that for $G \in \mathcal{G}_{\text{cl}}(E)$, there always exists a decomposition $G = G_1\oplus G_2$ such that the multiset $G_1\oplus X\oplus Y$ is simply connected. Intuitively, this means there are no ``holes'' in the contour consisting of $G_1\oplus X\oplus Y$. 
For example, as illustrated in Fig.\,\ref{w_cl_s_example}, consider $G = \{Z_1, Z_2, \ldots, Z_8\}$. The multiset $G_1\oplus X\oplus Y$ is simply connected if we choose $G_1 = \{Z_1, Z_2, Z_3\}$ or $G_1 = \{Z_4, Z_5, Z_6\}$. However, it is not simply connected if we take $G_1 = \{Z_1, Z_2, \ldots, Z_6\}$.
We remark that the concept of a simply connected multiset can be defined rigorously as follows, using the concept of a dual associated graph:
\begin{definition}[Dual Associated Graph]
The associated dual graph of a multiset $G$ is defined as follows: the elements in its underlying set $\mathcal{D}(G)=\{g_1, g_2, \ldots, g_m\}$ serve as vertices, and for any $i \neq j$, there exists an edge between vertices $g_i$ and $g_j$ if and only if $g_i$ overlaps with $g_j$ (i.e., $V_{g_i} \cap V_{g_j} \neq \emptyset$).
\end{definition}
Then we call a multiset $G$ simply connected if and only if its associated dual graph is acyclic.

Based on the analysis above, we denote by $\mathcal{SC}$ the set of all simply connected multisets whose elements are vertex subsets. We then introduce
\begin{equation}\label{g_cl_s}
\mathcal{G}_{\text{cl}}^{\text{s}}(E) \coloneqq \{G \in \mathcal{M}(E) \colon G\oplus X\oplus Y \in \mathcal{SC}\}
\end{equation}
to collect all multisets in $\mathcal{M}(E)$ that simply connect $X$ and $Y$. However, not all multisets connecting $X$ and $Y$ are simply connected. To address the remaining multisets, for a fixed $G \in \mathcal{G}_{\text{cl}}^{\text{s}}(E)$, we define the set
\begin{equation}
\mathcal{G}_{\text{cl}}^{\text{r}}(G) = \{G' \in \mathcal{M}(E) \colon G\oplus G'\oplus X\oplus Y \in \mathcal{C}\},
\end{equation}
whose elements can be combined with those in Eq.\,\eqref{g_cl_s} to constitute a multiset in $\mathcal{G}_{\text{cl}}(E)$.
\begin{figure}[h]
\centering
\includegraphics[width=0.6\linewidth]{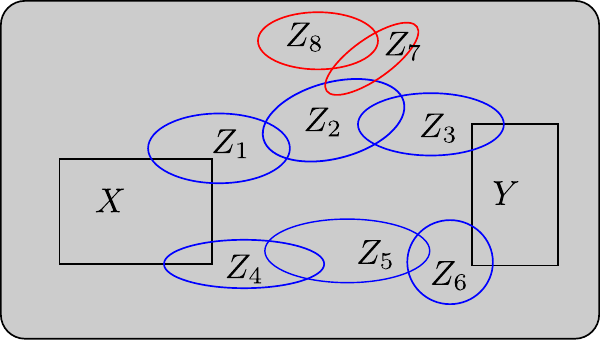}
\caption{Illustration of the concept of a ``simply connected'' multiset.}
\label{w_cl_s_example}
\end{figure}
We then have 
\begin{equation}\label{}
\begin{aligned}[b]
\sum_{G\in \mathcal{G}_{\cl}(E)}\frac{C^{|G|}}{|G|!}\sum_{w\in \mathcal{S}(G)} f(w)\leq \sum_{G^{\ri}\in \mathcal{G}_{\cl}^{\s}(E)}\sum_{w^{\ri}\in \mathcal{S}(G^{\ri})}f(w^{\ri})\cdot \sum_{G^{\rii}\in \mathcal{G}_{\cl}^{\rr}(G^{\ri})}\sum_{w^{\rii}\in \mathcal{S}(G^{\rii})}f(w^{\rii})\cdot \frac{|G|!}{|G^{\ri}|!|G^{\rii}|!}\cdot \frac{2^{|G|}}{|G|!}\cdot C^{|G|}.
\end{aligned}
\end{equation}
Here, the factor $2^{|G|}$ comes from the reorganization of the product $\prod_{x\in V_{w}}[m_{x}(w,\vec{s})!]^{1/2}$ into $2^{m_{x}(w,\vec{s})}\prod_{x\in V_{w}}[m_{x}(w^{\ri},\vec{s})!]^{1/2}\prod_{x\in V_{w}}[m_{x}(w^{\rii},\vec{s})!]^{1/2}$. This comes from the fact that $m_{x}(w,\vec{s})=m_{x}(w^{\ri},\vec{s})+m_{x}(w^{\rii},\vec{s})$, giving us $m_{x}(w,\vec{s})!\leq 2^{m_{x}(w,\vec{s})}m_{x}(w^{\ri},\vec{s})!m_{x}(w^{\rii},\vec{s})!$ and note that $\prod_{x\in V_{w}}2^{m_{x}(w,\vec{s})}=4^{|w|}=4^{|G|}$. Here, we have noticed that each off-site term in the Bose-Hubbard model offers exactly one operator to the sites in its support. With these facts, we move further via
\begin{equation}\label{sum_c_G_G_fac}
\begin{aligned}[b]
\sum_{G\in \mathcal{G}_{\cl}(E)}\frac{C^{|G|}}{|G|!}\sum_{w\in \mathcal{S}(G)} f(w)\leq \sum_{t_{1}=0}^{\infty}\sum_{t_{2}=0}^{\infty}(2C)^{t_{1}+t_{2}}\sum_{G^{\ri}\in \mathcal{G}_{\cl}^{\s}(E): |G^{\ri}|=t_{1}}\frac{1}{t_{1}!}\sum_{w^{\ri}\in \mathcal{S}(G^{\ri})}f(w^{\ri})\cdot \sum_{G^{\rii}\in \mathcal{G}_{\cl}^{\rr}(G^{\ri}): |G^{\rii}|=t_{2}}\frac{1}{t_{2}!}\sum_{w^{\rii}\in \mathcal{S}(G^{\rii})}f(w^{\rii}).
\end{aligned}
\end{equation}
The first summation can be estimated via Lemma \ref{lemma_repro} with,
\begin{equation}\label{sum_t_1}
\begin{aligned}[b]
\sum_{G^{\ri}\in \mathcal{G}_{\cl}^{\s}(E): |G^{\ri}|=t_{1}}\frac{1}{t_{1}!}\sum_{w^{\ri}\in \mathcal{S}(G^{\ri})}f(w^{\ri})&=\frac{1}{t_{1}!}\sum_{G\in \mathcal{G}_{\cl}^{\s}(E): |G|=t_{1}}|\mathcal{P}(G)|\prod_{x\in V_{w}}[m_{x}(w,\vec{s})!]^{1/2}\cdot \prod_{l=1}^{m}\qty(\sum_{s_{l}}|J_{w_{l}}^{(s_{l})}|)
\\&\leq \sum_{G\in \mathcal{G}_{\cl}^{\s}(E): |G|=t_{1}}\prod_{x\in V_{w}}[m_{x}(w,\vec{s})!]^{1/2}\cdot\prod_{l=1}^{m}\qty(\sum_{s_{l}}|J_{w_{l}}^{(s_{l})}|)
\\&\leq \sum_{G\in \mathcal{G}_{\cl}^{\s}(E): |G|=t_{1}}(\sqrt{2})^{t_{1}}\prod_{l=1}^{m}\qty(\sum_{s_{l}}|J_{w_{l}}^{(s_{l})}|)
\\&\leq (\sqrt{2})^{t_{1}}\sum_{i_{1}\in X}\sum_{i_{2}\in V}\sum_{i_{3}\in V}...\sum_{i_{t_{1}}\in V}\sum_{i_{t_{1}+1}\in Y}\qty(\sum_{s_{1}}|J_{i_{1}i_{2}}^{(s_{1})}|)\qty(\sum_{s_{2}}|J_{i_{2}i_{3}}^{(s_{2})}|)...\qty(\sum_{s_{t_{1}}}|J_{i_{t_{1}}i_{t_{1}+1}}^{(s_{t_{1}})}|)
\\&\leq (\sqrt{2}g)^{t_{1}}(2^{\as}\nu)^{t_{1}-1}|X||Y|\frac{1}{(1+d_{X,Y})^{\as}}
\end{aligned}
\end{equation}
where in the third line we used the fact that in this case $k= 2$. For the second summation in Eq.\,\eqref{sum_c_G_G_fac}, we simply
use Proposition~\ref{pro_cluster_sum} with $\gamma=g\nu, \nu \coloneq \sup_{i\in V}\sum_{j\in V}(1+d_{i,j})^{-\alpha}$ (note that the on-site terms are not involved in the interaction-picture cluster expansion) and $\phi(Z)=\sum_{s}|J_{Z}^{(s)}|$ to obtain 
\allowdisplaybreaks[4]
\begin{align*}\label{sum_t_2}
\sum_{G^{\rii}\in \mathcal{G}_{\cl}^{\rr}(G^{\ri}): |G^{\rii}|=t_{2}}\frac{1}{t_{2}!}\sum_{w^{\rii}\in \mathcal{S}(G^{\rii})}f(w^{\rii})&=\sum_{G\in \mathcal{G}_{\cl}^{\rr}(G^{\ri}): |G|=t_{2}}\frac{1}{t_{2}!}\sum_{w\in \mathcal{S}(G)}\prod_{x\in V_{w}}[m_{x}(w,\vec{s})!]^{1/2}\prod_{l=1}^{m}\qty(\sum_{s_{l}}|J_{w_{l}}^{(s_{l})}|)
\\&\leq \sum_{G\in \mathcal{G}_{\cl}^{\rr}(G^{\ri}): |G|=t_{2}}\frac{1}{t_{2}!}\sum_{w\in \mathcal{S}(G)}\prod_{x\in V_{w}}\qty[m_{x}(w,\vec{s})^{m_{x}(w,\vec{s})}]^{1/2}\prod_{l=1}^{m}\qty(\sum_{s_{l}}|J_{w_{l}}^{(s_{l})}|)
\\&\leq \sum_{G\in \mathcal{G}_{\cl}^{\rr}(G^{\ri}): |G|=t_{2}}\frac{1}{t_{2}!}\sum_{w\in \mathcal{S}(G)}\prod_{l=1}^{t_{2}}\qty(\sum_{s_{l}}|J_{w_{l}}^{(s_{l})}|)N(w_{l}|w)
\\&\leq e^{c_{1}|G^{\ri}|/k}(c_{2}g\nu k)^{t_{2}}
\\&\leq e^{c_{1}(t_{1}k+|X|+|Y|)/k}(c_{2}g\nu k)^{t_{2}}.
\refstepcounter{equation}\tag{\theequation}
\end{align*}
Here, for the third line, we also used Lemma \ref{lemma_estimate_m_x}.
Putting Eqs.\,\eqref{sum_t_1} and \eqref{sum_t_2} into Eq.\,\eqref{sum_c_G_G_fac} and from Eq.\,\eqref{cor_X_Y_proof_2} we know that
\begin{equation}\label{sum_t1_t2}
\begin{aligned}[b]
&|C_{\beta}(O_{X},O_{Y})|
\leq C_{4}^{|X|+|Y|}\Phi(\beta)\sum_{t_{1}=0}^{\infty}\sum_{t_{2}=0}^{\infty}(2C_{3}C_{5}^{k}c_{2}kg\nu\sqrt{\beta})^{t_{1}+t_{2}}(2^{\as}\nu)^{t_{1}-1}|X||Y|\frac{1}{(1+d_{X,Y})^{\as}}e^{c_{1}(t_{1}k+|X|+|Y|)/k}.
\end{aligned}
\end{equation}
To ensure the convergence of the summation, we require $2C_{3}C_{5}^{k}c_{2}kg\nu\sqrt{\beta}\cdot 2^{\as}\nu<1$ and $2C_{3}C_{5}^{k}c_{2}kg\nu\sqrt{\beta}<1$. Equivalently, we set the threshold temperature $\beta\leq\beta_{c}< (2^{\as+1}C_{3}C_{5}^{k}c_{2}kg\nu^{2})^{-2}$. Thus, we complete the proof of the present theorem. 
\end{proof}
\begin{remark}
Redefine $\Phi(\beta)\coloneq [\mathsf{N}(O_{X})!]^{1/2} [\mathsf{N}(O_{Y})!]^{1/2}\beta^{-[\mathsf{N}(O_{X})+\mathsf{N}(O_{Y})]/2}$ we can also obtain
\begin{equation}
|C_{\beta}(O_{X},O_{Y})| \leq \mathcal{O}(1)^{|X|+|Y|}\cdot   \frac{\Phi(\beta)}{(1+d_{X,Y})^{\alpha}}.
\end{equation} 
for free bosons. 
\end{remark}

{~}

\hrulefill{\bf [ End of Proof of Theorem~\ref{stheorem_boson_lr_clustering}]}

{~}

The clustering theorem directly leads us to the following thermal area law. 

\begin{corollary}\label{cor_thermal_area_law_sm}
Let $\rho_{\beta}$ be the thermal state of the long-range Bose-Hubbard model defined by Eqs.~\eqref{slr_bose_hubbard} with $\alpha > D$ on a finite lattice $V$. Let 
$V=A\sqcup B$. be a bypartition and $(\rho_{A},\rho_{B})\coloneq (\operatorname{Tr}_{B}\rho_{\beta}, \operatorname{Tr}_{A}\rho_{\beta})$ be the reduced density operators. At high temperatures $\beta < \beta_{c}=\mathcal{O}(1)$, the mutual information between subsystems $A$ and $B$ satisfies $\mathcal{I}(A:B)\leq C \cdot \beta^{1/2} \cdot |\partial_{AB}|$
with $C$ being an $\mathcal{O}(1)$ constant and $|\partial _{AB}|$ being the number of sites in the mutual boundary region of $A$ and $B$.
\end{corollary}
\begin{proof}
We introduce the free energy functional $F(\bullet)\coloneq \operatorname{Tr}\qty(H\bullet)-\beta^{-1}S(\bullet)$, where $S(\bullet)$ is the von Neumann entropy. From the Gibbs variational principle \cite{lemm2023thermal,alhambra2023quantum}, we have the inequality $F(\rho_{A} \otimes \rho_{B})\geq F(\rho_{\beta})$. This, combined with the definition of mutual information and the additivity of entropy for product states, $S(\rho_{A}\otimes \rho_{B})=S(\rho_{A})+S(\rho_{B})$, leads to
\begin{equation}\label{IAB_HI}
\mathcal{I}(A:B) \leq \beta \operatorname{Tr}[H(\rho_{A}\otimes \rho_{B}-\rho_{\beta})].
\end{equation}
By decomposing the Hamiltonian as $H = H_A + H_B + H_\partial$, we note that the bound is controlled entirely by the boundary term $H_{\partial}=-\sum_{i\in A,j\in B}J_{i,j}(a_{i}^{\dagger}a_{j}+a_{j}^{\dagger}a_{i})$.

A key simplification arises from observing that $\operatorname{Tr}[H_{\partial}(\rho_{A}\otimes \rho_{B})]=0$. This is a consequence of the total particle number conservation, which implies that single-annihilation-operator expectation values vanish; for instance, $\operatorname{Tr}_{A}(a_{i}\rho_{A}) = \operatorname{Tr}(a_{i}\rho_{\beta}) = 0$. Since terms in the expansion of $\operatorname{Tr}[H_{\partial}(\rho_{A}\otimes \rho_{B})]$ factorize into products of such expectation values (e.g., $\operatorname{Tr}_A(a_i^\dagger \rho_A) \operatorname{Tr}_B(a_j \rho_B)$), the entire term is zero. The inequality for the mutual information thus becomes
\begin{equation}\label{}
\begin{aligned}[b]
\mathcal{I}(A:B) \leq \beta  \sum_{i\in A,j\in B} \qty{|J_{i,j}| \left|\operatorname{Tr}\qty[(a_{i}^{\dagger}a_{j}+a_{j}^{\dagger}a_{i})\rho_{\beta}]\right|+|\widetilde{J}_{i,j}| \left|\operatorname{Tr}\qty[(a_{i}^{\dagger}a_{j}^{\dagger}+a_{j}a_{i})\rho_{\beta}]\right|}.
\end{aligned}
\end{equation}
The vanishing single-operator expectation values also allow us to recast this bound in terms of two-point correlation functions:
\begin{equation}\label{area_law_3}
\mathcal{I}(A:B) \leq \beta  \sum_{i\in A,j\in B}\qty{|J_{i,j}| \qty[|C_{\beta}(a_{i}^{\dagger},a_{j})|+|C_{\beta}(a_{j}^{\dagger},a_{i})|]+|\widetilde{J}_{i,j}| \qty[|C_{\beta}(a_{i}^{\dagger},a_{j}^{\dagger})|+|C_{\beta}(a_{j},a_{i})|]}.
\end{equation}
We now apply the clustering property of correlations from Theorem \ref{theorem_boson_lr_clustering}. Noting that the local operator norms are $\mathsf{N}(a_{i}^{\dagger})=\mathsf{N}(a_{j})=1$, we obtain the bound $|C_{\beta}(a_{i}^{\dagger},a_{j})| \leq C_{1}^{2}\beta^{-1/2}d_{i,j}^{-\alpha}$ for $i\in A, j\in B$, where $C_{1}$ is the $\mathcal{O}(1)$ constant from the theorem. Similar inequalities hold for other correlation functions in Eq.\,\eqref{area_law_3}. Substituting these bounds and the condition on the couplings $\max\{|J_{i,j}|,|\widetilde{J}_{i,j}|\}\leq g (1+d_{i,j})^{-\as}$ from Eq.~\eqref{slr_condition} into Eq.~\eqref{area_law_3}, we find 
\begin{equation}\label{}
\begin{aligned}[b]
\mathcal{I}(A:B)\leq 2C_{1}^{2}g\beta^{1/2}\sum_{i\in A,j\in B}d_{i,j}^{-2\alpha}.
\end{aligned}
\end{equation}
It is known that for $2\alpha>D+1$, where $D$ is the lattice dimension, the summation is upper-bounded by a term proportional to the boundary area, $\sum_{i\in A,j\in B}d_{i,j}^{-2\alpha} \lesssim |\partial_{AB}|$ \cite{minato2022fate}. This completes the proof.
\end{proof}
{~}

\hrulefill{\bf [ End of Proof of Corollary~\ref{cor_thermal_area_law_sm}]}

{~}

\section{Low Boson Density and Implication}
The low-boson-density inequality, which provides a concentration bound on the local particle number, is a foundational result with significant applications, as detailed in the main text. To the best of our knowledge, this inequality was first rigorously established for high-temperature thermal states of the finite-range Bose-Hubbard model in Ref.\,\cite{tong2024boson}. In the present work, we extend this analysis to the long-range Bose-Hubbard model, defined in Eq.\,\eqref{slr_bose_hubbard}. Furthermore, we demonstrate a key application of this bound to the classical simulation of bosonic systems. Specifically, the concentration of the particle number permits a controlled truncation of the infinite-dimensional local Hilbert space, thereby enabling efficient classical algorithms with a provably negligible error.

\begin{theorem}[Low Density Inequality]\label{stheorem_low_density}
In long-range Bose-Hubbard model ($\alpha > D$) defined in Eq.\,\eqref{slr_bose_hubbard} on a finite lattice $V$, above a temperature threshold $\beta \leq \beta_{c}$, we have for $l\in \mathbb{N}$ and $i\in V$ that $\operatorname{Tr}(n_{i}^{l}\rho_{\beta})\leq C^{l}\cdot l!\cdot \beta^{-l/2}$ with $C=\mathcal{O}(1)$ independent from $i$ and $l$. 
\end{theorem}
\begin{proof}
To address the low-density condition, we begin with the thermal average of a local operator $O_X$:
\begin{equation}\label{}
\begin{aligned}[b]
\langle O_X \rangle_{\beta H}=\frac{1}{\mathcal{Z}}\operatorname{Tr}\qty(O_Xe^{-\beta H})=\frac{1}{\mathcal{Z}}\operatorname{Tr}\qty[O_Xe^{-\beta W}S(\beta)]
\end{aligned}
\end{equation}
with
\begin{equation}\label{dyson_single}
\begin{aligned}[b]
S(\beta)=\sum_{G\in \mathcal{M}(E)}\sum_{w\in \mathcal{S}(G)}F(w), \quad F(w)\coloneq  \ii{0}{\beta}{\tau_{1}}\ii{0}{\tau_{1}}{\tau_{2}}...\ii{0}{\tau_{|w|-1}}{\tau_{|w|}} h_{w_{1}}(\tau_{1})h_{w_{2}}(\tau_{2})...h_{w_{|w|}}(\tau_{|w|}).
\end{aligned}
\end{equation}
Then we proceed as follows (we omitted the label $G\in \mathcal{M}(E)$ when no potential confusion arises.)
\begin{equation}\label{X_ensemble}
\begin{aligned}[b]
\langle O_X \rangle_{\beta H}&=\frac{1}{\mathcal{Z}}\sum_{G\in \mathcal{M}(E)}\sum_{w\in \mathcal{S}(G)}\operatorname{Tr}\qty[O_Xe^{-\beta W}F(w)]
\\&=\frac{1}{\mathcal{Z}}\qty(\sum_{G:V_{G}\cap X\neq \emptyset}+\sum_{G:V_{G}\cap X= \emptyset})\sum_{w\in \mathcal{S}(G)}\operatorname{Tr}\qty[O_Xe^{-\beta W}F(w)]
\\&\leq \frac{1}{\mathcal{Z}}\sum_{G:V_{G}\cap X\neq \emptyset}\sum_{w\in \mathcal{S}(G)}\operatorname{Tr}\qty[O_Xe^{-\beta W}F(w)]+ \frac{\abs{\operatorname{Tr}_{X}\qty(O_X e^{-\beta W_{X}})}\cdot \operatorname{Tr}_{X^{\cc}}e^{-\beta H_{X^{\cc}}}}{\operatorname{Tr}e^{-\beta H}}.
\end{aligned}
\end{equation}
Here we used $\sum_{G: V_{G}\cap X=\emptyset}\sum_{w\in \mathcal{S}(G)}\operatorname{Tr}\qty[X e^{-\beta W}F(w)]=\operatorname{Tr}_{X}\qty(X e^{-\beta W_{X}})\cdot \operatorname{Tr}_{X^{\cc}}e^{-\beta H_{X^{\cc}}}$. Generally, by using Proposition~\ref{pro_subsystem_partition}, we upper bound RHS of Eq.\,\eqref{X_ensemble} via
\begin{equation}\label{X_ens_result}
\begin{aligned}[b]
\langle O_X \rangle_{\beta H}\leq \frac{1}{\mathcal{Z}}\sum_{G:V_{G}\cap X\neq \emptyset}\sum_{w\in \mathcal{S}(G)}\operatorname{Tr}\qty[O_Xe^{-\beta W}F(w)]+ \langle O_X \rangle_{\beta W_{X}} ,
\end{aligned}
\end{equation} 
In a more specific case where $X=n_{i}^{l}$ with $l\in \mathbb{N}$, supported on a single site $i$, then Eq.\,\eqref{X_ens_result} reduces to 
\begin{equation}\label{X_ens_result}
\begin{aligned}[b]
\langle n_{i}^{l} \rangle_{\beta H}\leq \frac{1}{\mathcal{Z}}\sum_{G\in \mathcal{M}(E):V_{G}\ni i}\sum_{w\in \mathcal{S}(G)}\operatorname{Tr}\qty[n_{i}^{l}e^{-\beta W}F(w)]+   \langle n_{i}^{l} \rangle_{\beta W_{i}}=\term{1}+\term{2} ,
\end{aligned}
\end{equation}

Now, the only task is to estimate the first term in Eq.\,\eqref{X_ens_result}, which is similar to the treatment in Theorem \ref{theorem_boson_lr_clustering} but is not implemented in the doubled Hilbert space. We start by noting that
\begin{equation}\label{X_ens_result_first_term}
\begin{aligned}[b]
\terms{1}{X_ens_result}=\operatorname{Tr}(n_{i}^{l}e^{-\beta W}\rho_{\cl}),
\end{aligned}
\end{equation}
where by further using Lemma \ref{lemma_shuffle_F} with setting $f_{Z}(\tau)=h_{Z}(\tau)$ we arrive at
\begin{equation}\label{rho_cl_single}
\begin{aligned}[b]
\rho_{\cl}=\frac{1}{\mathcal{Z}}\sum_{G\in \mathcal{M}(E):V_{G}\ni i}\sum_{w\in \mathcal{S}(G)}F(w)=\frac{1}{\mathcal{Z}}\sum_{G^{\ri}\in \mathcal{G}_{\cl}(i)}\sum_{w^{\ri}\in \mathcal{S}(G^{\ri})}\sum_{G^{\rii}\in \mathcal{G}_{\cl}^{\cc}(G^{\ri})}\sum_{w^{\rii}\in \mathcal{S}(G^{\rii})}F(w^{\ri})\cdot F(w^{\rii}).
\end{aligned}
\end{equation}
In Eq.\,\eqref{rho_cl_single}, we denote $\mathcal{G}_{\cl}(i)\coloneq \{G\in \mathcal{M}(E)\colon G\oplus\{i\}\in \mathcal{C}\}$ to collect all the multiset that constitute a connected multiset with $\{i\}$.  To obtain the RHS, we apply the same treatment as in Eq.\,\eqref{cluster_decom}.
By Eq.\,\eqref{cluster_decom_interacting}, we recognize that
\begin{equation}\label{exp_beta_W_rho_cl}
\begin{aligned}[b]
e^{-\beta W}\rho_{\cl}=\frac{1}{\mathcal{Z}}\sum_{G\in \mathcal{G}_{\cl}(i)}e^{-\beta W_{\widetilde{i}}}\sum_{w\in \mathcal{S}(G)}F(w) e^{-\beta H_{\widetilde{i}^{\cc}}}=\frac{1}{\mathcal{Z}}\sum_{G\in \mathcal{G}_{\cl}(i)}e^{-\beta W_{V_{G}}}\sum_{w\in \mathcal{S}(G)}F(w) e^{-\beta H_{V_{G}^{\cc}}}.
\end{aligned}
\end{equation}
Here, we use $i$ as a shorthand of $V_{i}=\{i\}$ and we note that $\widetilde{i}=i \cup V_{G}=V_{G}$. Putting Eq.\,\eqref{exp_beta_W_rho_cl} into Eq.\,\eqref{X_ens_result_first_term} we obatin
\begin{equation}\label{X_ens_result_first_term}
\begin{aligned}[b]
\terms{1}{X_ens_result}&=\frac{1}{\mathcal{Z}}\sum_{G\in \mathcal{G}_{\cl}(i)}\sum_{w\in \mathcal{S}(G)}\operatorname{Tr}\qty[n_{i}^{l}e^{-\beta W_{V_{G}}}F(w)e^{-\beta H_{V_{G}^{\cc}}}]
\\&\leq \sum_{G\in \mathcal{G}_{\cl}(i)}\frac{\operatorname{Tr}e^{-\beta H_{V^{\cc}_{G}}}}{\operatorname{Tr}e^{-\beta H}}\sum_{w\in \mathcal{S}(G)}\norm{n_{i}^{l}e^{-\beta W_{V_{G}}}F(w)}_{1}
\\&\leq \sum_{G\in \mathcal{G}_{\cl}(i)}\frac{1}{\operatorname{Tr}_{V_{G}}e^{-\beta W_{V_{G}}}}\sum_{w\in \mathcal{S}(G)}\norm{n_{i}^{l}e^{-\beta W_{V_{G}}}F(w)}_{1} .
\end{aligned}
\end{equation}
Here, we used Proposition~\ref{pro_subsystem_partition} to obtain the last line of Eq.\,\eqref{X_ens_result_first_term}. In the last line, we actually change the notation $\bullet(\tau)\coloneq e^{\tau W}\bullet e^{-\tau W}$ to  $\bullet(\tau)\coloneq e^{\tau W_{V_{G}}}\bullet e^{-\tau W_{V_{G}}}$ in $F(w)$.

The further estimation for the norm on the RHS of Eq.\,\eqref{X_ens_result_first_term} will be achieved by using Lemmas~\ref{lem_trace_f_w} and \ref{lemma_estimate_m_x}, i.e., we put the following inequality
\begin{equation}\label{}
\begin{aligned}[b]
\norm{n_{i}^{l}e^{-\beta W_{V_{G}}}F(w)}_{1}\leq \sqrt{(2l)!}\qty(\frac{2C_{1}}{\sqrt{\beta}})^{l}  \qty(\sqrt{\frac{2C_{1}C_{2}}{\beta}})^{|V_{G}|}\frac{(2eC_{1}\sqrt{\beta})^{m}}{m!}\prod_{l=1}^{m}\qty(\sum_{s_{l}}|J_{w_{l}}^{(s_{l})}|)N(w_{l}|w),
\end{aligned}
\end{equation}
into Eq.\,\eqref{X_ens_result_first_term} to arrive at
\begin{equation}\label{X_ens_result_first_term_2}
\begin{aligned}[b]
&\terms{1}{X_ens_result}
\\\leq& \sqrt{(2l)!}\qty(\frac{2C_{1}}{\sqrt{\beta}})^{l}\sum_{m=0}^{\infty}\sum_{G\in \mathcal{G}_{\cl}(i): |G|=m}\frac{1}{\operatorname{Tr}_{V_{G}}e^{-\beta W_{V_{G}}}}\sum_{w\in \mathcal{S}(G)} \qty(\sqrt{\frac{2C_{1}C_{2}}{\beta}})^{|V_{G}|}\frac{(2eC_{1}\sqrt{\beta})^{m}}{m!}\prod_{l=1}^{m}\qty(\sum_{s_{l}}|J_{w_{l}}^{(s_{l})}|)N(w_{l}|w)
\\\leq & \sqrt{(2l)!}\qty(\frac{2C_{1}}{\sqrt{\beta}})^{l}\sum_{m=0}^{\infty}\frac{(2eC_{1}C_{3}^{k}\sqrt{\beta})^{m}}{m!}\sum_{G\in \mathcal{G}_{\cl}(i): |G|=m}\sum_{w\in \mathcal{S}(G)}\prod_{l=1}^{m}\qty(\sum_{s_{l}}|J_{w_{l}}^{(s_{l})}|)N(w_{l}|w)
\\\leq & \sqrt{(2l)!}\qty(\frac{2C_{1}}{\sqrt{\beta}})^{l}\sum_{m=0}^{\infty}\frac{(2eC_{1}C_{3}^{k}\sqrt{\beta})^{m}}{m!}e^{c_{1}/k}(c_{2}g\nu k)^{m}m!.
\end{aligned}
\end{equation}
Here, we used $\operatorname{Tr}_{V_{G}}e^{-\beta W_{V_{G}}}\leq (C_{3,0}\beta^{1/2})^{|V_{G}|}$ with $C_{3,0}=\mathcal{O}(1)$ (see Lemma~5 in Ref.\,\cite{tong2024boson}) and denote $C_{3}\coloneq C_{3,0}\max\{2C_{1}C_{2},1\}=\mathcal{O}(1)$. To deal with the third line in Eq.\,\eqref{X_ens_result_first_term_2}, we simply
use Proposition~\ref{pro_cluster_sum} with $\gamma=g\nu, \nu \coloneq \sup_{i\in V}\sum_{j\in V}(1+d_{i,j})^{-\alpha}$ (note that the on-site terms are not involved in the interaction-picture cluster expansion) and $\phi(Z)=\sum_{s}|J_{Z}^{(s)}|$.
Clearly, the series on the right-hand side of Eq.\,\eqref{X_ens_result_first_term_2} converges above a threshold temperature, and therefore, we finish the proof. Note that we also need to use $(2l)!\leq 2^{2l}(l!)^{2}$. 
\end{proof}
{~}

\hrulefill{\bf [ End of Proof of Theorem~\ref{stheorem_low_density}]}

{~}
\begin{remark}
Parallel result:
in long-range free-boson model, above a temperature threshold $\beta \leq \beta_{c}$, we have for $l\in \mathbb{N}$ and $i\in V$ that $\langle n_{i}^{l} \rangle_{\beta H}\leq C^{l}\cdot l! \cdot \beta^{-l}$. Note that in this case, the series in Eq.\,\eqref{X_ens_result_first_term_2} will be temperature independent, and we should impose some extra conditions on the parameters of free bosons to ensure the convergence of the corresponding series (as well as the existence of the Gibbs state). 
\end{remark}
\begin{remark}
The temperature scaling here is optimal, i.e., it is the same as on-site case (with all $J_{ij}=\widetilde{J}_{ij}=0$).
\end{remark}

\begin{corollary}[Boson-Number Concentration]\label{coro_boson_concen}
Let $\{\ket{n}_{i}\}_{n\in \mathbb{N}}$ be the eigenstates of $n_{i}$ and define $p^{(i)}_{n}\coloneq \prescript{}{i}{\bra{n}}\operatorname{Tr}_{i^{\cc}}\rho_{\beta}\ket{n}_{i} $ as the boson-number concentration at state $n\in \mathbb{N}$ for $i\in V$. Then for high temperature $\beta\leq \beta_{c}$, we have $p_{n}^{(i)}\leq 2e^{-\sqrt{\beta}n/(2C)}$. Here, the constant $C=\mathcal{O}(1)$ is the same as that in the conclusion of Theorem \ref{stheorem_low_density}.
\end{corollary}
\begin{proof}
From Theorem \ref{stheorem_low_density}, we have for all $l\in \mathbb{N}, i\in V$ and $\beta\leq \beta_{c}$ that
\begin{equation}\label{moment_summation}
\begin{aligned}[b]
\operatorname{Tr}(n_{i}^{l}\rho_{\beta})=\sum_{n=0}^{\infty}n^{l}p^{(i)}_{n}\leq C^{l} l! \beta^{-l/2}.
\end{aligned}
\end{equation} 
We omit the subscript ``$(i)$'' for $p_{n}^{(i)}$ for brevity and consider
\begin{equation}\label{boson_concen_1}
\begin{aligned}[b]
\sum_{n=0}^{\infty}e^{cn}p_{n}=\sum_{n=0}^{\infty}\sum_{m=0}^{\infty}\frac{(cn)^{m}}{m!}p_{n}
= \sum_{m=0}^{\infty}\frac{c^{m}}{m!}\sum_{n=0}^{\infty}n^{m}p_{n} \leq \sum_{m=0}^{\infty}  \frac{c^{m}}{m!} (C\beta^{-1/2})^{m}m!=\sum_{m=0}^{\infty} (cC\beta^{-1/2})^{m}.
\end{aligned}
\end{equation}
By choosing $c=0.5C^{-1}\beta^{1/2}$ we arrive at
\begin{equation}\label{}
\begin{aligned}[b]
\sum_{n=0}^{\infty}e^{cn}p_{n}=\sum_{n=0}^{\infty}\sum_{m=0}^{\infty}\frac{(cn)^{m}}{m!}p_{n}=\frac{1}{1-cC\beta^{-1/2}},
\end{aligned}
\end{equation} 
which clearly implies
\begin{equation}\label{}
\begin{aligned}[b]
p_{n}\leq \frac{e^{-cn}}{1-cC\beta^{-1/2}}=\frac{e^{-0.5C^{-1}\beta^{1/2}n}}{1-0.5C^{-1}\beta^{1/2}C\beta^{-1/2}}=2e^{-\sqrt{\beta} n/(2C)}
\end{aligned}
\end{equation}
and finished the proof.
\end{proof}
{~}

\hrulefill{\bf [ End of Proof of Corollary~\ref{coro_boson_concen}]}

{~}

The results presented above enable us to estimate the error due to the truncation of on-site boson particle number, summarized as follows: 
\begin{lemma}\label{lemma_truncation_dim}
For any site $i\in V$, let $\{\ket{m}_{i}\}_{m\in \mathbb{N}}$ be the eigenstates of $n_{i}$ and define $\Pi_{i,\leq q}\coloneq \sum_{m=0}^{q}\ket{m}_{i}\prescript{}{i}{\bra{m}}$ as the local projection operator onto the subspace spanned by the first $q+1$ eigenstates of $n_{i}$. Then we have 
\begin{equation}\label{}
\begin{aligned}[b]
\|\Pi_{V,q}\rho_{\beta}\Pi_{V,q}-\rho_{\beta}\|_{1}\le 2\sum_{j=1}^{|V|}\qty[\mathrm{Tr}(\Pi_{j,>q}\rho_{\beta})]^{1/2}
\end{aligned}
\end{equation} 
with $\Pi_{X,q}\coloneq \bigotimes_{i\in X}\Pi_{i,\leq q}$ for any $X\subseteq V$ and $\Pi_{i,>q}\coloneq \sum_{m=q}^{\infty}\ket{m}_{i}\prescript{}{i}{\bra{m}}$.
\end{lemma}
\begin{proof}
Let us label $V=\{1,2,...,N\}$ and denote $X_{j}\coloneq \{1,2,...,j\}$ as well as $X_{0}=\emptyset$. Then we decompose
\begin{equation}\label{}
\begin{aligned}[b]
\Pi_{V,q}e^{-\beta H}\Pi_{V,q}-e^{-\beta H} = \sum_{j=0}^{|V|-1}(\Pi_{X_{j+1},q}e^{-\beta H}\Pi_{X_{j+1},q}-\Pi_{X_{j},q}e^{-\beta H}\Pi_{X_{j},q}),
\end{aligned}
\end{equation}
and have
\allowdisplaybreaks[4]
\begin{align*}\label{}
\|\Pi_{V,q}e^{-\beta H}\Pi_{V,q}-e^{-\beta H}\|_{1} &\le \sum_{j=0}^{|V|-1}\|\Pi_{X_{j+1},q}e^{-\beta H}\Pi_{X_{j+1},q}-\Pi_{X_{j},q}e^{-\beta H}\Pi_{X_{j},q}\|_{1}\\
& \le \sum_{j=0}^{|V|-1}\|\Pi_{X_{j},q}(\Pi_{j+1,\le q}e^{-\beta H}\Pi_{j+1,\le q}-e^{-\beta H})\Pi_{X_{j},q}\|_{1}\\
& \le \sum_{j=0}^{|V|-1}\|\Pi_{j+1,\le q}e^{-\beta H}\Pi_{j+1,\le q}-e^{-\beta H}\|_{1}  \\
& = \sum_{j=0}^{|V|-1}\|(\id_{j}-\Pi_{j+1,>q})e^{-\beta H}(\id_{j}-\Pi_{j+1,>q})-e^{-\beta H}\|_{1}  \\
& = \sum_{j=0}^{|V|-1}\|-\Pi_{j+1,>q}e^{-\beta H}(\id_{j}-\Pi_{j+1,>q})\|_{1} 
\le 2\sum_{j=0}^{|V|-1}\|\Pi_{j+1,>q}e^{-\beta H}\|_{1}.
\refstepcounter{equation}\tag{\theequation}
\end{align*}
Here, the operator $\id_{j}$ denotes the identity operator on the local Hilbert space $\mathcal{H}_{j}$. 
Finally, using the H\"{o}lder's inequality, we obtain
\begin{align*}
\|\Pi_{j+1,>q}e^{-\beta H}\|_{1} & \le \|\Pi_{j+1,>q}e^{-\beta H/2}\|_{2}\|e^{-\beta H/2}\|_{2}
\le \mathcal{Z}_{\beta}\left[\frac{\mathrm{Tr}(\Pi_{j+1,>q}e^{-\beta H}\Pi_{j+1,>q})}{\mathcal{Z}_{\beta}}\right]^{1/2},
\end{align*}
and hence by noticing $\rho_{\beta}=e^{-\beta H}/\mathcal{Z}_{\beta}$ we obtain 
\begin{equation}\label{}
\begin{aligned}[b]
\|\Pi_{V,q}\rho_{\beta}\Pi_{V,q}-\rho_{\beta}\|_{1}\le2\sum_{j=0}^{|V|-1}[\mathrm{Tr}(\Pi_{j+1,>q}\rho_{\beta})]^{1/2} =2\sum_{j=1}^{|V|}[\mathrm{Tr}(\Pi_{j,>q}\rho_{\beta})]^{1/2}
\end{aligned}
\end{equation}
to finish the proof.
\end{proof}
{~}

\hrulefill{\bf [ End of Proof of Lemma~\ref{lemma_truncation_dim}]}

{~}

\begin{corollary}\label{corollary_truncation_error}
Under the same settings of Lemma \ref{lemma_truncation_dim}, let $\mathcal{Z}_{\beta}^{(q)}\coloneq \operatorname{Tr}(\Pi_{V,q}e^{-\beta H})$ be the $q$-boson number truncated partition function. Denote $\theta>0$, $0<\eta<1$ and $f_{\beta}\coloneq 2\sqrt{2}\qty(1-e^{-\sqrt{\beta}/(2C)})^{-1/2}$ with the constant $C=\mathcal{O}(1)$ being the same as that in the conclusion of Theorem \ref{stheorem_low_density}. If we choose $q=4 C\beta^{-1/2}[(\theta+1)+\ln (\eta^{-1}f_{\beta})/\ln 2]\ln |V|$, then for high temperature $\beta\leq \beta_{c}$, we have  
\begin{equation}\label{}
\begin{aligned}[b]
|\mathcal{Z}_{\beta}-\mathcal{Z}_{\beta}^{(q)}|\leq \eta |V|^{-\theta} \mathcal{Z}_{\beta}.
\end{aligned}
\end{equation}
Here we have naturally assumed that $|V|\geq 2$ so that the Bose-Hubbard model is well-defined.
\end{corollary}
\begin{proof}
By noticing $\Pi_{V,q}^{2}=\Pi_{V,q}$ so that $\mathcal{Z}_{\beta}^{(q)}=\operatorname{Tr}(\Pi_{V,q}e^{-\beta H}\Pi_{V,q})$, we arrive at
\begin{equation}\label{re_cor_trun}
\begin{aligned}[b]
|\mathcal{Z}_{\beta}-\mathcal{Z}_{\beta}^{(q)}|=\operatorname{Tr}\qty(e^{-\beta H}-\Pi_{V,q}e^{-\beta H}\Pi_{V,q})\leq \mathcal{Z}_{\beta} \|\Pi_{V,q}\rho_{\beta}\Pi_{V,q}-\rho_{\beta}\|_{1}\leq 2\sum_{j=1}^{|V|}[\mathrm{Tr}(\Pi_{j,>q}\rho_{\beta})]^{1/2}\mathcal{Z}_{\beta},
\end{aligned}
\end{equation}
where we have already used Lemma \ref{lemma_truncation_dim}. Then, we upper bound the RHS of Eq.\,\eqref{re_cor_trun} as follows:
\begin{equation}\label{}
\begin{aligned}[b]
\text{RHS of \eq{re_cor_trun}}\leq 2\mathcal{Z}_{\beta}\sum_{j=1}^{|V|}\qty[\sum_{m=q}^{\infty}\mathrm{Tr}\qty(\ket{m}_{i}\prescript{}{i}{\bra{m}}\rho_{\beta})]^{1/2}=2\mathcal{Z}_{\beta}\sum_{j=1}^{|V|}\sqrt{\sum_{m=q}^{\infty}p_{m}^{(j)}}.
\end{aligned}
\end{equation}
From Corollary \ref{coro_boson_concen}, we have
$p_{m}^{(j)}\leq 2e^{-\sqrt{\beta}m/(2C)}$ and therefore
\begin{equation}\label{}
\begin{aligned}[b]
\text{RHS of \eq{re_cor_trun}}\leq 2\sqrt{2}\mathcal{Z}_{\beta} |V|\sqrt{\frac{e^{-\sqrt{\beta}q/(2C)}}{1-e^{-\sqrt{\beta}/(2C)}}}=2\mathcal{Z}_{\beta}\qty[\frac{2}{1-e^{-\sqrt{\beta}/(2C)}}]^{1/2}|V|e^{-\sqrt{\beta}q/(4C)}=f_{\beta}|V|e^{-\sqrt{\beta}q/(4C)}.
\end{aligned}
\end{equation}
Clearly, if we choose $q$ to be $4C\beta^{-1/2}\ln (\eta^{-1}f_{\beta}|V|^{\theta+1})=4C\beta^{-1/2}[\ln (\eta^{-1}f_{\beta})+(\theta+1)\ln |V|]$, the RHS of Eq.\,\eqref{re_cor_trun} equals $\eta |V|^{-\theta}$. Here note that $|V|\geq 2$, we choose a greater value for $q$, i.e., $4 C\beta^{-1/2}[(\theta+1)+\ln (\eta^{-1}f_{\beta})/\ln 2]\ln |V|$, which completes the proof.
\end{proof}
{~}

\hrulefill{\bf [ End of Proof of Corollary~\ref{corollary_truncation_error}]}

{~}
\begin{remark}
Note that $\eta |V|^{-\theta}<\eta<1$, the result presented above implies the following (to be aligned with the later sections we temprorarily use log to denote the logarithm), $|\log \mathcal{Z}_{\beta}-\log \mathcal{Z}^{(q)}_{\beta}|\leq -\log \qty(1-\eta |V|^{-\theta})$. By using $-\log(1-x) \leq -\log(1-x_{0})x/x_{0}$ for $0<x\leq x_{0}<1$ and choose $(x,x_{0})=(\eta |V|^{-\theta},\eta)$, we have  $|\log \mathcal{Z}_{\beta}-\log \mathcal{Z}^{(q)}_{\beta}|\leq -\log \qty(1-\eta)|V|^{-\theta}$. Finally by choosing $\eta=1-e^{-1}$ we arrive at the following more convenient form:
\begin{equation}\label{}
\begin{aligned}[b]
|\log \mathcal{Z}_{\beta}-\log \mathcal{Z}^{(q)}_{\beta}|\leq |V|^{-\theta}.
\end{aligned}
\end{equation}
\end{remark}

\section{Classical Algorithm for Partition Function}
To construct an efficient algorithm for the high-temperature partition function of the long-range Bose-Hubbard model [cf.\,Eq.\,(\ref{slr_bose_hubbard})], we first introduce the framework of abstract cluster expansion for polymer models. This method has proven to be powerful for various complexity problems in spin systems, as detailed in Refs.\,\cite{mann2024algorithmic,sanchez2025high}. However, it is not directly applicable to bosonic systems. By utilizing this method in conjunction with the interaction-picture cluster expansion and several crucial technical lemmas, we present a systematic approach to analyze bosonic systems from a complexity-theoretic perspective. For simplification, we focus on the canonical long-range Bose-Hubbard model, i.e., we set $\widetilde{J}_{ij}=0$ for any $i,j\in V$ in Eq.\,\eqref{slr_bose_hubbard}, which implies the conservation of total particle numbers.

First, we present the formal definition of a polymer model and summarize established results for later convenience. For a detailed discussion of polymer models, we refer the reader to Refs.\,\cite{kotecky1986cluster,malyshev1980cluster} or Chapter~5 of Ref.\,\cite{friedli2018statistical}. Recent works \cite{mann2024algorithmic,sanchez2025high} also provide a concise overview. Now, we begin with the definitions.
\begin{definition}[Abstract Polymer Model]\label{def_polymer_model}
Let $\mathfrak{C}$ be a countable set, whose elements are called polymers. Let $w: \mathfrak{C}\to\mathbb{C}$ be a weight function that assigns to each polymer $\gamma\in \mathfrak{C}$ a complex number $w_{\gamma}\in \mathbb{C}$, called the weight of $\gamma$. Let $\sim$ be a symmetric, irreflexive compatibility relation on $\mathfrak{C}$. The tuple $(\mathfrak{C}, w, \sim)$ is called an abstract polymer model.
\end{definition}
We are particularly interested in sets of pairwise compatible polymers, which are called admissible sets. The collection of all admissible sets is denoted by $\mathscr{G}$, i.e.,
\begin{equation}\label{eq:admissible_sets}
\mathscr{G}\coloneqq \{\Gamma\subseteq \mathfrak{C} \colon \forall \gamma,\gamma'\in \Gamma, \gamma\neq\gamma' \implies \gamma\sim \gamma'\}.
\end{equation}
By convention, the empty set $\emptyset$ is admissible. With these notations, the abstract polymer partition function is defined as
\begin{equation}\label{abs_paritition_function}
\mathcal{Z}(\mathfrak{C},w)\coloneqq \sum_{\Gamma\in \mathscr{G}}\prod_{\gamma\in \Gamma}w_{\gamma}.
\end{equation}
To proceed, we introduce the notion of an incompatibility graph.
\begin{definition}[Incompatibility Graph]\label{def_incom_graph}
Let $a = (\gamma_1, \dots, \gamma_n) \in \bigcup_{k=1}^{\infty} \mathfrak{C}^{\times k}$ be a non-empty ordered sequence of polymers. Let $\mathcal{D}(a)$ be the set of distinct polymers in $a$. The incompatibility graph $\mathsf{G}_{a}$ is the graph whose vertex set is $\mathcal{D}(a)$ and for any $i \neq j$, there exists an edge between vertices $\gamma_i$ and $\gamma_j$ if and only if $\gamma_i\nsim \gamma_{j}$.
\end{definition}
Let $\mathscr{G}_{\mathrm{c}}$ denote the collection of all finite, ordered sequences of polymers whose incompatibility graph is connected:
\begin{equation}\label{eq:connected_graphs}
\mathscr{G}_{\mathrm{c}}\coloneqq \left\{ a\in \bigcup_{k=1}^{\infty} \mathfrak{C}^{\times k}\colon \mathsf{G}_{a} \text{ is connected} \right\}.
\end{equation}
In some literatures \cite{kotecky1986cluster,mann2024algorithmic,sanchez2025high}, the elements of $\mathscr{G}_{\mathrm{c}}$ are called clusters. We avoid this terminology to prevent confusion with the previously introduced interaction-picture cluster expansion.

The abstract cluster expansion provides a series representation for the logarithm of the partition function. To align with standard notation in related fields, we use $\log$ rather than $\ln$ in this section. The expansion is given by
\begin{equation}\label{log_abstract_polymer}
\log \mathcal{Z}(\mathfrak{C},w) = \sum_{a\in \mathscr{G}_{\mathrm{c}}} \varphi(\mathsf{G}_{a}) \prod_{\gamma\in a} w_{\gamma},
\end{equation}
where $\varphi(\mathsf{G})$ is the Ursell function for a graph $\mathsf{G}$, defined as
\begin{equation}\label{ursell_function}
\varphi(\mathsf{G})\coloneqq \frac{1}{|V(\mathsf{G})|!} \sum_{\substack{S\subseteq E(\mathsf{G}):\\ (V(\mathsf{G}), S) \text{ is connected}}}(-1)^{|S|}.
\end{equation}
Here, $V(\mathsf{G})$ and $E(\mathsf{G})$ denote the vertex and edge sets of $\mathsf{G}$, respectively. The summation in Eq.\,\eqref{ursell_function} runs over all spanning, connected subgraphs $(V(\mathsf{G}),S)$. 
Since the ordering of elements in the sequence $a$ does not affect the terms in the sum on the RHS of Eq.\,\eqref{log_abstract_polymer}, the sum over sequences can be reformulated as a sum over multisets of polymers. A sequence $a$ is represented by the multiset denoted via $\widetilde{a} \coloneqq \bigcup_{\gamma\in a}(\gamma, \mu(\gamma))$, where the pairs consist of the distinct elements $\gamma$ in $a$ and their corresponding multiplicities $\mu_(\gamma)$.
Using this multiset representation, Eqs.\,\eqref{log_abstract_polymer} and \eqref{ursell_function} can be rewritten as
\begin{subequations}\label{eq:multiset_reformulation}
\begin{align}
\log \mathcal{Z}(\mathfrak{C},w) &= \sum_{\widetilde{a}\in \widetilde{\mathscr{G}}_{\mathrm{c}}} \varphi(\mathsf{G}_{\widetilde{a}}) \prod_{(\gamma, \mu) \in \widetilde{a}} \frac{w_{\gamma}^{\mu}}{\mu!} \label{eq:logZ_multiset} \\
\varphi(\mathsf{G}) &\coloneqq \sum_{\substack{S\subseteq E(\mathsf{G}):\\ (V(\mathsf{G}), S) \text{ is connected}}}(-1)^{|S|}, \label{eq:ursell_multiset}
\end{align}
\end{subequations}
where $\widetilde{\mathscr{G}}_{\mathrm{c}}$ is the collection of all finite multisets of polymers whose incompatibility graph is connected:
\begin{equation}\label{eq:connected_graphs_multiset}
\widetilde{\mathscr{G}}_{\mathrm{c}}\coloneqq \left\{ \bigcup_{\gamma\in a}(\gamma, \mu(\gamma)) \colon a\in \bigcup_{k=1}^{\infty} \mathfrak{C}^{\times k}, \text{the incompatibility graph } \mathsf{G}_{a} \text{ is connected} \right\}.
\end{equation}

The convergence of the series in Eq.\,\eqref{log_abstract_polymer} is a non-trivial problem. A sufficient condition is given by the following theorem.
\begin{theorem}[Main Theorem in Ref.\,\cite{kotecky1986cluster}]
\label{thm:KP}
Let $(\mathfrak{C},w,\sim)$ be an abstract polymer model. If there exist functions $\delta_{1}, \delta_{2}: \mathfrak{C} \mapsto \mathbb{R}^{+}$ such that for all $\gamma_{0}\in \mathfrak{C}$,
\begin{equation}\label{eq:convergence_condition}
\sum_{\gamma\nsim \gamma_{0}}|w_{\gamma}|e^{\delta_{1}(\gamma)+\delta_{2}(\gamma)}\leq \delta_{1}(\gamma_0),
\end{equation}
then the abstract cluster expansion converges absolutely, $\mathcal{Z}(\mathfrak{C},w)\neq 0$, and for any $\gamma_0 \in \mathfrak{C}$, we have the bound
\begin{equation}\label{eq:bound_on_clusters}
\sum_{\substack{a\in \mathscr{G}_{\mathrm{c}}: a\ni \gamma_{0}}} \left|\varphi(\mathsf{G}_{a})\prod_{\gamma\in a}w_{\gamma}\right| e^{\sum_{\gamma\in a}\delta_{2}(\gamma)}\leq \delta_{1}(\gamma_{0}).
\end{equation}
\end{theorem}
This theorem lays the foundation for constructing an approximation algorithm for the partition function, as we will demonstrate.

To apply these results to $k$-local quantum lattice models defined on a finite set of sites $V$, we identify our polymers with connected multisets, whose elements are drawn from $E$ (the vertex subsets with cardinality at most $k$). Formally, we have $\mathfrak{C}=\mathcal{C}(E)$ (see discussion below Eq.\,\eqref{w_cl}). To maintain consistency with the literature on polymer models, we hereafter denote multisets by $\gamma$ instead of $G$. The compatibility relation is defined by disjointness: two polymers $\gamma$ and $\gamma'$ are compatible, denoted $\gamma\sim \gamma'$, if and only if their underlying vertex sets do not overlap, i.e., $V_{\gamma}\cap V_{\gamma'}=\emptyset$. 

The next step is to connect the physical partition function $\mathcal{Z}_{\beta}$ to the polymer-model partition function $\mathcal{Z}(\mathfrak{C},w)$ as defined in Eq.\,\eqref{abs_paritition_function}. 
We first define $\mathcal{Z}^{(q)}_{W}\coloneq \operatorname{Tr}\qty(\Pi_{V,q}e^{-\beta W})$ and consider the following quantity related to the truncated partition function, $\mathcal{Z}^{(q)}_\beta/\mathcal{Z}^{(q)}_{W}=\operatorname{Tr}\qty{\Pi_{V,q}e^{-\beta W}S[E](\beta)}/\operatorname{Tr}\qty(\Pi_{V,q}e^{-\beta W})$ with $S[A]\coloneq \sum_{G\in \mathcal{M}(A)}\sum_{w\in \mathcal{S}(G)}F(w)$ being the Dyson series in single Hilbert space [cf.\,Eq.\,(\ref{dyson_single})] over the subset $A\subseteq E$. Then by denoting the function
\begin{equation}\label{g_A_complexity}
\begin{aligned}[b]
g(A)\coloneq  \frac{\operatorname{Tr}\qty{\Pi_{V,q}e^{-\beta W}S[A](\beta)}}{\operatorname{Tr}\qty(\Pi_{V,q}e^{-\beta W})},
\end{aligned}
\end{equation}
we move forward with 
\begin{equation}\label{boson_pari_zw}
\begin{aligned}[b]
\frac{\mathcal{Z}_{\beta}^{(q)}}{\mathcal{Z}_{W}^{(q)}}=g(E)
= \sum_{S \subseteq E}(-1)^{|S|} \sum_{T \subseteq S}(-1)^{|T|} g(T)
= \sum_{S \subseteq E}\prod_{\gamma \in \Gamma(S)}(-1)^{|\gamma|} \sum_{T \subseteq \gamma}(-1)^{|T|} g(T)
=\sum_{S \subseteq E} \prod_{\gamma \in \Gamma(S)} w_\gamma
=\sum_{\Gamma \in \mathcal{G}} \prod_{\gamma \in \Gamma} w_\gamma
\end{aligned}
\end{equation}
Here, for the second equal sign, we used Lemma \ref{lemma_inclusion_exclusion}. On the RHS of the third equal sign, we use  $\Gamma(S)$ to denote the set of maximally connected components of the multiset $S$. For example, if $S=\{Z_{1},Z_{2},...,Z_{8}\}$ as in Fig.\,\ref{w_cl_c_example}, then we have $\Gamma(S)=\{\{Z_{1},Z_{2},Z_{3}\},\{Z_{4},Z_{5}\},\{Z_{6}\},\{Z_{7},Z_{8}\}\}$. In Eq.\,\eqref{boson_pari_zw}, we have also defined the weight function as
\begin{equation}\label{boson_fr_weight}
\begin{aligned}[b]
w_{\gamma}\coloneq (-1)^{|\gamma|} \sum_{T \subseteq \gamma}(-1)^{|T|} g(T)
&=(-1)^{|\gamma|} \sum_{T \subseteq \gamma}(-1)^{|T|} \frac{\operatorname{Tr}_{V_{T}}\qty{\Pi_{V_{T},q}e^{-\beta W_{V_{T}}}S[T](\beta)}}{\operatorname{Tr}_{V_{T}}\qty(\Pi_{V_{T},q}e^{-\beta W_{V_{T}}})}
\\&=(-1)^{|\gamma|} \sum_{T \subseteq \gamma}(-1)^{|T|} \frac{\operatorname{Tr}_{V_{T}}\qty(\Pi_{V_{T},q}e^{-\beta H_{V_{T}}})}{\operatorname{Tr}_{V_{T}}\qty(\Pi_{V_{T},q}e^{-\beta W_{V_{T}}})},
\end{aligned}
\end{equation}
here we have used $e^{-\beta W_{L}}S[L](\beta)=e^{-\beta H_{L}}$ with $H_{L}\coloneq \sum_{Z\in E: Z\subseteq L}h_{Z}$ for any region $L\subseteq V$.
One may need the following result to justify the decomposition used in the third equal sign of Eq.\,\eqref{boson_pari_zw}:
\begin{corollary}\label{corollary_shuffle}
Let $g$ be the function defined in Eq.\,\eqref{g_A_complexity}, then for two compatible multisets, i.e., $\gamma^{\ri}\sim \gamma^{\rii}$, we have $g(\gamma^{\ri}\oplus \gamma^{\rii})=g(\gamma^{\ri})g(\gamma^{\rii})$.
\end{corollary}
\begin{proof}
By assumption, we have $V_{\gamma^{\ri}}\bigcap V_{\gamma^{\rii}}=\emptyset$. Then for any $G\in \mathcal{M}(\gamma^{\ri}\oplus \gamma^{\rii})$, we can always decompose it as $G=G^{\ri}\oplus G^{\rii}$ such that $G^{\ri}\in \mathcal{M}(\gamma^{\ri})$ and $G^{\rii}\in \mathcal{M}(\gamma^{\rii})$. Then by Lemma \ref{lemma_shuffle_F} with $f_{Z}(\tau)$ chosen to be $h_{Z}(\tau)$, we obtain
\begin{equation}\label{}
\begin{aligned}[b]
    S[\gamma^{\ri}\oplus \gamma^{\rii}](\beta)&=\sum_{G\in \mathcal{M}(\gamma^{\ri}\oplus \gamma^{\rii})}\sum_{w\in \mathcal{S}(G)}F(w)=\sum_{G^{\ri}\in \mathcal{M}(\gamma^{\ri})}\sum_{G^{\rii}\in \mathcal{M}(\gamma^{\rii})}\sum_{w\in \mathcal{S}(G^{\ri}\oplus G^{\rii})}F(w)
    \\&=\sum_{G^{\ri}\in \mathcal{M}(\gamma^{\ri})}\sum_{G^{\rii}\in \mathcal{M}(\gamma^{\rii})}\sum_{w^{\ri}\in \mathcal{S}(G^{\ri})}F(w^{\ri})\sum_{w^{\rii}\in \mathcal{S}( G^{\rii})}F(w^{\rii})=S[\gamma^{\ri}](\beta)S[\gamma^{\rii}](\beta).
\end{aligned}
\end{equation}
Then we note the following relation 
\begin{equation}\label{}
\begin{aligned}[b]
g(\gamma^{\ri}\oplus \gamma^{\rii})&=\frac{\operatorname{Tr}\qty{\Pi_{V,q}e^{-\beta W}S[\gamma^{\ri}\oplus \gamma^{\rii}](\beta)}}{\operatorname{Tr}\qty(\Pi_{V,q}e^{-\beta W})}=
\frac{\operatorname{Tr}_{V_{\gamma^{\ri}\oplus \gamma^{\rii}}}\qty{\Pi_{V_{\gamma^{\ri}\oplus \gamma^{\rii}},q}e^{-\beta W_{V_{\gamma^{\ri}\oplus \gamma^{\rii}}}}S[\gamma^{\ri}](\beta)S[\gamma^{\rii}](\beta)}}{\operatorname{Tr}_{V_{\gamma^{\ri}\oplus \gamma^{\rii}}}\qty(\Pi_{V_{\gamma^{\ri}\oplus \gamma^{\rii}},q}e^{-\beta W_{V_{\gamma^{\ri}\oplus \gamma^{\rii}}}})}
\\&=\frac{\operatorname{Tr}_{V_{\gamma^{\ri}}}\qty{\Pi_{V_{\gamma^{\ri}},q}e^{-\beta W_{V_{\gamma^{\ri}}}}S[\gamma^{\ri}](\beta)}}{\operatorname{Tr}_{V_{\gamma^{\ri} }}\qty(\Pi_{V_{\gamma^{\ri} },q}e^{-\beta W_{V_{\gamma^{\ri} }}})}\cdot \frac{\operatorname{Tr}_{V_{\gamma^{\rii}}}\qty{\Pi_{V_{\gamma^{\rii}},q}e^{-\beta W_{V_{\gamma^{\rii}}}}S[\gamma^{\rii}](\beta)}}{\operatorname{Tr}_{V_{\gamma^{\rii} }}\qty(\Pi_{V_{\gamma^{\rii} },q}e^{-\beta W_{V_{\gamma^{\rii} }}})}=g(\gamma^{\ri})g(\gamma^{\rii})
\end{aligned}
\end{equation}
to finish the proof.
\end{proof}
{~}

\hrulefill{\bf [ End of Proof of Corollary~\ref{corollary_shuffle}]}

{~}

By establishing Eqs.\,\eqref{boson_pari_zw} and \eqref{boson_fr_weight}, we have succefully reformulate the quantity $\mathcal{Z}^{(q)}_\beta/\mathcal{Z}^{(q)}_{W}$ in the form of abstract polymer partition function. To approximate the latter, we introduce the following truncated expansion \cite{mann2024algorithmic} for $\mathcal{Z}^{(q)}_\beta/\mathcal{Z}^{(q)}_{W}$:
\begin{equation}\label{truncat_T_m}
\begin{aligned}[b]
T_{m}\coloneq \sum_{a\in \mathscr{G}_{\mathrm{c}}: |a|\leq m} \varphi(\mathsf{G}_{a}) \prod_{\gamma\in a} w_{\gamma}=\sum_{\widetilde{a}\in \widetilde{\mathscr{G}}_{\mathrm{c}}: |\widetilde{a}|\leq m} \varphi(\mathsf{G}_{\widetilde{a}}) \prod_{(\gamma, \mu) \in \widetilde{a}} \frac{w_{\gamma}^{\mu}}{\mu!},
\end{aligned}
\end{equation}
with the length of sequences in $\mathscr{G}_{\cc}$ is defined as the summation of the cardinalities of all the multisets it contains, i.e., $|a|\coloneq \sum_{\gamma\in a}|\gamma|$ and $|\widetilde{a}|\coloneq \sum_{(\gamma,\mu)\in \widetilde{a}}\mu|\gamma|$.

Then the following proposition provides the estimation for this truncation formula:
\begin{proposition}\label{pro_truncation}
For the long-range Bose-Hubbard model ($\alpha > D$) defined by Eq.\,(\ref{slr_bose_hubbard}) over a finite lattice $V$, the expansion in the form of Eq.\,\eqref{log_abstract_polymer} for $\mathcal{Z}^{(q)}_\beta/\mathcal{Z}_{W}^{(q)}$ converges absolutely with the truncation error for $m\in \mathbb{Z}^{+}$:
\begin{equation}\label{}
\begin{aligned}[b]
\qty|\log\frac{\mathcal{Z}_{\beta}^{(q)}}{\mathcal{Z}_{W}^{(q)}}-T_{m}|\leq |V| e^{-m},
\end{aligned}
\end{equation}
for high temperatures $\beta\leq \beta_{c}$.
\end{proposition}
\begin{proof}
To proceed, we need to find an alternative expression for the weight function $w$. To that end, we first show
\begin{equation}\label{mobius_interacting}
\begin{aligned}[b]
\sum_{T \subseteq\gamma}(-1)^{|T|}S[T](\beta)=\sum_{T \subseteq \gamma}(-1)^{|T|}\sum_{w\in T^{\ast}}F(w) =(-1)^{|\gamma|}\sum_{w\in \gamma^{\ast}: \mathcal{D}(w)=\gamma}F(w).
\end{aligned}
\end{equation}
Here, we use the shorthand $\gamma^{\ast}\coloneq \bigcup_{G\in \mathcal{M}(\gamma)}\mathcal{S}(G)$ to denote the set of all the sequences with elements drawn from $\gamma$. The constraint $\mathcal{D}(w)=\gamma$ on the RHS of the second equal sign of Eq.\,\eqref{mobius_interacting} requires that the sequence $a$ contains every element of $\gamma$ at least once. To justify the second equal sign, we simply note that for any fixed word $w$ with underlying edge subset denoted by $S$, only those subsets $T\supseteq S$ contribute to the term $F(w)$. More precisely, they only contribute exactly once. Then we have the corresponding coefficients as $\sum_{T: \gamma\supseteq T\supseteq S}(-1)^{|T|}$ and by Lemma \ref{lemma_inclusion_exclusion_2} we obtain the desired result. 
Based on Eqs.\,\eqref{mobius_interacting} and \eqref{boson_fr_weight}, we arrive at a much more compact form of the weight formulated as
\begin{equation}\label{weight_ens_f_w}
\begin{aligned}
w_{\gamma}=\sum_{w\in \gamma^{\ast}: \mathcal{D}(w)=\gamma}\frac{\operatorname{Tr}\qty[\Pi_{V,q}e^{-\beta W}F(w)]}{\operatorname{Tr}\qty(\Pi_{V,q}e^{-\beta W})}=\sum_{w\in \gamma^{\ast}: \mathcal{D}(w)=\gamma}\frac{\operatorname{Tr}_{V_{\gamma}}\qty[e^{-\beta W_{V_{\gamma}}}F(w)\Pi_{V_{\gamma},q}]}{\operatorname{Tr}_{V_{\gamma}}\qty(e^{-\beta W_{V_{\gamma}}}\Pi_{V_{\gamma},q})}.
\end{aligned}
\end{equation}
To avoid confusion, we emphasize that the notations $w$ and $w_{\gamma}$ represent the sequence and weight, respectively.

Note that the magnitude of the weight function $w_{\gamma}$ can be bounded as 
\begin{equation}\label{}
\begin{aligned}[b]
|w_{\gamma}|&\leq \sum_{w\in \gamma^{\ast}: \mathcal{D}(w)=\gamma}\qty|\frac{\operatorname{Tr}_{V_{\gamma}}\qty[e^{-\beta W_{V_{\gamma}}}F(w)\Pi_{V_{\gamma},q}]}{\operatorname{Tr}_{V_{\gamma}}\qty(e^{-\beta W_{V_{\gamma}}}\Pi_{V_{\gamma},q})}|
\leq \sum_{w\in \gamma^{\ast}: \mathcal{D}(w)=\gamma}\frac{(C \sqrt{\beta})^{|w|}}{|w|!}\sum_{s_{1}s_{2}...s_{m}}\prod_{i=1}^{m}|J_{w_{i}}^{(s_{i})}|\cdot  \prod_{x\in V_{w}}[m_{x}(w,\vec{s})!]^{1/2}
\\&\leq \sum_{w\in \gamma^{\ast}: \mathcal{D}(w)=\gamma} \frac{(C\beta^{1/2})^{|w|}}{|w|!}\prod_{i=1}^{m}\qty(\sum_{s_{i}}|J_{w_{i}}^{(s_{i})}|)N(w_{i}|w)
\end{aligned}
\end{equation}
with $C=\mathcal{O}(1)$.
Here, for high temperature $\beta\leq \beta_{c}$ we have used Lemma \ref{lem_trace_f_w_q} and we use Lemma \ref{lemma_estimate_m_x} in the last line. For $k$-local Hamiltonian (which equals $2$ for Bose-Hubbard model), we set $\delta_{1}(\gamma)=k^{-1}|V_{\gamma}|$ and $\delta_{2}(\gamma)=|\gamma|$. Then for any site $x\in V$, we consider the following summation over all multisets $\gamma$ containing $x$ (denoted by $\gamma\nsim x$):
\allowdisplaybreaks[4]
\begin{align*}\label{kp_cri_long_boson_x}
\sum_{\gamma\nsim x}|w_{\gamma}|e^{\delta_{1}(\gamma)+\delta_{2}(\gamma)}&\leq \sum_{\gamma\nsim x}e^{\delta_{1}(\gamma)+\delta_{2}(\gamma)} \sum_{w\in \gamma^{\ast}: \mathcal{D}(w)=\gamma} \frac{(C\beta^{1/2})^{|w|}}{|w|!}\prod_{i=1}^{|w|}\qty(\sum_{s_{i}}|J_{w_{i}}^{(s_{i})}|)N(w_{i}|w)
\\&\leq \sum_{\gamma\nsim x} \sum_{w\in \gamma^{\ast}: \mathcal{D}(w)=\gamma} e^{2|w|}\frac{(C\beta^{1/2})^{|w|}}{|w|!}\prod_{i=1}^{|w|}\qty(\sum_{s_{i}}|J_{w_{i}}^{(s_{i})}|)N(w_{i}|w)
\\&= \sum_{m=1}^{\infty}\frac{(e^{2}C\beta^{1/2})^{m}}{m!}\sum_{G\in \mathcal{G}^{E}_{m}(x)}|\mathcal{P}(G)|\prod_{i=1}^{|w|}\qty(\sum_{s_{i}}|J_{w_{i}}^{(s_{i})}|)N(w_{i}|w)
\\&\leq \sum_{m=1}^{\infty}\frac{(e^{2}C\beta^{1/2})^{m}}{m!} c_{1}^{1/k}(c_{2}g\nu k)^{m}m!= c_{1}^{1/k} \frac{e^{2}c_{2}Cg\nu k\beta^{1/2}}{1-e^{2}c_{2}Cg\nu k\beta^{1/2}},
\refstepcounter{equation}\tag{\theequation}
\end{align*}
Here, for the second line, we used $|V_\gamma|\leq (k-1)|\gamma|+1\leq k|\gamma|$ and $|\gamma|\leq |w|$. In the third line, we use $\mathcal{G}^{E}_{m}(L) \coloneqq \{G \in \mathcal{M}(E) \colon |G| = m, G\oplus L \in \mathcal{C}\}$ to the define
the set of all multisets of size $m$ such that they constitute a connected multiset with the region $L$. To obtain the last line of Eq.\,\eqref{kp_cri_long_boson_x}, we used Proposition \ref{pro_cluster_sum} with $\gamma=g\nu, \nu \coloneq \sup_{i\in V}\sum_{j\in V}(1+d_{i,j})^{-\alpha}$ and $\phi(Z)=\sum_{s}|J_{Z}^{(s)}|$. To ensure the convergence of the summation, we also require $\beta \leq \beta_{c}< (e^{2}c_{2}Cg\nu k)^{2}$. We further impose 
\begin{equation}\label{extra_beta_c}
\begin{aligned}[b]
c_{1}^{1/k} \frac{e^{2}c_{2}Cg\nu k\beta^{1/2}}{1-e^{2}c_{2}Cg\nu k\beta^{1/2}} \leq k^{-1} \implies \beta \leq \qty[(e^{2}c_{2}Cg\nu k)^{-1}f_{0}^{-1}\qty(k^{-1}c_{1}^{-1/k})]^{2}=\mathcal{O}(1),
\end{aligned}
\end{equation}
with  $f_{0}(x)\coloneq x/(1-x)$. Based on Eq.\,\eqref{extra_beta_c}, for a fixed multiset $\gamma_{0}\in \mathfrak{C}$, we repeat the estimation for every site in $\gamma_{0}$ and obtain
\begin{equation}\label{kp_lr_boson}
\begin{aligned}[b]
\sum_{\gamma\nsim\gamma_{0}}|w_{\gamma}|e^{\delta_{1}(\gamma)+\delta_{2}(\gamma)}\leq \sum_{x\in \gamma_{0}}\sum_{\gamma\nsim x}|w_{\gamma}|e^{\delta_{1}(\gamma)+\delta_{2}(\gamma)}\leq c_{1}^{1/k} \frac{e^{2}c_{2}Cg\nu k\beta^{1/2}}{1-e^{2}c_{2}Cg\nu k\beta^{1/2}}|V_{\gamma_{0}}|\leq k^{-1}|V_{\gamma_{0}}|=\delta_{1}(\gamma_{0}).
\end{aligned}
\end{equation}
Equation~\eqref{kp_lr_boson} is exactly the precondition of Theorem \ref{thm:KP}. Then by setting $\gamma_{0}=\{x\}$ in Eq.\,\eqref{eq:bound_on_clusters} we arrive at
\begin{equation}\label{}
\sum_{\substack{a\in \mathscr{G}_{\mathrm{c}}: a\ni x}} \left|\varphi(\mathsf{G}_{a})\prod_{\gamma\in a}w_{\gamma}\right| e^{\sum_{\gamma\in a}\delta_{2}(\gamma)}\leq \delta_{1}(x)< 1.
\end{equation}
Note that $\sum_{\gamma\in a}\delta_{2}(\gamma)=\sum_{\gamma\in a}|\gamma|=|a|$ equals the length of sequence $a$, which leads us to

\begin{equation}\label{}
\begin{aligned}[b]
|V|=\sum_{x\in V}\sum_{\substack{a\in \mathscr{G}_{\mathrm{c}}}: a\ni x} \left|\varphi(\mathsf{G}_{a})\prod_{\gamma\in a}w_{\gamma}\right| e^{|a|}\geq \sum_{\substack{a\in \mathscr{G}_{\mathrm{c}}}} \left|\varphi(\mathsf{G}_{a})\prod_{\gamma\in a}w_{\gamma}\right| e^{|a|}=\sum_{m=1}^{\infty}\sum_{\substack{a\in \mathscr{G}_{\mathrm{c}}:|a|=m}} \left|\varphi(\mathsf{G}_{a})\prod_{\gamma\in a}w_{\gamma}\right| e^{m}.
\end{aligned}
\end{equation}
Therefore, the summation for the subsequence adheres to
\begin{equation}\label{kp_result_final}
\begin{aligned}[b]
\sum_{\substack{a\in \mathscr{G}_{\mathrm{c}}:|a|> m}} \left|\varphi(\mathsf{G}_{a})\prod_{\gamma\in a}w_{\gamma}\right|\leq |V|e^{-m}.
\end{aligned}
\end{equation}
Combining Eqs.\,\eqref{kp_result_final} and \eqref{truncat_T_m}, we complete the proof of the present theorem.
\end{proof}

{~}

\hrulefill{\bf [ End of Proof of Proposition~\ref{pro_truncation}]}

{~}

Then the following lemmas establish the estimation for the computational complexity of the truncated expansion $T_{m}$ for the quantity $\mathcal{Z}_{\beta}^{(q)}/\mathcal{Z}^{(q)}_{W}$.

\begin{lemma}[Lemma 7 of Ref.\,\cite{mann2024algorithmic}]\label{lemma_ursell}
The Ursell function $\varphi(\mathsf{G})$ can be computed in runtime $\exp (\mathcal{O}(|V(\mathsf{G})|))$.
\end{lemma}

\begin{lemma}\label{lemma_cluster_list}
Let $V$ be a finite lattice, and let $k, m \in \mathbb{N}$. Let polymer set $\mathfrak{C}$ be the set of all multisets with elements drawn from $E$ (vertex subsets of $V$ with cardinality at most $k$). Then, the multiset of polymers with size at most $m$ and a connected associated incompatibility graph can be enumerated in $|V|^{\mathcal{O}(km)}$ time. Here we have naturally assumed that $|V|\geq 2$ since it is usually a quite large integer.
\end{lemma}
\begin{proof}
From Lemma 6 in Ref.\,\cite{mann2024algorithmic}, if the interacting system is represented by a (hyper)graph with maximum degree $\mathfrak{d}$, then the runtime to finish the required task is at most (we denote $ N=|V|$)
\begin{equation}\label{complexity_cluster_m}
\begin{aligned}[b]
\mathfrak{d}^{2m}\times e^{\mathcal{O}(m)}\times N^{\mathcal{O}(1)}.
\end{aligned}
\end{equation}
For long-range systems, the maximum degree $\mathfrak{d}$ is not an $\mathcal{O}(1)$ constant and should be estimated.  We address this problem in the long-range $k$-local system. By definition, for a fixed site $i\in V$, the number of $\ell$-local edge including $i$ is at most $\binom{N-1}{\ell-1}$. Then we bound the degree at site $i$ as
\begin{equation}\label{}
\begin{aligned}[b]
\mathfrak{d}_{i}\leq \sum_{l=2}^{k}\binom{N-1}{\ell-1}\leq k N^{k-1},
\end{aligned}
\end{equation}
where the RHS can be used to bound the maximum degree for such a long-range system. Then clearly from Eq.\,\eqref{complexity_cluster_m} we understand now the complexity should be modified as $N^{\mathcal{O}(km)}$. Here we regard $e^{\mathcal{O}(m)}$ as $N^{\mathcal{O}(km)}$ since usually $N\geq 2$ is a quite large integer.
\end{proof}

{~}

\hrulefill{\bf [ End of Proof of Lemma~\ref{lemma_cluster_list}]}

{~}

\begin{lemma}\label{lemma_weight}
In the canonical long-range Bose-Hubbard model defined in Eq.\,\eqref{slr_bose_hubbard} with $\widetilde{J}_{i,j}=0$ on a finite lattice $V$, let the polymer $\gamma$ be the set of multisets with elements drawn from $E$ (vertex subsets of $V$ with cardinality at most $k$). Then the weight function $w_{\gamma}$ defined in Eq.\,\eqref{boson_fr_weight} with fixed truncated boson number $q$ can be computed in $\mathcal{O}\left( (4keq)^{3k|\gamma|}\right)$ time.
\end{lemma}
\begin{proof}
We repeat the expression of the weight function below again for the readers' convenience,
\begin{equation}\label{boson_fr_weight_re}
\begin{aligned}[b]
w_{\gamma}
=(-1)^{|\gamma|} \sum_{T \subseteq \gamma}(-1)^{|T|} \frac{\operatorname{Tr}_{V_{T}}\qty(\Pi_{V_{T},q}e^{-\beta H_{V_{T}}})}{\operatorname{Tr}_{V_{T}}\qty(\Pi_{V_{T},q}e^{-\beta W_{V_{T}}})}.
\end{aligned}
\end{equation}
First, we note that in the canonical Bose-Hubbard model, any sub-Hamiltonian $H_{L}\coloneq \sum_{Z\subseteq L}h_{Z}$ defined on $L\subseteq V$ preserves the total boson number on this region. This motivates us first to denote $\Pi_{i,m}\coloneq \ket{m}_{i}\prescript{}{i}{\bra{m}}$ with $\{\ket{m}_{i}\}_{m\in \mathbb{N}}$ being the eigenstates of $n_{i}$. Then we introduce the following projector
\begin{equation}\label{}
\begin{aligned}[b]
\widetilde{\Pi}_{L,q}\coloneq \sum_{\substack{m_{1},m_{2},...,m_{|L|}\in \mathbb{N}:\\m_{1}+m_{2}+...+m_{|L|}\leq qL}} \bigotimes_{i\in L}\Pi_{i,m_{i}}.
\end{aligned}
\end{equation}
This operator projects onto the subspace of states for which the total number of bosons in the region $L$ is at most $q|L|$. It is easy to check the relation $\Pi_{L,q}=\Pi_{L,q}\widetilde{\Pi}_{L,q}=\widetilde{\Pi}_{L,q}\Pi_{L,q}$ . More importantly, we have $[H_{L},\widetilde{\Pi}_{L,q}]=0$ since $H_{L}$ preserves the total boson number on the region $L$. If we denote $\id_{L}\coloneq \bigotimes_{i\in L}\id_{i}$ with $\id_{i}$ being the identity of $\mathcal{H}_{i}$, then we have $H_{L}=\widetilde{\Pi}_{L,q}H_{L}\widetilde{\Pi}_{L,q}+(\id_{L}- \widetilde{\Pi}_{L,q})H_{L}(\id_{L}- \widetilde{\Pi}_{L,q})$. This relation implies that $e^{-\beta H_{L}}= \exp(-\beta\widetilde{\Pi}_{L,q}H_{L}\widetilde{\Pi}_{L,q}) \exp[-\beta (\id_{L}- \widetilde{\Pi}_{L,q})H_{L}(\id_{L}- \widetilde{\Pi}_{L,q})]$. Then note that
\begin{equation}\label{Pi_exp_H_Pi}
\begin{aligned}[b]
\widetilde{\Pi}_{L,q}e^{-\beta H_{L}}\widetilde{\Pi}_{L,q}&= \widetilde{\Pi}_{L,q}\exp(-\beta\widetilde{\Pi}_{L,q}H_{L}\widetilde{\Pi}_{L,q}) \exp[-\beta (\id_{L}- \widetilde{\Pi}_{L,q})H_{L}(\id_{L}- \widetilde{\Pi}_{L,q})]\widetilde{\Pi}_{L,q}
\\&=\exp(-\beta\widetilde{\Pi}_{L,q}H_{L}\widetilde{\Pi}_{L,q})\widetilde{\Pi}_{L,q} \exp[-\beta (\id_{L}- \widetilde{\Pi}_{L,q})H_{L}(\id_{L}- \widetilde{\Pi}_{L,q})]\widetilde{\Pi}_{L,q}
=\exp(-\beta\widetilde{\Pi}_{L,q}H_{L}\widetilde{\Pi}_{L,q})\widetilde{\Pi}_{L,q}.
\end{aligned}
\end{equation}
Here, we have used the fact that $[\widetilde{\Pi}_{L,q},O]=0$ implies $[\widetilde{\Pi}_{L,q},e^{O}]=0$, which can be rigorously proven via spectral theorem \cite{reed1972methods} even for unbound operator $O$ as long as it is also self-adjoint (Hermitian). From Eq.\,\eqref{Pi_exp_H_Pi} we further obtain
\begin{equation}\label{tr_L_pi_L_q}
\begin{aligned}[b]
\operatorname{Tr}_{L}\qty(\Pi_{L,q}e^{-\beta H_{L}})&=\operatorname{Tr}_{L}\qty(\Pi_{L,q}\widetilde{\Pi}_{L,q}e^{-\beta H_{L}}\widetilde{\Pi}_{L,q})=\operatorname{Tr}_{L}\qty[\Pi_{L,q}\exp(-\beta\widetilde{\Pi}_{L,q}H_{L}\widetilde{\Pi}_{L,q})\widetilde{\Pi}_{L,q}]
\\&=\operatorname{Tr}_{L}\qty[\Pi_{L,q}\widetilde{\Pi}_{L,q}\exp(-\beta\widetilde{\Pi}_{L,q}H_{L}\widetilde{\Pi}_{L,q})\widetilde{\Pi}_{L,q}]=\operatorname{Tr}_{L}\qty[\widetilde{\Pi}_{L,q}\exp(-\beta\widetilde{\Pi}_{L,q}H_{L}\widetilde{\Pi}_{L,q})\widetilde{\Pi}_{L,q}].
\end{aligned}
\end{equation}
Clearly, the dimension of the range of the projector  $\widetilde{\Pi}_{L,q}$ equals the number of non-negative integer solutions to the inequality $\sum_{i=1}^{|L|}m_{i}\leq q|L|$, which is exactly given by $\binom{(q+1)|L|}{|L|}$. Also, the corresponding dimension of the range of $\Pi_{L,q}$ is $\qty[(q+1)|L|]^{|L|}$ and greater than  of $\widetilde{\Pi}_{L,q}$, therefore we can get rid of $\Pi_{L,q}$ in the last line of Eq.\,\eqref{tr_L_pi_L_q}. Note that for finite dimensional matrix $A$, the computational complexity of $\operatorname{Tr}e^{A}$ is of the order $\mathcal{O}((\operatorname{dim}A)^{3})$. Therefore
The computational complexity of $\operatorname{Tr}_{L}\qty(\Pi_{L,q}e^{-\beta H_{L}})$ can be roughly upper bounded by $\mathcal{O}((2eq)^{3|L|})$, since $\binom{(q+1)|L|}{|L|}\leq [e(q+1)]^{|L|}\leq (2eq)^{|L|}$. Here, we have used the elementary bound $\binom{n}{m} \leq (ne/m)^m$. 

It follows directly that the complexity of computing the ratio $\operatorname{Tr}_{V_{T}}\qty(\Pi_{V_{T},q}e^{-\beta H_{V_{T}}})/\operatorname{Tr}_{V_{T}}\qty(\Pi_{V_{T},q}e^{-\beta W_{V_{T}}})$ is thus upper-bounded by $\mathcal{O}\left( (2eq)^{3|V_{T}|}\right)$.
Given that $|V_{T}|\leq |V_{\gamma}|\leq k|\gamma|$ and that the expression for $w_{\gamma}$ in Eq.\,\eqref{boson_fr_weight_re} is a sum over $2^{|\gamma|}$ such terms, the total computational complexity to determine $w_{\gamma}$ is of the order $\mathcal{O}\left(2^{|\gamma|} (2keq)^{6|\gamma|}\right)$. Therefore, we complete the proof.
\end{proof}

{~}

\hrulefill{\bf [ End of Proof of Lemma~\ref{lemma_weight}]}

{~}

\begin{remark}
A key aspect of our analysis is that the boson number truncation $q = \mathcal{O}(\log |V|)$ for high temperature $\beta\leq \beta_{c}$, as established in Corollary \ref{corollary_truncation_error}. This fundamentally distinguishes our result from Lemma 19 of Ref.\,\cite{mann2024algorithmic}. 
\end{remark}

By using the lemmas above, we are able to estimate the computational complexity for $T_{m}$, as stated in the following result,
\begin{proposition}\label{pro_T_m}
In the canonical long-range Bose-Hubbard model defined on a finite lattice $V$, let the polymer $\gamma$ be the set of multisets of sites with cardinality at most a fixed integer $k$. Define the corresponding weight function $w_{\gamma}$ as in Eq.\,\eqref{boson_fr_weight} with fixed truncated boson number $q$. Then the truncated expansion $T_{m}$ as defined in Eq.\,\eqref{truncat_T_m} can be computed in time $|V|^{\mathcal{O}(km)}\times \mathcal{O}((4keq)^{3km}) \times e^{\mathcal{O}(m)}$.
\end{proposition}
\begin{proof}
From Lemma \ref{lemma_cluster_list}, we first list all the elements in $\{\widetilde{a}\in \widetilde{\mathscr{G}}_{\mathrm{c}}: |\widetilde{a}|\leq m\}$ in time $|V|^{\mathcal{O}(km)}$. For each multiset of polymers, the Ursell function can be computed within time $\exp(\mathcal{O}(m))$ from Lemma \ref{lemma_ursell}. From Eq.\,\eqref{truncat_T_m} we know that for each sequence there are at most $m$ terms and each weight can be computed at most in time $\mathcal{O}\left( (4keq)^{3km}\right)$ by Lemma \ref{lemma_weight}. 
\end{proof}

{~}

\hrulefill{\bf [ End of Proof of Proposition~\ref{pro_T_m}]}

{~}

\begin{theorem}\label{theorem_s_lr_complexity}
For the canonical long-range Bose-Hubbard model defined by Eq.\,\eqref{slr_bose_hubbard} with $\widetilde{J}_{i,j}=0$ with decay exponent $\alpha > D$ on a finite lattice $V$. Denote $N=|V|$, then there exists a classical algorithm that, for any given precision $\epsilon>0$ and index $\theta>0$, computes an approximation $f_{\beta}$ satisfying
$|\log \mathcal{Z}_{\beta}-f_{\beta}|\leq N^{-\theta}+\epsilon.$
at any fixed temperature above a threshold, $\beta\leq \beta_{c}=\mathcal{O}(1)$. The runtime of the algorithm is quasi-polynomial in the system size, given by $e^{\mathcal{O}(\log^{2} (N/\epsilon))}.$
\end{theorem}
\begin{proof}
First, as detailed in Corollary \ref{corollary_truncation_error}, we truncate the on-site boson number up to $q\propto \log N$ (the constant here depends on $\beta$ and $\theta$) to obtain the truncated partition function $\mathcal{Z}_{\beta}^{(q)}$.  The estimation for the truncation error is given by $|\log \mathcal{Z}_{\beta}-\log \mathcal{Z}^{(q)}_{\beta}|\leq N^{-\theta}$. Then we utilize the abstract cluster expansion to obtain $T_{m}$ and the error is estimated in Proposition Eq.\,\eqref{pro_truncation}, in the form of $|\log \mathcal{Z}_{\beta}^{(q)}-\log \mathcal{Z}_{W}^{(q)}-T_{m}|\leq N e^{-m}$. Here, the quantity $\mathcal{Z}_{W}^{(q)}=\prod_{x\in V}\sum_{n=0}^{q}e^{-\beta W_{x}(n)}$ denotes the truncated partition function with all the hopping terms turned off. By combining these two steps, we arrive at
\begin{equation}\label{}
\begin{aligned}[b]
|\log \mathcal{Z}_{\beta}-\log \mathcal{Z}_{W}^{(q)}-T_{m}|\leq N^{-\theta}+N e^{-m}.
\end{aligned}
\end{equation}
Note that the computational time for $\log \mathcal{Z}_{W}^{(q)}$ is simply $\mathcal{O}(qN)=\mathcal{O}(N\log N)$, which together with Proposition \ref{pro_T_m} implies that the computational time for $\log \mathcal{Z}_{W}^{(q)}+T_{m}$ is given by 
\begin{equation}\label{}
\begin{aligned}[b]
\mathcal{O}(N\log N)+ N^{\mathcal{O}(m)}\times \mathcal{O}((\log N)^{3km}) \times e^{\mathcal{O}(m)}.
\end{aligned}
\end{equation}
By setting $m=\mathcal{O}(\log(N/\epsilon))$, the total runtime here is the order of 
\begin{equation}\label{}
\begin{aligned}[b]
\mathcal{O}(N\log N)+ N^{\mathcal{O}(\log(N/\epsilon))}\times \mathcal{O}((\log N)^{\mathcal{O}(\log(N/\epsilon))}) \times e^{\mathcal{O}(\log(N/\epsilon))},
\end{aligned}
\end{equation}
which simplifies to $e^{\mathcal{O}(\log^{2} (N/\epsilon))}$ and thus completes the proof. Note that the constants omitted in the final result also depend on $\beta,\beta_{c},\theta$, and parameters in the model.
\end{proof}

{~}

\hrulefill{\bf [ End of Proof of Theorem~\ref{theorem_s_lr_complexity}]}

{~}

\begin{remark}
The implicit constants in the $\mathcal{O}$-notation for the runtime depend on the chosen temperature $\beta$, threshold temperature $\beta_{c}$, the flexible index $\theta$ and model parameters ($\alpha, D, g,\mu,U_{\min}$ and $U_{\max}$).
\end{remark}

\section{Useful Propositions}

\begin{proposition}\label{pro_subsystem_partition}
In the long-range Bose-Hubbard model ($\alpha>D$) 
defined in Eq.\,\eqref{slr_bose_hubbard} over a finite lattice $V$, the following inequality holds for any region $L\subset V$:
\begin{equation}\label{subsystem_cor}
\begin{aligned}[b]
\frac{\operatorname{Tr}_{L^{\cc}}e^{-\beta H_{L^{\cc}}}}{\operatorname{Tr}e^{-\beta H}}\leq \frac{1}{\operatorname{Tr}_{L}e^{-\beta W_{L}}}.
\end{aligned}
\end{equation}
Here, $H_{L^{\cc}}\coloneq \sum_{Z\in E: Z\subseteq L^{\cc}}h_{Z}$ captures all the interactions not on a fix region $L\subseteq V$ and $W_{L}\coloneq \sum_{i\in L}W_{i}$ collects all the on-site interaction terms in $L$.
\end{proposition}
\begin{proof}
We introduced a free energy functional 
\begin{equation}\label{}
\begin{aligned}[b]
F(\bullet)\coloneq \operatorname{Tr}(H\bullet)-\beta^{-1}S(\bullet)
\end{aligned}
\end{equation}
with  $\rho_{1}\coloneq \rho_{H_{L^{\cc}}}\otimes\rho_{W_{L}}\coloneq   e^{-\beta H_{L^{\cc}}}/\operatorname{Tr}_{L^{\cc}}e^{-\beta H_{L^{\cc}}} \otimes e^{-\beta W_{L}}/\operatorname{Tr}_{L}e^{-\beta W_{L}}$ and $\rho_{H}\coloneq e^{-\beta H}/\operatorname{Tr}e^{-\beta H}$.
From the Gibbs variational principle \cite{lemm2023thermal,alhambra2023quantum} we have
\begin{equation}\label{free_energy}
\begin{aligned}[b]
F(\rho_{1})\geq F(\rho_{H}).
\end{aligned}
\end{equation} 
To proceed, we first note that 
\begin{equation}\label{tr_H_rho_1}
\begin{aligned}[b]
\operatorname{Tr}\qty(H\rho_{1})
=&\operatorname{Tr}\qty[\qty(H_{L^{\cc}}+H_{L}+H_{\partial L})\rho_{1}]
\\=&\operatorname{Tr}_{L}\qty(H_{L}\rho_{W_{L}})+\operatorname{Tr}_{L^{\cc}}\qty(H_{L^{\cc}}\rho_{H_{L^{\cc}}})+\operatorname{Tr}\qty(H_{\partial L}\rho_{1})
=E(\rho_{H_{L^{\cc}}})+\operatorname{Tr}_{L}(H_{L}\rho_{W_{L}})+\operatorname{Tr}\qty(H_{\partial L}\rho_{1})
\end{aligned}
\end{equation}
Here, the operators $H_{L}\coloneq \sum_{Z\in E: Z\subseteq L}h_{Z}$ and $H_{\partial L}\coloneq \sum_{Z\in E: Z\cap L,L^{\cc}\neq \emptyset}h_{Z}$ captures the interactions on $L$ and boundary of $L$, respectively. We also used $E(\rho_{\bullet})\coloneq \operatorname{Tr}\bullet \rho_{\bullet}$ to present the internal energy for a given thermal Gibbs state $\rho_{\bullet}$.
On the other hand, we have
\begin{equation}\label{S_rho_1}
\begin{aligned}[b]
\beta^{-1}S(\rho_{1})
=&\beta^{-1}S(\rho_{W_{L}})+\beta^{-1}S(\rho_{H_{L^{\cc}}})
\\=&E(\rho_{W_{L}})+E(\rho_{H_{L^{\cc}}})+\beta^{-1}\ln \operatorname{Tr}_{L}e^{-\beta W_{L}}
+\beta^{-1}\ln \operatorname{Tr}_{L^{\cc}}e^{-\beta H_{L^{\cc}}}.
\end{aligned}
\end{equation}
By substituting \eqtwo{tr_H_rho_1}{S_rho_1} into \eq{free_energy} we obtain
\begin{equation}\label{free_energy_results}
\begin{aligned}[b]
&\operatorname{Tr}(H\rho_{1})-\beta^{-1}S(\rho_{1})\geq -\beta^{-1}\ln \operatorname{Tr}e^{-\beta H}
\\\implies& \operatorname{Tr}_{L}(H_{L}\rho_{W_{L}})+\operatorname{Tr}\qty(H_{\partial L}\rho_{1})-\langle W_{L} \rangle_{\beta W_{L}}-\beta^{-1}\ln \operatorname{Tr}_{L}e^{-\beta W_{L}}
-\beta^{-1}\ln \operatorname{Tr}_{L^{\cc}}e^{-\beta H_{L^{\cc}}}\geq -\beta^{-1}\ln \operatorname{Tr}e^{-\beta H} 
\\\implies& \ln\frac{\operatorname{Tr}_{L}e^{-\beta W_{L}}\operatorname{Tr}_{L^{\cc}}e^{-\beta H_{L^{\cc}}}}{\operatorname{Tr}e^{-\beta H}}\leq \beta \qty[\operatorname{Tr}_{L}(H'_{L}\rho_{W_{L}})+\operatorname{Tr}\qty(H_{\partial L}\rho_{1})].
\end{aligned}
\end{equation}
Here, we used $H_{L}'$ to denote all the terms except for the on-site interactions in the region $L$, i.e., the hopping terms. Then we note that
\begin{equation}\label{tr_h_prime}
\begin{aligned}[b]
\operatorname{Tr}_{L}(H'_{L}\rho_{W_{L}})=\operatorname{Tr}\qty(H_{\partial L}\rho_{1})=0,
\end{aligned}
\end{equation}
since the hopping term does not preserve local particle number at any site. By putting Eq.\,\eqref{tr_h_prime} into Eq.\,\eqref{free_energy_results} we clearly finish the proof.
\end{proof}

{~}

\hrulefill{\bf [ End of Proof of Lemma~\ref{pro_subsystem_partition}]}

{~}

\begin{lemma}\label{lem_trace_f_w}
Let $w$ be a sequence, denote $V_{w}=\bigcup_{Z\in w}Z$ and $m=|w|$ as the support and length of $w$, respectively. Define
\begin{equation}\label{}
\begin{aligned}[b]
F(w)\coloneq  \ii{0}{\beta}{\tau_{1}}\ii{0}{\tau_{1}}{\tau_{2}}...\ii{0}{\tau_{m-1}}{\tau_{m}} h_{w_{1}}(\tau_{1})h_{w_{2}}(\tau_{2})...h_{w_{m}}(\tau_{m}), \quad \bullet(\tau)\coloneq e^{\tau W_{V_{w}}}\bullet e^{-\tau W_{V_{w}}}.
\end{aligned}
\end{equation} 
If $\{h_{Z}\}_{Z\in w}$ comes from the hopping and squeezing terms in the long-range Bose-Hubbard model [cf.\,Eq.\,(\ref{slr_bose_hubbard})] and write it as $h_{Z}=\sum_{s}J_{Z}^{(s)}h^{(s)}_{Z}$, then for any site $\overline{x}\in V_{w}$ the following inequality holds with $c_{1},c_{2},c_{3}=\mathcal{O}(1)$ independent from $\overline{x}$ for high temperatures $\beta\leq \beta_{c}$:
\begin{equation}\label{}
\begin{aligned}[b]
\norm{n_{\overline{x}}^{l}e^{-\beta W_{V_{w}}}F(w)}_{1}\leq \sqrt{(2l)!} \qty(\frac{c_{1}}{\sqrt{\beta}})^{l}\qty(\frac{c_{2}}{\sqrt{\beta}})^{|V_{w}|}\frac{(c_{3} \sqrt{\beta})^{m}}{m!}\sum_{s_{1}s_{2}...s_{m}}\prod_{i=1}^{m}|J_{w_{i}}^{(s_{i})}|\cdot  \prod_{x\in V_{w}}[m_{x}(w,\vec{s})!]^{1/2}.
\end{aligned}
\end{equation} 
Here, we introduced the vector $\vec{s}\coloneq (s_{1},s_{2},...,s_{m})$ and denote $\mathsf{N}_{x}(h^{(s_{1})}_{w_{1}}h^{(s_{2})}_{w_{2}}...h^{(s_{m})}_{w_{m}})$ as $m_{x}(w,\vec{s})$. 
\end{lemma}
\begin{remark}
In fact, $m_{x}(w,\vec{s})$ does not depend on $\vec{s}$ in the long-range Bose-Hubbard model [cf.\,Eq.\,(\ref{slr_bose_hubbard})], since each term offers exactly one operator in its support.
\end{remark}
\begin{proof}
By denoting $l_{x}=l$ for $x=\overline{x}$ and $l_{x}=0$ otherwise, we first notice that
\begin{equation}\label{nbarx_l}
\begin{aligned}[b]
n_{\overline{x}}^{l}e^{-\beta W_{V_{w}}}F(w)=\int [\tau_{1}\tau_{2}...\tau_{m}] \sum_{s_{1},s_{2},...,s_{m}}|J_{w_{1}}^{(s_{1})}||J_{w_{2}}^{(s_{2})}|...|J_{w_{m}}^{(s_{m})}|\cdot \prod_{x\in V_{w}} \qty[n_{x}^{l_{x}}e^{-\beta W_{x}}h_{w_{1},x}^{(s_{1})}(\tau_{1})h_{w_{2},x}^{(s_{2})}(\tau_{2})...h_{w_{m},x}^{(s_{m})}(\tau_{m})].
\end{aligned}
\end{equation}
Here, for fixed sequence $w$ and vector $\vec{s}\coloneq (s_{1},s_{2},...,s_{m})$, we have decomposed $h_{w_{i}}^{(s_{i})}(\tau_{i})=\prod_{x\in w_{i}}h_{i,x}(\tau_{i})$ for each $i=1,2,...,m$ into different sites. For each operator $h_{i,x}(\tau_{i})$, where the notations for $w$ and $\vec{s}$ will be omitted for simplification when no confusion arises, we denote the ``net'' number of bosonic operators within (i.e., \# of $a_{x}^{\dagger}$ minus \# of $a_{x}$) as $b_{i,x}$, then we have $h_{i,x}(\tau_{i})=e^{\tau_{i}W_{x}(n_{x})}h_{i,x}(\tau_{i})e^{-\tau_{i}W_{x}(n_{x})}=h_{i,x}e^{\tau_{i}W_{x}(n_{x}+b_{i,x})-\tau_{i}W_{x}(n_{x})}$. Therefore we note that 
\begin{equation}\label{e_beta_W_x_h}
\begin{aligned}[b]
e^{-\beta W_{x}}h_{w_{1},x}^{(s_{1})}(\tau_{1})h_{w_{2},x}^{(s_{2})}(\tau_{2})...h_{w_{m},x}^{(s_{m})}(\tau_{m})=h_{1,x}h_{2,x}...h_{m,x}e^{O_{x}(n_{x})} \eqcolon h_{x}(w,\vec{s})e^{O_{x}(n_{x})}.
\end{aligned}
\end{equation}
If we denote $S_{i,x}\coloneq \sum_{j=i}^{m}b_{j,x}$, then we can exciplitly write
\begin{equation}\label{O_x_n_x}
\begin{aligned}[b]
O_{x}(n_{x})&=-\beta W_{x}(n_{x}+S_{1,x})+\tau_{1}W_{x}(n_{x}+S_{1,x})-\tau_{1}W_{x}(n_{x}+S_{2,x})+\tau_{2}W_{x}(n_{x}+S_{2,x})-\tau_{2}W_{x}(n_{x}+S_{3,x})
\\&\quad +\tau_{3}W_{x}(n_{x}+S_{2,x})+...+\tau_{m}W_{x}(n_{x}+S_{m,x})-\tau_{m}W_{x}(n_{x})
\\&=-(\beta-\tau_{1})W_{x}(n_{x}+S_{1,x})-(\tau_{1}-\tau_{2})W_{x}(n_{x}+S_{2,x})-(\tau_{2}-\tau_{3})W_{x}(n_{x}+S_{3,x})
\\&\quad -...-(\tau_{m-1}-\tau_{m})W_{x}(n_{x}+S_{m})-\tau_{m}W_{x}(n_{x}).
\end{aligned}
\end{equation}
The quantity $b_{i,x}$ and $O_{x}(n_{x})$ also depend on $w$ and $\vec{s}$, but we omit the notations for clarity.
With these preparations, we proceed with the similar treatment in Appendix E\,2 of Ref.\,\cite{tong2024boson} to obtain for high temperature $\beta \leq \mathcal{O}(1)$ that 
\begin{equation}\label{norm_n_l_hws}
\begin{aligned}[b]
\norm{n_{x}^{l_{x}}h_{x}(w,\vec{s})e^{O_{x}(n_{x})}}_{1}\leq \sqrt{\frac{2C_{1}C_{2}}{\beta}}\qty(\frac{2C_{1}}{\sqrt{\beta}})^{l_{x}}\qty(\frac{2eC_{1}}{\sqrt{\beta}})^{m_{x}/2}\sqrt{m_{x}!}\sqrt{(2l_{x})!},
\end{aligned}
\end{equation}
with $C_{1},C_{2}=\mathcal{O}(1)$ depending only on the parameter of models (i.e., $U_{\min},U_{\max},\mu$) and the threshold temperature.
Here, we have also denoted the number of totol bosonic operators at site $x$ offer by a fixed sequence $w$ and vector $\vec{s}$ as $m_{x}(w,\vec{s})\coloneq \mathsf{N}(h_{x}(w,\vec{s}))=\mathsf{N}_{x}(h^{(s_{1})}_{w_{1}}h^{(s_{2})}_{w_{2}}...h^{(s_{m})}_{w_{m}})$ and shorted it as $m_{x}$ whenever no confusion arises. Putting Eqs.\,\eqref{e_beta_W_x_h} and \eqref{norm_n_l_hws}
into Eq.\,\eqref{nbarx_l}, we obtain
\begin{equation}\label{trace_norm_n_xbar_l}
\begin{aligned}[b]
&\norm{n_{\overline{x}}^{l}e^{-\beta W_{V_{w}}}F(w)}_{1}
\\\leq& \int [\tau_{1}\tau_{2}...\tau_{m}] \sum_{s_{1},s_{2},...,s_{m}}|J_{w_{1}}^{(s_{1})}||J_{w_{2}}^{(s_{2})}|...|J_{w_{m}}^{(s_{m})}| \cdot \prod_{x\in V_{w}}\norm{n_{x}^{l_{x}}h_{x}(w,\vec{s})e^{O_{x}(n_{x})}}_{1}
\\\leq& \int [\tau_{1}\tau_{2}...\tau_{m}] \sum_{s_{1},s_{2},...,s_{m}}|J_{w_{1}}^{(s_{1})}||J_{w_{2}}^{(s_{2})}|...|J_{w_{m}}^{(s_{m})}| \cdot \prod_{x\in V_{w}}\sqrt{\frac{2C_{1}C_{2}}{\beta}}\qty(\frac{2C_{1}}{\sqrt{\beta}})^{l_{x}}\qty(\frac{2eC_{1}}{\sqrt{\beta}})^{m_{x}(w,\vec{s})/2}\sqrt{m_{x}(w,\vec{s})!}\sqrt{(2l_{x})!}
\\\leq & \sqrt{(2l)!} \frac{\beta^{m}}{m!}\sum_{s_{1},s_{2},...,s_{m}}|J_{w_{1}}^{(s_{1})}||J_{w_{2}}^{(s_{2})}|...|J_{w_{m}}^{(s_{m})}| \qty(\sqrt{\frac{2C_{1}C_{2}}{\beta}})^{|V_{w}|}\qty(\frac{2C_{1}}{\sqrt{\beta}})^{l}\qty(\frac{2eC_{1}}{\sqrt{\beta}})^{m}\prod_{x\in V_{w}}\sqrt{m_{x}(w,\vec{s})!}.
\end{aligned}
\end{equation}
Here, to obtain the last line, we evaluate the integral and use the upper bound $l_{x}\leq l$ for sufficiently high temperatures. We have also used $\sum_{x\in V_{w}}m_{x}(w,\vec{s})=2|w|=2m$ for the Bose-Hubbard model, where every off-site term is $2$-local and offers exactly one operator to its support.

\end{proof}

{~}

\hrulefill{\bf [ End of Proof of Lemma~\ref{lem_trace_f_w}]}

{~}

\begin{remark}
For free bosons, $W_{x}(n_{x})$ is a linear function but Lemma~8 in Ref.\,\cite{tong2024boson}  is still applicable if we revise it to
\begin{equation}\label{revise_lemma_8}
\begin{aligned}[b]
\sum_{n=0}^{\infty}(n+p)^{p} 
\exp[\max_{m\in \{n-p,n+p\}}\qty(-a m)]&= \sum_{n=0}^{\infty}(n+p)^{p}e^{-a(n-p)} \leq e^{ap} \sum_{n=0}^{\infty}(n+2p)^{p}e^{-an} \leq e^{ap} c^{p}a^{-(p+1)}p^{p+1}
\\&\leq c^{p}a^{-(p+1)}p!,
\end{aligned}
\end{equation}
where we used Lemma~7 in Ref.\,\cite{tong2024boson}. 
Then Eq.\,\eqref{revise_lemma_8} gives us the alternative version of the present lemma for the long-range free bosons in the form of:
\begin{equation}\label{}
\begin{aligned}[b]
\norm{n_{x}^{l}e^{-\beta W_{V_{w}}}F(w)}_{1}\leq \sqrt{(2l)!} \qty(\frac{c_{1}}{\beta})^{l}\qty(\frac{c_{2}}{\beta})^{|V_{w}|}\frac{c_{3}^{m}}{m!}\sum_{s_{1}s_{2}...s_{m}}\prod_{i=1}^{m}|J_{w_{i}}^{(s_{i})}|\cdot  \prod_{x\in V_{w}}[m_{x}(w,\vec{s})!]^{1/2}.
\end{aligned}
\end{equation} 
\end{remark}

\begin{lemma}\label{lemma_trace_Ftil_w}
Let $w$ be a sequence, denote $V_{w}=\bigcup_{Z\in w}Z$ and $m=|w|$ as the support and length of $w$, respectively. Let $\widetilde{V}_{w} \supseteq V_{w}$ be some extension of $V_{w}$, define
\begin{equation}\label{}
\begin{aligned}[b]
\widetilde{F}(w)\coloneq  \ii{0}{\beta}{\tau_{1}}\ii{0}{\tau_{1}}{\tau_{2}}...\ii{0}{\tau_{m-1}}{\tau_{m}} h_{w_{1}}(\tau_{1})^{(+)}h_{w_{2}}(\tau_{2})^{(+)}...h_{w_{m}}(\tau_{m})^{(+)}, \quad \bullet(\tau)\coloneq e^{\tau W_{\widetilde{V}_{w}}}\bullet e^{-\tau W_{\widetilde{V}_{w}}}.
\end{aligned}
\end{equation} 
If $\{h_{Z}\}_{Z\in w}$ comes from the hopping terms in the long-range Bose-Hubbard model [cf.\,Eq.\,(\ref{slr_bose_hubbard})] and write it as $h_{Z}=\sum_{s}J_{Z}^{(s)}h^{(s)}_{Z}$, then for any two regions $X,Y\subset \widetilde{V}_{w}$ with $O_{X}$ and $O_{Y}$ being products of $\{a_{i}^{\dagger},a_{i}\}_{i\in X}$ and $\{a_{i}^{\dagger},a_{i}\}_{i\in Y}$ respectively, the following inequality holds with $c_{1},c_{2},c_{3}=\mathcal{O}(1)$ for high temperatures $\beta\leq \beta_{c}$:
\begin{equation}\label{}
\begin{aligned}[b]
&\norm{e^{-\beta W^{(+)}_{\widetilde{V}_{w}}}\widetilde{F}(w)O_{X}^{(0)}O_{Y}^{(1)}}_{1}
\leq \Phi(\beta) \qty(\frac{c_{1} }{\sqrt{\beta}})^{2|\widetilde{V}_{w}|}c_{2}^{[\mathsf{N}(O_{X})+\mathsf{N}(O_{Y})]/2}\frac{\qty(c_{3}\cdot \sqrt{\beta})^{m}}{m!}\sum_{s_{1}s_{2}...s_{m}}\prod_{l=1}^{m}|J_{w_{l}}^{(s_{l})}|\cdot  \prod_{x\in V_{w}}[m_{x}(w,\vec{s})!]^{1/2}
\end{aligned}
\end{equation}
Here, we introduced the vector $\vec{s}\coloneq (s_{1},s_{2},...,s_{m})$ and denote $\mathsf{N}_{x}(h^{(s_{1})}_{w_{1}}h^{(s_{2})}_{w_{2}}...h^{(s_{m})}_{w_{m}})$ as $m_{x}(w,\vec{s})$. We have also defined the function $\Phi(\beta)\coloneq [\mathsf{N}(O_{X})!]^{1/2} [\mathsf{N}(O_{Y})!]^{1/2}\beta^{-[\mathsf{N}(O_{X})+\mathsf{N}(O_{Y})]/4}$.
\end{lemma}
\begin{proof}
Then our next step is to investigate the structure of sequences on the right-hand side of Eq.\,\eqref{cor_x_y}: 
\begin{equation}\label{tr_vtil_g_tilF}
\begin{aligned}[b]
&\norm{e^{-\beta W^{(+)}_{\widetilde{V}_{w}}}\widetilde{F}(w)O_{X}^{(0)}O_{Y}^{(1)}}_{1}
\\=&\int[\tau_{1}\tau_{2}...\tau_{m}]\sum_{s_{1}s_{2}...s_{m}}|J_{w_{1}}^{(s_{1})}||J_{w_{2}}^{(s_{2})}|...|J_{w_{m}}^{(s_{m})}|\norm{e^{-\beta W_{\widetilde{V}_{w}}^{(+)}}h_{w_{1}}(\tau_{1})^{(+)}h_{w_{2}}(\tau_{2})^{(+)}...h_{w_{m}}(\tau_{m})^{(+)}O_{X}^{(0)}O_{Y}^{(1)}}_{1}.
\end{aligned}
\end{equation}
To handle the product of operators $h^{(s_{1})}_{w_{1}}(\tau_{1})^{(+)}h^{(s_{2})}_{w_{2}}(\tau_{2})^{(+)}...h^{(s_{m})}_{w_{m}}(\tau_{m})^{(+)} \eqcolon H(w,\vec{s})$, we introduce the concept of partition below.  Let $w=w^{\ri}\circ w^{\rii}$ be a partition of the sequence $w$, such that for $i\in 1,2,...,m$, either $w_{i}^{\ri}=w_{i}, w^{\rii}_{i}=\id$ or $w_{i}^{\ri}=\id, w^{\rii}_{i}=w_{i}$. Naturally, we have either $h_{w_{i}^{\ri}}^{(s_{i})}=h_{w_{i}}^{(s_{i})}, h_{w_{i}^{\rii}}^{(s_{i})}=\id$ or $h_{w_{i}^{\ri}}^{(s_{i})}=\id, h_{w_{i}^{\rii}}^{(s_{i})}=h_{w_{i}}^{(s_{i})}$. By defining $H(w^{\ri},\vec{s})=h_{w_{1}^{\ri}}^{(s_{1})}(\tau_{1})h_{w_{2}^{\ri}}^{(s_{2})}(\tau_{2})...h_{w_{m}^{\ri}}^{(s_{m})}(\tau_{m})$ and $H(w^{\rii},\vec{s})=h_{w_{1}^{\rii}}^{(s_{1})}(\tau_{1})h_{w_{2}^{\rii}}^{(s_{2})}(\tau_{2})...h_{w_{m}^{\rii}}^{(s_{m})}(\tau_{m})$, we obtain
\begin{equation}\label{}
\begin{aligned}[b]
H(w,\vec{s})=\sum_{\text{parition}}H(w^{\ri},\vec{s})\otimes H(w^{\rii},\vec{s}),
\end{aligned}
\end{equation}
which enables us to arrive at [cf.\,Eq.\,(\ref{tr_vtil_g_tilF})], 
\begin{equation}\label{sum_parition}
\begin{aligned}[b]
&\norm{e^{-\beta W_{\widetilde{V}_{w}}^{(+)}}h^{(s_{1})}_{w_{1}}(\tau_{1})^{(+)}h^{(s_{2})}_{w_{2}}(\tau_{2})^{(+)}...h^{(s_{m})}_{w_{m}}(\tau_{m})^{(+)}O_{X}^{(0)}O_{Y}^{(1)}}_{1}
\\\leq &\sum_{\text{parition}}\norm{\qty(e^{-\beta W_{\widetilde{V}_{w}}}\otimes e^{-\beta W_{\widetilde{V}_{w}}})\qty[H(w^{\ri},\vec{s})\otimes H(w^{\rii},\vec{s})]\qty(O_{X}O_{Y}\otimes \id- O_{X}\otimes O_{Y})}_{1}
\\=&\sum_{\text{parition}} \bigg\{ \norm{e^{-\beta W_{\widetilde{V}_{w}}}H(w^{\ri},\vec{s})O_{X}O_{Y}}_{1}\cdot \norm{e^{-\beta W_{\widetilde{V}_{w}}}H(w^{\rii},\vec{s})}_{1}
+ \norm{e^{-\beta W_{\widetilde{V}_{w}}}H(w^{\ri},\vec{s})O_{X}}_{1}\norm{e^{-\beta W_{\widetilde{V}_{w}}}H(w^{\rii},\vec{s})O_{Y}}_{1}\bigg\}.
\end{aligned}
\end{equation}
Obviously, there are $2^{m}$ such partitions, and we focus on the first term in the summation on the right-hand side of Eq.\,\eqref{sum_parition}, 
\begin{equation}\label{trace_norm_exp_H_w_ri}
\begin{aligned}[b]
\norm{e^{-\beta W_{\widetilde{V}_{w}}}H(w^{\ri},\vec{s})O_{X}O_{Y}}_{1}&=\norm{\prod_{x\in \widetilde{V}_{w}}\qty[h_{x}(w^{\ri},\vec{s})]\prod_{x\in \widetilde{V}_{w}}e^{O_{x}(n_{x})} \cdot O_{X}O_{Y}}_{1}
\\&\leq \qty(\operatorname{Tr}_{\widetilde{V}_{w}}\qty{ \prod_{x\in \widetilde{V}_{w}}\qty[h_{x}(w^{\ri},\vec{s})^{\dagger}h_{x}(w^{\ri},\vec{s})]e^{O_{x}(n_{x})}}\cdot \operatorname{Tr}_{\widetilde{V}_{w}}\qty{\qty|O_{X}O_{Y}|^{2} \prod_{x\in \widetilde{V}_{w}}e^{O_{x}(n_{x})}})^{1/2}.
\end{aligned}
\end{equation}
For the second term, we denote $o_{x}\coloneq \mathsf{N}_{x}(O_{X}O_{Y})$ for simplification and bound it as 
\begin{equation}\label{tr_v_w_ox_oy_2}
\begin{aligned}[b]
&\operatorname{Tr}_{\widetilde{V}_{w}}\qty{\qty|O_{X}O_{Y}|^{2} \prod_{x\in \widetilde{V}_{w}}e^{O_{x}(n_{x})}}
\leq \prod_{x\in \widetilde{V}_{w}}\qty(\frac{2C_{1}}{\sqrt{\beta}})^{o_{x}+1}o_{x}!
\\\leq & \qty(\frac{2C_{1}}{\sqrt{\beta}})^{\mathsf{N}(O_{X})+\mathsf{N}(O_{Y})}\qty(\frac{2C_{1}}{\sqrt{\beta}})^{|V_{w}|}\qty[\mathsf{N}(O_{X})+\mathsf{N}(O_{Y})]!
\\\leq & \qty(\frac{2C_{1}}{\sqrt{\beta}})^{\mathsf{N}(O_{X})+\mathsf{N}(O_{Y})}\qty(\frac{2C_{1}}{\sqrt{\beta}})^{|V_{w}|}2^{\mathsf{N}(O_{X})+\mathsf{N}(O_{X})} \mathsf{N}(O_{X})! \mathsf{N}(O_{Y})!
\\= & \qty(\frac{4C_{1}}{\sqrt{\beta}})^{\mathsf{N}(O_{X})+\mathsf{N}(O_{Y})}\qty(\frac{2C_{1}}{\sqrt{\beta}})^{|V_{w}|} \mathsf{N}(O_{X})! \mathsf{N}(O_{Y})!
.
\end{aligned}
\end{equation}
On the other hand [cf.\,Ref,\,\cite{tong2024boson}]
\begin{equation}\label{tr_v_w_h_x_w}
\begin{aligned}[b]
\operatorname{Tr}_{\widetilde{V}_{w}}\qty{ \prod_{x\in \widetilde{V}_{w}}\qty[h_{x}(w^{\ri},\vec{s})^{\dagger}h_{x}(w^{\ri},\vec{s})]e^{O_{x}(n_{x})}}&\leq \prod_{x\in \widetilde{V}_{w}}\qty(\frac{2C_{1}}{\sqrt{\beta}})^{m_{x}(w^{\ri},\vec{s})+1}m_{x}(w^{\ri},\vec{s})!
\\&=\qty(\frac{2C_{1}}{\sqrt{\beta}})^{|\widetilde{V}_{w}|}\prod_{x\in \widetilde{V}_{w}}\qty(\frac{2C_{1}}{\sqrt{\beta}})^{m_{x}(w^{\ri},\vec{s})}m_{x}(w^{\ri},\vec{s})! .
\end{aligned}
\end{equation}
By putting Eqs.\,\eqref{tr_v_w_ox_oy_2} and \eqref{tr_v_w_h_x_w} into Eq.\,\eqref{trace_norm_exp_H_w_ri}, we obtain 
\begin{equation}\label{}
\begin{aligned}[b]
&\norm{e^{-\beta W_{\widetilde{V}_{w}}}H(w^{\ri},\vec{s})O_{X}O_{Y}}_{1}
\\\leq& \qty(\frac{2C_{1}}{\sqrt{\beta}})^{|\widetilde{V}_{w}|} \qty(\frac{4C_{1}}{\sqrt{\beta}})^{[\mathsf{N}(O_{X})+\mathsf{N}(O_{Y})]/2}[\mathsf{N}(O_{X})!]^{1/2} [\mathsf{N}(O_{Y})!]^{1/2}\prod_{x\in \widetilde{V}_{w}}\qty(\frac{2C_{1}}{\sqrt{\beta}})^{m_{x}(w^{\ri},\vec{s})/2}[m_{x}(w^{\ri},\vec{s})!]^{1/2}.
\end{aligned}
\end{equation}
We can easily see that 
\allowdisplaybreaks[4]
\begin{align*}\label{e_beta_W_H_w_i}
&\norm{e^{-\beta W_{\widetilde{V}_{w}}}H(w^{\ri},\vec{s})O_{X}O_{Y}}_{1}\cdot \norm{e^{-\beta W_{\widetilde{V}_{w}}}H(w^{\rii},\vec{s})}_{1}
\\\leq& \qty(\frac{2C_{1}}{\sqrt{\beta}})^{|\widetilde{V}_{w}|} \qty(\frac{4C_{1}}{\sqrt{\beta}})^{[\mathsf{N}(O_{X})+\mathsf{N}(O_{Y})]/2}[\mathsf{N}(O_{X})!]^{1/2} [\mathsf{N}(O_{Y})!]^{1/2}\prod_{x\in \widetilde{V}_{w}}\qty(\frac{2C_{1}}{\sqrt{\beta}})^{m_{x}(w^{\ri},\vec{s})/2}[m_{x}(w^{\ri},\vec{s})!]^{1/2}
\\\quad &\times
\qty(\frac{2C_{1}}{\sqrt{\beta}})^{|V_{w}|} \prod_{x\in \widetilde{V}_{w}}\qty(\frac{2C_{1}}{\sqrt{\beta}})^{m_{x}(w^{\rii},\vec{s})/2}[m_{x}(w^{\rii},\vec{s})!]^{1/2}
\\\leq & \qty(\frac{2C_{1}}{\sqrt{\beta}})^{2|\widetilde{V}_{w}|}\qty(\frac{4C_{1}}{\sqrt{\beta}})^{[\mathsf{N}(O_{X})+\mathsf{N}(O_{Y})]/2}[\mathsf{N}(O_{X})!]^{1/2} [\mathsf{N}(O_{Y})!]^{1/2}\prod_{x\in \widetilde{V}_{w}}\qty(\frac{2C_{1}}{\sqrt{\beta}})^{m_{x}(w,\vec{s})/2}[m_{x}(w,\vec{s})!]^{1/2}
\\\leq & \qty(\frac{2C_{1}}{\sqrt{\beta}})^{2|\widetilde{V}_{w}|}\qty(\frac{4C_{1}}{\sqrt{\beta}})^{[\mathsf{N}(O_{X})+\mathsf{N}(O_{Y})]/2}[\mathsf{N}(O_{X})!]^{1/2} [\mathsf{N}(O_{Y})!]^{1/2}\qty(\frac{2C_{1}}{\sqrt{\beta}})^{m}\prod_{x\in \widetilde{V}_{w}}[m_{x}(w,\vec{s})!]^{1/2}
\\= & \qty(\frac{2C_{1}}{\sqrt{\beta}})^{2|\widetilde{V}_{w}|}\qty(\frac{4C_{1}}{\sqrt{\beta}})^{[\mathsf{N}(O_{X})+\mathsf{N}(O_{Y})]/2}[\mathsf{N}(O_{X})!]^{1/2} [\mathsf{N}(O_{Y})!]^{1/2}\qty(\frac{2C_{1}}{\sqrt{\beta}})^{m}\prod_{x\in V_{w}}[m_{x}(w,\vec{s})!]^{1/2}.
\refstepcounter{equation}\tag{\theequation}
\end{align*}
Here, in the second line we used $m_{x}(w^{\ri},\vec{s})+m_{x}(w^{\rii},\vec{s})=m_{x}(w,\vec{s})$ and $m_{x}(w^{\ri},\vec{s})!m_{x}(w^{\rii},\vec{s})!\leq m_{x}(w,\vec{s})!$. To obtain the last line we note that $m_{x}(w,\vec{s})=0$ for $x\in \widetilde{V}_{w}\backslash V_{w}$. Obviously, from the similar treatment we can also bound the second term in the summation of Eq.\,\eqref{sum_parition} in the same form as Eq.\,\eqref{e_beta_W_H_w_i}. Note that there are $2^{m}$ partitions in total and we obtain from Eq.\,\eqref{sum_parition} that  
\begin{equation}\label{}
\begin{aligned}[b]
&\norm{e^{-\beta W_{\widetilde{V}_{w}}^{(+)}}h^{(s_{1})}_{w_{1}}(\tau_{1})^{(+)}h^{(s_{2})}_{w_{2}}(\tau_{2})^{(+)}...h^{(s_{m})}_{w_{m}}(\tau_{m})^{(+)}O_{X}^{(0)}O_{Y}^{(1)}}_{1}
\\\leq& \qty(\frac{2C_{1}}{\sqrt{\beta}})^{2|\widetilde{V}_{w}|}\qty(\frac{4C_{1}}{\sqrt{\beta}})^{[\mathsf{N}(O_{X})+\mathsf{N}(O_{Y})]/2}[\mathsf{N}(O_{X})!]^{1/2} [\mathsf{N}(O_{Y})!]^{1/2}\qty(\frac{8C_{1}}{\sqrt{\beta}})^{m}\prod_{x\in V_{w}}[m_{x}(w,\vec{s})!]^{1/2}.
\end{aligned}
\end{equation}
By putting this inequality into Eq.\,\eqref{tr_vtil_g_tilF}, we arrive at
\begin{equation}\label{}
\begin{aligned}[b]
&\norm{e^{-\beta W^{(+)}_{\widetilde{V}_{w}}}\widetilde{F}(w)O_{X}^{(0)}O_{Y}^{(1)}}_{1}
\\\leq& \qty(\frac{2C_{1}}{\sqrt{\beta}})^{2|\widetilde{V}_{w}|}\qty(\frac{4C_{1}}{\sqrt{\beta}})^{[\mathsf{N}(O_{X})+\mathsf{N}(O_{Y})]/2}[\mathsf{N}(O_{X})!]^{1/2} [\mathsf{N}(O_{Y})!]^{1/2}\frac{\qty(8C_{1}\sqrt{\beta})^{m}}{m!} \sum_{s_{1}s_{2}...s_{m}}\prod_{l=1}^{m}|J_{w_{l}}^{(s_{l})}|\cdot  \prod_{x\in V_{w}}[m_{x}(w,\vec{s})!]^{1/2}
\end{aligned}
\end{equation}

{~}

\hrulefill{\bf [ End of Proof of Lemma~\ref{lemma_trace_Ftil_w}]}

{~}

\begin{remark}
For free bosons,
\begin{equation}\label{}
\begin{aligned}[b]
&\norm{e^{-\beta W^{(+)}_{\widetilde{V}_{w}}}\widetilde{F}(w)O_{X}^{(0)}O_{Y}^{(1)}}_{1}
\leq \Phi(\beta) \qty(\frac{c_{1}}{\beta})^{2|V_{w}|}c_{2}^{[\mathsf{N}(O_{X})+\mathsf{N}(O_{Y})]/2}\frac{c_{3}^{m}}{m!}\sum_{s_{1}s_{2}...s_{m}}\prod_{l=1}^{m}|J_{w_{l}}^{(s_{l})}|\cdot  \prod_{x\in V_{w}}[m_{x}(w,\vec{s})!]^{1/2}
\end{aligned}
\end{equation}
\end{remark}

\end{proof}

\begin{lemma}[Lemma 2 of Ref.\,\cite{kim2025thermal}]\label{lemma_repro}
Letthe $|V|\times |V|$ matrix $\bm{L}\coloneq \{L_{i,j}\}_{i,j\in V}$ be such that $L{i,j}\leq g (1+d_{i,j})^{-\as}$ then the following inequality holds,
\begin{equation}\label{}
\begin{aligned}[b]
\qty[\bm{L}^{\ell}]_{i,j}\leq \frac{g^{\ell}(2^{\as}\nu)^{\ell-1}}{(1+d_{i,j})^{\as}}=(2^{\as}g\nu)^{\ell-1}L_{i,j}.
\end{aligned}
\end{equation}
Here, the finiteness of $\nu \coloneq \sup_{i\in V}\sum_{j\in V}(1+d_{i,j})^{-\alpha}$ is guaranteed by the condition $\as>D$.
\end{lemma}

One of the most significant steps in establishing the clustering theorem is the following proposition concerning cluster summation. It was initiated as Proposition 4 in Ref.\ cite{kuwahara2020gaussian}, but an error in the proof had been pointed out (We would like to thank Jeongwan Haah for pointing out this issue). Here we present the corrected proof. This result appears ubiquitously when employing the cluster expansion technique:
\begin{proposition}[Cluster Summation]\label{pro_cluster_sum}
Let $V$ be a finite vertex set and $E$ be the set of all subsets of $V$ with cardinality at most $k$. Denote by $\mathcal{M}(E)$ the set of all multisets with elements from $E$, and by $\mathcal{C}$ the set of all connected multisets. Let $\phi: \{Z\colon Z\subseteq V\}  \mapsto \mathbb{R}_{\geq 0}$ be a function that maps vertex subsets to non-negative real numbers.  Suppose that for any vertex $i \in V$, the following inequality holds: $\sum_{Z\in E: Z \ni i} \phi(Z) \leq \gamma$. For any vertex subset $L \subseteq V$, define
$\mathcal{G}^{E}_{m}(L) \coloneqq \{G \in \mathcal{M}(E) \colon |G| = m, G\oplus L \in \mathcal{C}\}$
as the set of all multisets of size $m$ such that they constitute a connected multiset with $L$. Denote $N(G_{s}|G\oplus L)\coloneq |\{Z\in G\oplus L\colon V_Z\cap V_{G_{s}}\neq \emptyset  \}|$ to represent the number of the elements in $G\oplus L$ overlapping with $G_{s}$, with $G_s$ representing the $s$-th element of the multiset $G$ under some fixed ordering. Then we have
\begin{equation}
\sum_{G \in \mathcal{G}^{E}_{m}(L)} |\mathcal{P}(G)| \prod_{s=1}^{m}N(G_{s}|G\oplus L) \phi(G_{s}) \leq c_{1}^{|L|/k}(c_{2}\gamma k)^{m}m!
\end{equation}
where $|\mathcal{P}(G)|$ is the number of multiset permutations of $G$ and $(c_{1},c_{2})=(e^{405}, e^{1231})$.
\end{proposition}
The proof for this proposition is based on the following lemma:
\begin{lemma}[Lemma for Cluster Summation]\label{lemma_cluster_sum}
Under the same setting of Proposition~\ref{pro_cluster_sum}, let $\overline{G} \in \mathcal{M}(E)$ be a multiset with cardinality $\overline{m}$. For a vertex subset $L \subseteq V$, we define
$\mathcal{I}^{E}_{m}(L) \coloneqq \{G \in \mathcal{M}(E) \colon |G| = m, \, G_{q} \cap L \neq \emptyset, \forall q = 1, 2, \ldots, m\}$
to be the collection of all multisets in $\mathcal{M}(E)$ with cardinality $m$ such that each element of the multiset overlaps with $L$.
Then the following inequality holds
\begin{equation}\label{eqlemma_7}
\begin{aligned}[b]
\sum_{G\in \mathcal{I}^{E}_{m}(\overline{G})}|\mathcal{P}(G)|\prod_{s=1}^{m}N(G_{s}|\overline{G}\oplus G)\phi(G_{s})\leq 2^{5}(4\gamma k)^{m}m^{m}e^{405 \overline{m}}.
\end{aligned}
\end{equation} 
For a multiset $G$, we understand $\mathcal{I}^{E}_{m}(\overline{G})$ as $\mathcal{I}^{E}_{m}(V_{\overline{G}})$.
\end{lemma}
\begin{proof}
The LHS of Eq.\,\eqref{eqlemma_7} is not convenient to deal with, so we first rewrite it as (we omit the constraint like $Z\in E$ whenever no potential confusion arises)
\begin{equation}\label{}
\begin{aligned}[b]
\text{LHS of Eq.\,\eqref{eqlemma_7}}&=\sum_{Z_{1}\cap V_{\overline{G}}\neq \emptyset}\sum_{Z_{2}\cap V_{\overline{G}}\neq \emptyset}...\sum_{Z_{m}\cap V_{\overline{G}}\neq \emptyset}\prod_{s=1}^{m}N(Z_{s}|\overline{G}\oplus Z_{1}\oplus Z_{2}\oplus...\oplus Z_{m})\phi(G_{s})
\\&=\sum_{Z_{1},Z_{2},...Z_{m}\cap V_{\overline{G}}\neq \emptyset}\prod_{s=1}^{m}N(Z_{s}|\overline{G}\oplus Z_{1}\oplus Z_{2}\oplus...\oplus Z_{m})\phi(G_{s}).
\end{aligned}
\end{equation}
Note that we always have the decompositions $N(Z_{s}|\overline{G}\oplus Z_{1}\oplus Z_{2}\oplus...\oplus Z_{m})= N(Z_{s}|\overline{G})+\sum_{t=1}^{m}N(Z_{s}|Z_{t})$ and we denote $\overline{G}=\{Z_{01},Z_{02},...,Z_{0\overline{m}}\}$. To adopt a more compact notation, we allow the index $t$ takes the value from $\{1,2,...,m,01,02,...,0\overline{m}\}\eqcolon \mathbb{M}$, corresponding to $Z_{1},Z_{2},...,Z_{m},Z_{01},Z_{02},...,Z_{0\overline{m}}$. Therefore, we have
\begin{equation}\label{lhs_eqlemma_cluster}
\begin{aligned}[b]
\text{LHS of Eq.\,\eqref{eqlemma_7}}&= \sum_{Z_{1},Z_{2},...Z_{m}\cap V_{\overline{G}}\neq \emptyset} \prod_{s=1}^{m}\sum_{t\in \mathbb{M}}N(Z_{s}|Z_{t})\phi(Z_{s})
=\sum_{Z_{1},Z_{2},...Z_{m}\cap V_{\overline{G}}\neq \emptyset}\sum_{t_{1},t_{2},...,t_{m}\in \mathbb{M}}\prod_{s=1}^{m}N(Z_{s}|Z_{t_{s}})\phi(Z_{s})
\\&=\sum_{t_{1},t_{2},...,t_{m}\in \mathbb{M}}\sum_{Z_{1},Z_{2},...Z_{m}\cap V_{\overline{G}}\neq \emptyset}\prod_{s=1}^{m}N(Z_{s}|Z_{t_{s}})\phi(Z_{s}).
\end{aligned}
\end{equation}
On RHS of Eq.\,\eqref{lhs_eqlemma_cluster}, the factor $N(Z_{s}|Z_{t_{s}})$ equals $1$ if $Z_{s}$ overlaps with $Z_{t_{s}}$ and vanishes otherwise. Therefore, for fixed vector $\vec{t}\in \mathbb{M}^{m}$, the nonzero contribution only comes from configurations $\{Z_{1},Z_{2},...,Z_{m}\}$ that satisfies both $Z_{s}\cap Z_{t_{s}}\neq \emptyset$ and $Z_{s}\cap V_{\overline{G}}\neq \emptyset$ for all $s\in [m]$. 
For an intuitive presentation, we map each vector  $\vec{t}\in \mathbb{M}^{m}$ to an undirected graph. Here, we regard $\mathbb{M}$ as the vertex set and put an edge connecting $i\in [m]$ (note that $[m]\subset \mathbb{M}$) and $j\in \mathbb{M}$ if and only if $t_{i}=j$. The corresponding graph is denoted by $\mathsf{G}$.

Now we turn to investigate the structure of the graph. First For convenience we denote $\mathbb{O}\coloneq \mathbb{M}\backslash [m]= \{01,02,...,0\overline{m}\}$. 
We interpret $t_{i}=j$ as a map $t \colon [m]\mapsto \mathbb{M}$ and proceed to identify the set of vertices in $[m]$ whose trajectories, under repeated application of $t$, terminate in the set $\mathbb{O}$. Namely, we denote $\mathbb{U}_{0}\coloneq \{i\in [m]\colon \exists n\in \mathbb{N}\text{ such that } t^{n}(i)\in \mathbb{O}\}$. The $n$-th iterate $t^n(i) = t(t^{n-1}(i))$ is well-defined for $n \ge 1$ so long as the intermediate vertices $t^k(i)$ for all $k \in \{0, 1, \dots, n-1\}$ lie within the domain $[m]$. At the graphic level, we merge the vertices in $\mathbb{O}$ as a single vertex and denote it by $\bm{0}$. Then by construction, the graph (denoted by $\mathsf{G}_{0}$) with the vertex set $\mathbb{U}_{0}\bigcup \{\bm{0}\}$ and the edge set $\{(i,t_{i})\}_{i\in \mathbb{U}_{0}}$ is a connected graph. By Lemma \ref{lemma_cyclomatic}, we further know that $\mathsf{G}_{0}$ is a tree, since the number of vertices is greater than the number of edges within the graph $\mathsf{G}_{0}$. 
Next, we analyze the structure of the remaining graph. Let $\mathbb{U}_{1} \coloneqq [m] \setminus \mathbb{U}_{0}$ be the set of vertices not connected to the subgraph $\mathsf{G}_{0}$, and let $\mathsf{G}_{1}$ be the subgraph of $\mathsf{G}$ induced by the vertex set $\mathbb{U}_{1}$. The edge set of $\mathsf{G}_{1}$ is given by $\{(i, t_i) \colon i \in \mathbb{U}_{1}\}$. By construction, for any vertex $i \in \mathbb{U}_{1}$, its successor $t_i$ must also belong to $\mathbb{U}_{1}$. 
Consider any connected component $\mathsf{G}'_1$ of $\mathsf{G}_1$ with vertex set $\mathbb{U}'_1 \subseteq \mathbb{U}_1$. Within this component, every vertex $i \in \mathbb{U}'_1$ has exactly one outgoing edge $(i, t_i)$ to a vertex $t_i \in \mathbb{U}'_1$. Consequently, the number of vertices in $\mathsf{G}'_1$ is equal to the number of edges. It follows from Lemma~\ref{lemma_cyclomatic} that such a graph must contain exactly one cycle.

In summary, the graph $\mathsf{G}$ exhibits the following structural properties:
\begin{enumerate}
\item The connected component containing the vertex $\bm{0}$ is a tree and thus contains no cycles.
\item Every other connected component of $\mathsf{G}$ contains exactly one cycle.
\end{enumerate}
We introduce some more notations to characterize these structural properties. Let $q_0$ be the number of connected components in $\mathsf{G}$ that do not contain the vertex $\bm{0}$, so that the total number of connected components is $q_0+1$. We denote the set of components by $\{\mathsf{G}_q\}_{q=0}^{q_0}$, where $\mathsf{G}_0$ is the unique component containing $\bm{0}$.
For each $q \in \{0, 1, \dots, q_0\}$, let $\mathbb{U}_q$ be the vertex set of the component $\mathsf{G}_q$
and let $m_q \coloneq |\mathbb{U}_q|$ denote its order (i.e., the number of vertices)
. A further parameter of interest is the in-degree of the root vertex $\bm{0}$, which we define as the number of vertices that are mapped directly to it $\ell_0 \coloneq |\{i \in \mathbb{U}_0 \colon t(i) = \bm{0}\}|$.
By definition, we have $\ell_0 \leq m_0$. See Fig. \ref{G_decompose} for an intuitive illustration. 
\begin{figure}[h]
\centering                                                     
\includegraphics[width=\linewidth]{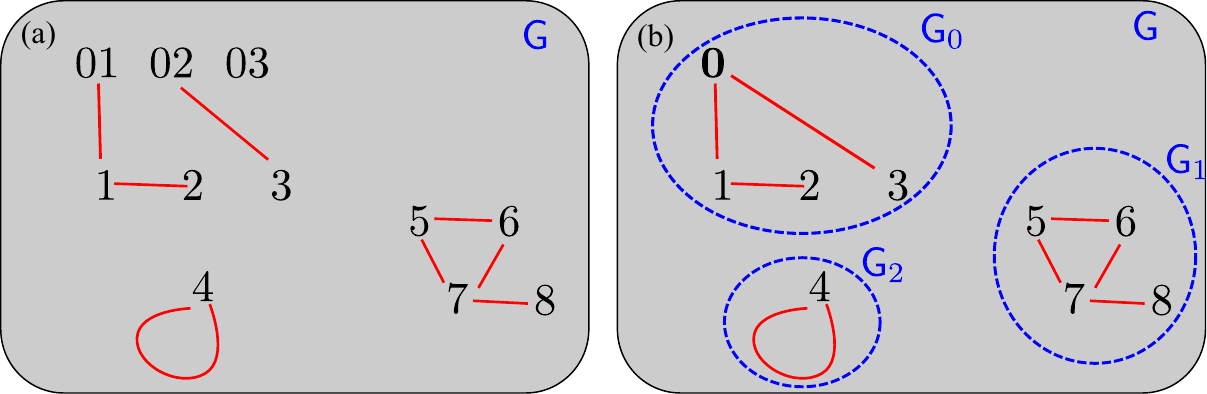}
\caption{We set $(m,\overline{m})=(8,3)$ and $\vec{t}=(01,1,02,4,6,7,5,7)$ to obtain the corresponding graph $\mathsf{G}$ in Fig (a). Then we have the decomposition in Fig. (b) with $(q_{0},\ell_{0})=(2,2)$ and $(m_{0},m_{1},m_{2})=(4,4,1)$. }
\label{G_decompose}
\end{figure}

From Lemmas \ref{sublemma_clustersum_1} and \ref{sublemma_clustersum_2}, we upper bound LHS of Eq.\,\eqref{eqlemma_7} as follows:
\allowdisplaybreaks[4]
\begin{align*}\label{}
&\sum_{G\in \mathcal{I}^{E}_{m}(\overline{G})}|\mathcal{P}(G)|\prod_{s=1}^{m}N(G_{s}|\overline{G}\oplus G)\phi(G_{s})
\\\leq& \sum_{q_{0}=0}^{m}\sum_{m_{0}=0}^{m-q_{0}}\sum_{\ell_{0}=0}^{m_{0}}(\gamma k)^m \overline{m}^{q_0} \times \binom{m}{m_0} \binom{m_{0}}{\ell_{0}}\overline{m}^{\ell_{0}}m_{0}^{m_{0}-\ell_{0}}\frac{(m-m_{0})^{m-m_{0}}2^{-(m-m_{0})+5}e^{m-m_{0}+6q_{0}}}{[q_0^2 / (m-m_0)]^{q_0}}
\\=&2^{5}(\gamma k)^{m}\sum_{q_{0}=0}^{m}\sum_{m_{0}}^{m-q_{0}}\overline{m}^{q_{0}}\binom{m}{m_{0}}(\overline{m}+m_{0})^{m_{0}}(m-m_{0})^{m-m_{0}}2^{-(m-m_{0})+5}e^{m-m_{0}+6q_{0}}\qty(\frac{m-m_{0}}{q_{0}^{2}})^{q_{0}}
\\\leq & 2^{5}(\gamma k)^{m}\sum_{q_{0}=0}^{m}\frac{\overline{m}^{q_{0}}e^{6q_{0}}}{q_{0}^{q_{0}}}\sum_{m_{0}=0}^{m}\binom{m}{m_{0}}(\overline{m}+m_{0})^{m_{0}}(m-m_{0})^{m-m_{0}}\qty(\frac{e^{1+1/e}}{2})^{m-m_{0}}
\\\leq & 2^{5}(\gamma k)^{m}\sum_{q_{0}=0}^{m}\frac{\overline{m}^{q_{0}}e^{6q_{0}}}{q_{0}^{q_{0}}}\sum_{m_{0}=0}^{m}\binom{m}{m_{0}}(\overline{m}+m)^{m_{0}}(m+\overline{m})^{m-m_{0}}\qty(\frac{e^{1+1/e}}{2})^{m-m_{0}}
\\= & 2^{5}(\gamma k)^{m}(\overline{m}+m)^{m}\sum_{q_{0}=0}^{m}\frac{\overline{m}^{q_{0}}e^{6q_{0}}}{q_{0}^{q_{0}}}\qty(\frac{e^{1+1/e}}{2}+1)^{m}
\\\leq  & 2^{5}(\gamma k)^{m}(\overline{m}+m)^{m}\qty(\frac{e^{1+1/e}}{2}+1)^{m}\sum_{q_{0}=0}^{\infty}\frac{\overline{m}^{q_{0}}e^{6q_{0}}}{q_{0}!}\leq 2^{5}(\gamma k)^{m}(\overline{m}+m)^{m}\qty(\frac{e^{1+1/e}}{2}+1)^{m}e^{e^{6}\overline{m}}
\refstepcounter{equation}\tag{\theequation}
\end{align*}
Here, in the second line we used $\sum_{\ell_0=0}^{m_0} \binom{m_0}{\ell_0} m_0^{m_0-\ell_0} \overline{m}^{\ell_0} =   \left( m_0 +\overline{m} \right)^{m_0}$. For the second line we used $[( m-m_{0})/q_{0} ]^{q_0} \le e^{(m-m_0)/e}$. Finally we used $(\overline{m}+m)^{m}\leq m^{m}(1+\overline{m}/m)^{m}\leq m^{m}e^{\overline{m}}$, $e^{1+1/e}/2+1<4$ and $\ln e^{e^6+1}< 405$ to finish the proof of the present Lemma.  
\end{proof}

{~}

\hrulefill{\bf [ End of Proof of Lemma~\ref{lemma_cluster_sum}]}

{~}

\begin{lemma}[Page 27-30 of Ref.\,\cite{berge2001theory}]\label{lemma_cyclomatic}
The cyclomatic number $r$ of a graph is defined as the number of independent cycles in the graph, given by 
$r=e-v+k$,
where $v$ is the number of edges in the given graph, $v$ is the number of vertices, and $k$ is the number of connected components.
\end{lemma}

\begin{lemma}\label{sublemma_clustersum_1}
Under the same setting of Proposition~\ref{pro_cluster_sum},assume that the graph generated by a fixed $\vec{t}\in \mathbb{M}^{m}$ has $q_{0}+1$ connected components, then the summation of $\prod_{s=1}^{m}N(Z_{s}|Z_{t_{s}})\phi(Z_{s})$ under the conditions of $Z_s \cap V_{\overline{G}} \neq \emptyset$ for $s \in [m]$ is upper-bounded as follows
\begin{equation}\label{conclusion_sublemma_clustersum_1}
\begin{aligned}[b]
\sum_{Z_{1},Z_{2},...Z_{m}\cap V_{\overline{G}}\neq \emptyset}\prod_{s=1}^{m}N(Z_{s}|Z_{t_{s}})\phi(Z_{s}) \le (\gamma k)^m \overline{m}^{q_0}.
\end{aligned}
\end{equation}

\end{lemma}
\begin{proof}
In parallel to $N(\bullet|\circ)$, we introduce the notation $\Delta (Z_{s}|V_{\overline{G}})$, whose value is defined to be $1$ if $Z_{s}$ overlaps with $V_{\overline{G}}$ and to be $0$ otherwise. Then LHS of Eq.\,\eqref{conclusion_sublemma_clustersum_1} can be recast in the form of unrestricted summation: 
\begin{equation}\label{}
\begin{aligned}[b]
\text{LHS of Eq.\,\eqref{conclusion_sublemma_clustersum_1}}=\sum_{Z_{1},Z_{2},...,Z_{m}\in E}\prod_{s=1}^{m}N(Z_{s}|Z_{t_{s}})\Delta (Z_{s}|V_{\overline{G}})\phi(Z_{s}).
\end{aligned}
\end{equation}
By assumption, the graph $\mathsf{G}$ generated by has $q_{0}+1$ connected components and we denote the corresponding vertex subsets as $\mathbb{U}_{i}$ for $i\in \{0,1,...,q_{0}\}$. Then by introducing the shorthand $f(Z_{s})\coloneq N(Z_{s}|Z_{t_{s}})\Delta (Z_{s}|V_{\overline{G}})\phi(Z_{s})$, we have
\begin{equation}\label{s_mbU_0_1_2_q}
\begin{aligned}[b]
\text{LHS of Eq.\,\eqref{conclusion_sublemma_clustersum_1}}&=\sum_{Z_{1},Z_{2},...,Z_{m}\in E} \prod_{s\in \mathbb{U}_{0}}f(Z_{s})\times \prod_{s\in \mathbb{U}_{1}}f(Z_{s})\times... \times \prod_{s\in \mathbb{U}_{q_{0}}}f(Z_{s})
\\&=\sum_{\{Z_{i}\in E\}_{i\in \mathbb{U}_{0}}}\prod_{s\in \mathbb{U}_{0}}f(Z_{s})\times \sum_{\{Z_{i}\in E\}_{i\in \mathbb{U}_{1}}}\prod_{s\in \mathbb{U}_{1}}f(Z_{s})\times... \times \sum_{\{Z_{i}\in E\}_{i\in \mathbb{U}_{q_0}}}\prod_{s\in \mathbb{U}_{q_0}}f(Z_{s}),
\end{aligned}
\end{equation}
where we do not introduce different notations for the dummy indices $s$ for clarity. Then we consider the factors on the RHS of Eq.\,\eqref{s_mbU_0_1_2_q} by cases.

(\emph{i}) For any $k\in [q_{0}]$, we know that the graph $\mathsf{G}_{k}$ contains exactly one cycle. We denote the vertex set of this cycle as $\mathbb{C}_{k}$ and then write $|\mathbb{C}_{k}|=c$. Note that in the cycle, we can always relabel $\{Z_{s}\}_{s\in \mathbb{C}_{k}}$ by $\{\oZ_{i}\}_{i\in [c]}$ with the convention $\oZ_{c+1}\equiv \oZ_{1}$. With these notations, we focus on the factor labelled by $\mathbb{U}_{k}$ on the RHS of  Eq.\,\eqref{s_mbU_0_1_2_q},
\begin{equation}\label{u_k_cycle_decom}
\begin{aligned}[b]
&\sum_{\{Z_{i}\in E\}_{i\in \mathbb{U}_{k}}}\prod_{s\in \mathbb{U}_{k}}f(Z_{s})
\\=&\sum_{\{Z_{i}\in E\}_{i\in \mathbb{U}_{k}}}\prod_{s\in \mathbb{C}_{k}}f(Z_{s})\times \prod_{s\in \mathbb{U}_{k}\backslash \mathbb{C}_{k}}f(Z_{s})
\\=&\sum_{\{Z_{i}\in E\}_{i\in \mathbb{U}_{k}}}\prod_{i=1}^{c}N(\oZ_{i}|\oZ_{i+1})\Delta (\oZ_{i}|V_{\overline{G}})\phi(\oZ_{i}) \times \prod_{s\in \mathbb{U}_{k}\backslash \mathbb{C}_{k}}N(Z_{s}|Z_{t_{s}})\Delta (Z_{s}|V_{\overline{G}})\phi(Z_{s})
\\=& \sum_{\oZ_{1},\oZ_{2},...,\oZ_{c}\in E}\prod_{i=1}^{c}N(\oZ_{i}|\oZ_{i+1})\Delta (\oZ_{i}|V_{\overline{G}})\phi(\oZ_{i}) \times \sum_{\{Z_{i}\in E\}_{i\in \mathbb{U}_{k}\backslash \mathbb{C}_{k}}} \prod_{s\in \mathbb{U}_{k}\backslash \mathbb{C}_{k}}N(Z_{s}|Z_{t_{s}})\Delta (Z_{s}|V_{\overline{G}})\phi(Z_{s}). 
\end{aligned}
\end{equation}
In the cyclic structure, i.e., the first factor on the right-hand side of Eq.\,\eqref{u_k_cycle_decom}, without loss of generality, we ``break'' the cycle at the vertex $\oZ_{c}$. Note that we always have $N(\bullet|\circ), \Delta (\bullet| \circ)\leq 1$, which enables us to obtain
\begin{equation}\label{u_k_cycle_decom_cycle}
\begin{aligned}[b]
 \sum_{\oZ_{1},\oZ_{2},...,\oZ_{c}\in E}\prod_{i=1}^{c}N(\oZ_{i}|\oZ_{i+1})\Delta (\oZ_{i}|V_{\overline{G}})\phi(\oZ_{i})&\leq  \sum_{\oZ_{1},\oZ_{2},...,\oZ_{c}\in E} \Delta (\oZ_{1}|V_{\overline{G}})\phi(\oZ_{1})\prod_{i=1}^{c-1}N(\oZ_{i}|\oZ_{i+1})\phi(\oZ_{i})
 \\&=\sum_{\oZ_{1}\cap V_{\overline{G}}\neq \emptyset}\phi(\oZ_{1})\sum_{\oZ_{2}\cap \oZ_{1}\neq \emptyset}\phi(\oZ_{2})... \sum_{\oZ_{c-1}\cap \oZ_{c}\neq \emptyset}\phi(\oZ_{c-1})
 \\&\leq \gamma |V_{\overline{G}}|\cdot (\gamma k)^{c-1}. 
\end{aligned}
\end{equation}
Here, we have already used $\sum_{Z\in E: Z \ni i} \phi(Z) \leq \gamma$ to obtain $\sum_{Z\cap L\neq \emptyset}\phi(Z)\leq \sum_{i\in L}\sum_{Z\ni i}\phi(Z)\leq \gamma |L|$ for any region $L\subset V$.
For the second factor, there is no cycle structure, therefore we simply have
\begin{equation}\label{u_k_cycle_decom_nocycle}
\begin{aligned}[b]
\sum_{\{Z_{i}\in E\}_{i\in \mathbb{U}_{k}\backslash \mathbb{C}_{k}}} \prod_{s\in \mathbb{U}_{k}\backslash \mathbb{C}_{k}}N(Z_{s}|Z_{t_{s}})\Delta (Z_{s}|V_{\overline{G}})\phi(Z_{s})&\leq \sum_{\{Z_{i}\in E\}_{i\in \mathbb{U}_{k}\backslash \mathbb{C}_{k}}} \prod_{s\in \mathbb{U}_{k}\backslash \mathbb{C}_{k}}N(Z_{s}|Z_{t_{s}})\phi(Z_{s})
\\&=\sum_{\{Z_{s}\cap Z_{t_{s}}\neq \emptyset\}_{s\in \mathbb{U}_{k}\backslash \mathbb{C}_{k}}}\prod_{s\in \mathbb{U}_{k}\backslash \mathbb{C}_{k}}\phi(Z_{s})\leq (\gamma k)^{|\mathbb{U}_{k}|-c}.
\end{aligned}
\end{equation}
By combining Eqs.\,\eqref{u_k_cycle_decom_cycle} and \eqref{u_k_cycle_decom_nocycle}, we conclude that
\begin{equation}\label{U_k_result}
\begin{aligned}[b]
\sum_{\{Z_{i}\in E\}_{i\in \mathbb{U}_{k}}}\prod_{s\in \mathbb{U}_{k}}f(Z_{s})\leq \gamma|V_{\overline{G}}|(\gamma k)^{|\mathbb{U}_{k}|-1}
\end{aligned}
\end{equation}
for any $k\in [q_{0}]$.

(\emph{ii}) For the graph $\mathsf{G}_{0}$, we know that it does not contain a cycle, and therefore we simply have
\begin{equation}\label{U_0_result}
\begin{aligned}[b]
\sum_{\{Z_{i}\in E\}_{i\in \mathbb{U}_{0}}}\prod_{s\in \mathbb{U}_{0}}f(Z_{s})\leq \sum_{\{Z_{i}\in E\}_{i\in \mathbb{U}_{0}}}\prod_{s\in \mathbb{U}_{0}}N(Z_{s}|Z_{t_{s}})\phi(Z_{s})=\sum_{\{Z_{s}\cap Z_{t_{s}}\}_{s\in \mathbb{U}_{0}}\neq \emptyset}\prod_{s\in \mathbb{U}_{0}}\phi(Z_{s})\leq (\gamma k)^{|\mathbb{U}_{0}|}.
\end{aligned}
\end{equation}
Note that, for any $s\in \mathbb{U}_{0}$, we have $Z_{t_{s}}\in \{Z_{01},Z_{02},...,Z_{0\overline{m}}\}=\overline{G}$ by construction and all the elements in $\overline{G}$ has the cardinality also at most $k$.

By putting Eqs.\,\eqref{U_k_result} and \eqref{U_0_result} into Eq.\,\eqref{s_mbU_0_1_2_q}, we finally obtain
\begin{equation}\label{}
\begin{aligned}[b]
\text{LHS of Eq.\,\eqref{conclusion_sublemma_clustersum_1}}\leq (\gamma k)^{|\mathbb{U}_{0}|} \prod_{i=1}^{q_{0}}\gamma|V_{\overline{G}}|(\gamma k)^{|\mathbb{U}_{k}|-1}\leq (\gamma k)^{m-q_{0}} |V_{\overline{G}}|^{q_{0}}\leq \gamma^{m} k^{m-q_{0}} (k \overline{m})^{q_{0}}=\text{RHS of Eq.\,\eqref{conclusion_sublemma_clustersum_1}},
\end{aligned}
\end{equation}
where we used $|V_{\overline{G}}|\leq k \overline{m}$ since $|\overline{G}|=m$ and the fact that $m=\sum_{q=0}^{q_{0}}|\mathbb{U}_{q}|$. Thus, we complete the proof.
\end{proof}

{~}

\hrulefill{\bf [ End of Proof of Lemma~\ref{sublemma_clustersum_1}]}

{~}

\begin{lemma}\label{sublemma_clustersum_2}
Under the same setting of Lemma~\ref{lemma_cluster_sum}, the total number of graphs $\vec{t}\in \mathbb{M}^{m}$with fixed $q_0, m_0(\equiv |\mathbb{U}_{0}|)$ and $\ell_0$ is upper-bounded by 
$$
\binom{m}{m_0} \binom{m_{0}}{\ell_{0}}^{m_0-\ell_0}\overline{m}^{\ell_{0}}m_{0}^{m_{0}-\ell_{0}}\frac{(m-m_{0})^{m-m_{0}}2^{-(m-m_{0})+5}e^{m-m_{0}+6q_{0}}}{[q_0^2 / (m-m_0)]^{q_0}}, 
$$
where we notice that $[q_0^2 / (m-m_0)]^{q_0} = 1$ in applying the above upper bound to the case of $q_0=0$ (we have already adopted the convention $0^{0}=1$).
\end{lemma}
\begin{proof}
We first focused on the graph $\mathsf{G}_{0}$ with fixed $m_{0}$ and $\ell_{0}$, the total number of the graphs adhering to such conditions is given by 
\begin{equation}\label{q_0_graph_count}
\begin{aligned}[b]
\binom{m}{m_{0}}\binom{m_{0}}{\ell_{0}}\overline{m}^{\ell_{0}}m_{0}^{m_{0}-\ell_{0}}.
\end{aligned}
\end{equation}
Here, the term $\overline{m}^{\ell_{0}}$ counts the number of ways how $\ell_{0}$ vertices link to $\overline{m}$ vertices (merged as $\bm{0}$). Then the factor $(m_{0}-\ell_{0})^{m_{0}-\ell_{0}}$ counts how the rest of $m_{0}-\ell_{0}$ vertices connect to $m_{0}$ vertices (rather than connecting to $\bm{0}$). If $q_{0}=0$, then 
Eq.\,\eqref{q_0_graph_count} gives the upper bound for the total number of graphs.

For $q_0\geq 1$, we consider the number of combination for subgraphs of $\{\mathsf{G}_q\}_{q=1}^{q_0}$ with $\sum_{q=1}^{q_{0}} |\mathbb{U}_{q}| = m-m_0$.
In the following, let us set $M=m-m_0>0$ for convenience. Let $n_{s}$ be the number of subgraphs in $\{\mathsf{G}_q\}_{q=1}^{q_0}$ containing $s$ vertices, i.e., $n_{s}\coloneq |\{k\in [q_{0}]: |\mathbb{U}_{k}|=s \}|$. Then we naturally have the following equations:
\begin{align}\label{cluster_pattern_s_cond}
\sum_{s=1}^M s n_s  = M ,  \quad
\sum_{s=1}^M n_s  = q_0,
\end{align} 
with $n_{1},n_{2},...,n_{M}$ being non-negative integers. For example, the decomposition in Fig.\,\ref{G_decompose} corresponds to the case of $(q_{0},M,n_{1},n_{4})=(2,8,1,2)$ and $n_{2}=n_{3}=n_{5}=n_{6}=n_{7}=n_{8}=0$. For a fixed $\{n_{1}, n_2,n_3,\ldots,n_M\}$, the number of ways to assign $M$ vertices into this configuration to constitute the subgraphs is given by 
\begin{align}
\label{ineq_patters_of_n_s}
& \binom{M}{n_1} \binom{M-n_1}{2n_2} \frac{(2n_2)!}{(2!)^{n_2} n_2!} (2^{2})^{n_2} \cdots  \binom{M-2n_2-\cdots - (M-1) n_{M-1}}{Mn_M}  \frac{(Mn_M)!}{(M!)^{n_M} n_M!} (M^M)^{n_M}   \notag \\
&=\prod_{s=1}^M\binom{M-n_1-2n_2\cdots - (s-1) n_{s-1}}{s n_s} \frac{(sn_s)!}{(s!)^{n_s} n_s!} s^{sn_s} .
\end{align} 
Note that in the product, the cases $n_{s}=0$ simply contribute $1$ and thus do not affect the final result.
Here, $\binom{M-n_1-2n_2-\cdots - (s-1) n_{s-1}}{s n_{s}} $ is the number of combinations for vertices in the $n_s$ size-$s$ subgraphs.
The combinatorial factor $(sn_s)!/[(s!)^{n_s} n_s!]$ represents the number of ways to partition a set of \(sn_s\) vertices into \(n_s\) indistinguishable subgraphs, where each subgraph contains exactly \(s\) vertices. The term \((sn_s)!\) in the numerator corresponds to the total number of permutations of the vertices. This initial count is then divided by the terms in the denominator to correct for redundancies. The term \((s!)^{n_s}\) corrects for the overcounting of arrangements of vertices within each of the \(n_s\) subgraphs, since the internal ordering of vertices does not define a new partition. On the other hand, the term \(n_s!\) corrects for the overcounting of permutations \emph{among} the \(n_s\) subgraphs themselves, as they are considered indistinguishable.
Finally, in each of the connected subgraphs, the number of patterns of connection is at most $s^s$, which yields $s^{s n_s}$ in total.   
The expression~\eqref{ineq_patters_of_n_s} further reduces to 
\allowdisplaybreaks[4]
\begin{align*}\label{ineq_patters_of_n_s_simple}
\prod_{s=1}^M\binom{M-n_1-2n_2-\cdots - (s-1) n_{s-1}}{s n_s} \frac{(sn_s)!}{(s!)^{n_s} n_s!} s^{sn_s}  
&=\frac{M!}{\prod_{s=1}^M (sn_s)!} \prod_{s=1}^m \frac{(sn_s)!}{(s!)^{n_s} n_s!} s^{sn_s}
\\&=M!  \prod_{s=1}^M \frac{s^{sn_s} }{(s!)^{n_s} n_s!} 
\le M! \prod_{s=1}^M \left( \frac{e^{1+s}}{n_s} \right)^{n_s}
\refstepcounter{equation}\tag{\theequation}
\end{align*}
where the inequality is derived from $x! \ge (x/e)^x$ for all $x\in \mathbb{N}$.
By applying Lemma \ref{lemma_n_s_M_Q} to Eq.\,\eqref{ineq_patters_of_n_s_simple}, we obtain 
\begin{align}
\label{ineq_patters_of_n_s_simple_fin_Pre}
\prod_{s=1}^M\binom{M-n_1-2n_2-\cdots - (s-1) n_{s-1}}{s n_s} \frac{(sn_s)!}{(s!)^{n_s} n_s!} s^{sn_s}  
\le \frac{e^{M+6q_0}  M! }{(q_0^2/M)^{q_0} }  ,
\end{align}
where we use $\sum_s n_s=q_0$ and $\sum_s s n_s=M$. 
By considering all the combinations of $\{n_{1},n_2,n_3,\ldots,n_M\}$, we obtain the upper bound on the patterns of subgraphs with fixed $M$ and $q_0$ as $p(M)e^{M+6q_0}  M!/(q_0^2/M)^{q_0}$. Here, the function $p(M)$ is the partition function (in number theory), representing the number of possible partitions of a non-negative integer $M$. Clearly, the non-negative integer solution $\{n_{1},n_2,n_3,\ldots,n_M\}$ is smaller than $p(M)$. Then, by using the upper bound 
\begin{equation}\label{p_m_ra}
\begin{aligned}[b]
p(M)\leq M^{M}2^{-M+5}/M!,
\end{aligned}
\end{equation}
we complete the proof with noticing $M=m-m_{0}$ by definition. To see the inequality \eqref{p_m_ra} , we first invoke a well-known upper bound $p(M)\leq e^{\pi \sqrt{2M/3}}$ \cite{de2009simple}. It is easy to show that $e^{\pi \sqrt{2M/3}}\leq M^{M}2^{-M+5}/M!$ for $M>70$. Then, for $M\leq 70$, we verify Eq.\,\eqref{p_m_ra} by numerical calculation to finish the proof.

In the case of $q_0=0$, the number of the graph patterns (except for the subgraph $\mathsf{G}_{0}$) is trivially equal to $1$, and hence by noticing $(q_0^2/M)^{q_0}=1$ for $q_0=0$, the above upper bound can also be applied to the case of $q_0=0$. 

\end{proof}
{~}

\hrulefill{\bf [ End of Proof of Lemma~\ref{sublemma_clustersum_2}]}

{~}

\begin{lemma}\label{lemma_n_s_M_Q}
Let $\{n_{s}\}_{s=1}^{M}$ be a set of integers satisfying $\sum_{s=1}^M s n_s  = M$ and $\sum_{s=1}^M n_s  = q_{0}>0$, then the following inequality holds true:
\begin{equation}\label{result_lemma_n_s_M_Q}
\begin{aligned}[b]
\prod_{s=1}^{M}n_{s}^{-n_{s}}\leq \qty(\frac{q_{0}^{2}}{M})^{-q_{0}}e^{5q_{0}}.
\end{aligned}
\end{equation}
Here, we have already adopted the convention $0^{0}=1$.
\end{lemma}
\begin{proof}

For the proof, we first define $s_j = 2^{j} M/q_0$ for $j\geq 1$ and $I_j=[s_{j-1},s_j)$. For $j=0$, we denote $s_0=1$ to exclude the cases where $n_{s}=0$, since they do not contribute to the LHS of Eq.\,\eqref{result_lemma_n_s_M_Q}.
Based on these notations, we define
$
q_j := \sum_{s\in I_j} n_s . 
$
Note that $q_{j}$ is nonzero only for finite range of $j$ by construction (and $M\geq q_{0}$). We simply set $q_{j}=0$ if $M< \inf I_{j}=s_{j-1}$ and therefore we extend the range of $j$ to all positive integers.
Then, because of the condition $\sum_{s\in I_j} s n_s \le M$, we have for $j\ge 2$
\begin{align}
\label{upper_bound_q/j_wefe}
&M \ge \sum_{s\in I_j} s n_s \ge s_{j-1} \sum_{s\in I_j} n_s = 2^{j-1} Mq_j /q_0 \quad   \implies \quad q_j \le \frac{q_0}{2^{j-1}}  , 
\end{align}
where the above upper bound for $q_j$ is trivially applied to $q_1$ since $q_1\le q_0$.  
Also, for arbitrary $s\in I_j$, the lower bound for $\prod_{s\in I_j} n_s^{n_s}$ is given by 
\begin{align}
\label{s_in_l_j_n_s_ns_}
\prod_{s\in I_j} n_s^{n_s} \geq \qty(\frac{q_{j}}{|\mathbb{N}\cap I_{j}|})^{q_{j}} \ge  \left( \frac{q_j}{s_j-s_{j-1}+1} \right)^{q_j} \ge \left( \frac{q_j}{s_j} \right)^{q_j}.
\end{align}
The first inequality is a direct consequence of Jensen's inequality for the convex function $f(x) = x \ln x$, with $q_j = \sum_{s \in I_j} n_s$  and where $|\mathbb{N}\cap I_{j}|$ being the number of total arguments $n_{s}$ involved in the estimation. 
By using the Eq.\,\eqref{s_in_l_j_n_s_ns_}, we obtain 
\begin{align}
\label{desired_lower_bound/_preform}
\prod_{s=1}^m n_s^{n_s} \ge  \prod_{j= 1}^{\infty} \left( \frac{q_j}{s_j} \right)^{q_j},
\end{align}
where in the production on RHS, we have set $0^{0}=1$ to cover the case $q_{j}=0$.  

In estimating the lower bound of $\prod_{j=1}^{\infty} q_j^{q_j}$ under the condition $\sum_{j=1}^{\infty} q_j = q_0$,  
we first let $q_j=c_j q_0$ and obtain 
\begin{align}
\label{q_j_q_j_j_up}
\prod_{j= 1}^{\infty} q_j^{q_j} =q_0^{q_0}  \prod_{j= 1}^{\infty} e^{c_j q_0 \ln c_j} , 
\end{align}
where we use $\sum_{j=1}^{\infty} q_j=q_0$. 
From Eq.\,\eqref{upper_bound_q/j_wefe}, we have $c_j \le 2^{-j+1}$. Also, 
$x\ln (x)$ is minimized for $x=1/e$ and monotonically decreasing for $0\le x\le 1/e$.
Hence, because of $0\le c_j< 1/e$ for $j \ge 3$, we obtain (note that in such a case, $x\log x$ is negative for $x=2^{-j+1}$)  
\begin{align}
\sum_j c_j  \ln c_j \ge 
2\inf_{x\in[0,1]} (x\log x) +  \sum_{j=3}^\infty c_j  \ln c_j  \ge 
-\frac{2}{e} + \sum_{j=3}^\infty 2^{-j+1} \ln (2^{-j+1})=-\frac{2}{e} + \ln (2) \sum_{j=3}^\infty (-j+1) 2^{-j+1}   \ge -2 , 
\end{align}
which reduces Eq.~\eqref{q_j_q_j_j_up} to 
\begin{align}
\label{q_j_q_j_j_up_fin}
\prod_{j= 1}^{\infty} q_j^{q_j} \ge q_0^{q_0} e^{-2 q_0} , 
\end{align}

On the other hand, for $\prod_{j= 1}^{\infty} s_j^{q_j}$, we first obtain
\begin{align}
\label{upper_bound_s_j_q_j_1}
\prod_{j= 1}^{\infty}s_j^{q_j} =\prod_{j= 1}^{\infty}   ( 2^{j} M/ q_0)^{q_j} =
(M/q_0)^{q_0}   2^{\sum_{j=1}^{\infty} j q_j}  .
\end{align}
By using the Eq.\,\eqref{upper_bound_q/j_wefe}, we have 
$\sum_{j=1}^{\infty}  j q_j \le  \sum_{j=1}^{\infty}j q_0/2^{j-1} = 4q_0$,
where we use $\sum_{j=1}^{\infty}j2^{-j+1}=4$. 
Then from Eq.\,\eqref{upper_bound_s_j_q_j_1}, we have 
\begin{align}
\label{upper_bound_s_j_q_j_3}
\prod_{j= 1}^{\infty}s_j^{q_j}  \le
(M/q_0)^{q_0}   2^{4q_0}  .
\end{align}
By applying the equations~\eqref{q_j_q_j_j_up_fin} and~\eqref{upper_bound_s_j_q_j_3} to \eqref{desired_lower_bound/_preform}, we have 
\begin{align}
\label{desired_lower_bound/_preform_2}
\prod_{s=1}^m n_s^{n_s} \ge (q_0^2/M)^{q_0}  2^{-4q_0} e^{-2 q_0}
\ge  (q_0^2/M)^{q_0} e^{-5q_0} .
\end{align}
This completes the proof. 

\end{proof}
{~}

\hrulefill{\bf [ End of Proof of Lemma~\ref{lemma_n_s_M_Q}]}

{~}

Then by Lemma \ref{lemma_cluster_sum} we move on to
\begin{proof}[Proof of Proposition.\,\ref{pro_cluster_sum}]
The primary strategy for estimation is to decompose $G \in \mathcal{G}_{m}^{E}(L)$ into distinct layers based on their distance from $L$. 
Specifically, we write $G\oplus L = X_{0}\oplus X_{1}\oplus X_{2}\oplus \ldots\oplus X_{l}$ with $X_{0} = L$ and $X_{j} \coloneqq \{Z \in G \colon \operatorname{dist}(Z,L) = j\}$ for $j \geq 1$, where the multiset $X_j$ represents the collection of all elements in $G$ whose multiset distance from $L$ equals $j$. Here,
for any two elements $Z, Z' \in G$, their ``multiset distance" $\operatorname{dist}(Z,Z')$ is defined as their graph distance in the dual associated graph of the multiset $G$, which is, the length of the shortest path connecting the corresponding vertices of $Z$ and $Z'$. We use $l \in \{1,2,\ldots,m\}$ to denote the total number of layers.
For illustration, consider the example in Fig.\,\ref{layers_example} (a) and (b), where we have the decomposition:
\begin{equation}\label{sdecomlayer}
\begin{aligned}[b]
X_{1}=\{Z_{1},Z_{2},Z_{3},Z_{4}\}, X_{2}=\{Z_{5},Z_{6},Z_{7},Z_{10}\}, X_{3}=\{Z_{9},Z_{10}\}
\end{aligned}
\end{equation}
for $G = \{Z_{1}, \ldots, Z_{10}\}$.
We extend this notion of distance to multisets: for two multisets $G$ and $G'$, we write $\operatorname{dist}(G,G') = j_{0}$ if $\operatorname{dist}(Z,Z') = j_{0}$ for all $Z \in G$ and $Z' \in G'$. 
For later notational convenience, we also denote $q_{j} \coloneqq |X_{j}|$ as the size of the $j$-th layer and $n_{G} \equiv |\mathcal{P}(G)|$ as the shorthand for the number of multiset permutations over a given multiset $G$.
With these preparations established, we now proceed to the following estimation:
\allowdisplaybreaks[4]
\begin{align*}\label{propo_4}
&\sum_{G\in \mathcal{G}_{m}^{E}(L)}n_{G}\prod_{s=1}^{m}N(G_{s}|G\oplus L)\phi(G_{s})
\\= &\sum_{l=1}^{m}\sum_{\substack{q_{1},q_{2}...q_{l}\geq 1:\\
q_{1}+q_{2}...q_{l}=m}}\sum_{\substack{X_{1}\in \mathcal{G}_{q_{1}}^{E}(L):\\  \operatorname{dist}(X_{1},L)=1}}\sum_{\substack{X_{2}\in \mathcal{G}_{q_{2}}^{E}(L\oplus X_{1}):\\ \operatorname{dist}(X_{2},L)=2}}...\sum_{\substack{X_{l}\in \mathcal{G}_{q_{l}}^{E}(L\oplus X_{1}\oplus X_{2}\oplus...\oplus X_{l-1}):\\ \operatorname{dist}(X_{l},L)=l}}n_{G}\prod_{s=1}^{m}N(G_{s}|G\oplus L)\phi( G_{s})
\\\leq &\sum_{l=1}^{m}\sum_{\substack{q_{1},q_{2}...q_{l}\geq 1:\\
q_{1}+q_{2}...q_{l}=m}}\sum_{X_{1}\in \mathcal{I}^{E}_{q_{1}}(L)}\sum_{X_{2}\in \mathcal{I}^{E}_{q_{2}}(X_{1})}...\sum_{X_{l}\in \mathcal{I}^{E}_{q_{l}}(X_{l-1})}n_{G}\prod_{s=1}^{m}N(G_{s}|G\oplus L)\phi(G_{s})
\\= &\sum_{l=1}^{m}\sum_{\substack{q_{1},q_{2}...q_{l}\geq 1:\\
q_{1}+q_{2}...q_{l}=m}}\sum_{X_{1}\in \mathcal{I}^{E}_{q_{1}}(L)}\sum_{X_{2}\in \mathcal{I}^{E}_{q_{2}}(X_{1})}...\sum_{X_{l}\in \mathcal{I}^{E}_{q_{l}}(X_{l-1})}\frac{n_{G}}{n_{X_{1}}n_{X_{2}}...n_{X_{l}}}\cdot n_{X_{1}}n_{X_{2}}...n_{X_{l}}\prod_{s=1}^{m}N(G_{s}|G\oplus L)\phi(G_{s})
\\= &m!\sum_{l=1}^{m}\sum_{\substack{q_{1},q_{2}...q_{l}\geq 1:\\
q_{1}+q_{2}...q_{l}=m}}(q_{1}!q_{2}!...q_{l}!)^{-1}\sum_{X_{1}\in \mathcal{I}^{E}_{q_{1}}(L)}\sum_{X_{2}\in \mathcal{I}^{E}_{q_{2}}(X_{1})}...\sum_{X_{l}\in \mathcal{I}^{E}_{q_{l}}(X_{l-1})} n_{X_{1}}n_{X_{2}}...n_{X_{l}}\prod_{s=1}^{m}N(G_{s}|G\oplus L)\phi(G_{s})
\\= &m!\sum_{l=1}^{m}\sum_{\substack{q_{1},q_{2}...q_{l}\geq 1:\\
q_{1}+q_{2}...q_{l}=m}}(q_{1}!q_{2}!...q_{l}!)^{-1}
\sum_{X_{1}\in \mathcal{I}^{E}_{q_{1}}(X_{0})}\sum_{X_{2}\in \mathcal{I}^{E}_{q_{2}}(X_{1})}...\sum_{X_{l}\in \mathcal{I}^{E}_{q_{l}}(X_{l-1})}
 n_{X_{1
}}\prod_{s_{1}=1}^{q_{1}}N(G_{s_{1}}^{(1)}|X_{1}\oplus X_{0} \oplus X_{2})\phi(G^{(1)}_{s_{1}})
\\&\quad\times n_{X_{2}}\prod_{s_{2}=1}^{q_{2}}N(G_{s_{2}}^{(2)}|X_{2}\oplus X_{1}\oplus X_{3})\phi(G^{(2)}_{s_{2}})... 
 \times n_{X_{l}}\prod_{s_{l}=1}^{q_{l}}N(G_{s_{l}}^{(l)}|X_{l}\oplus X_{l-1})\phi(G^{(l)}_{s_{l}}).
\refstepcounter{equation}\tag{\theequation}
\end{align*}
Here, in the second line, we decompose all multisets into shells based on their multiset distances. For the fifth line, we used the relation $n_G/(n_{X_1}n_{X_2} \cdots n_{X_l}) = m!/(q_1!q_2! \cdots q_l!)$. This can be interpreted as counting the number of ways to interleave $l$ ordered sequences with fixed lengths $q_1, q_2, \ldots, q_l$ into a longer sequence while preserving their respective internal orders. In the last line of Eq.\,\eqref{propo_4}, we use $G_{s_{p}}^{(p)}$ to represent the $s_{p}$-th element of the multiset $X_{p}$ under some fixed ordering, where $p=1,2,...,m$ and $s_{p}=1,2,...,q_{p}$. We also noticed that for $N(G_{s_{j}}^{(j)}|X_{j}\oplus X_{j-1}\oplus X_{j+1})$ for $j\in [l-1], s_{j}\in [q_{j}]$ and $N(G_{s_{l}}^{(l)}|X_{l}\oplus X_{l-1})$ for $s_{l}\in [q_{l}]$, since only the elements in the nearest-neighbor shells of a given $G_{s}\in G$ contribute to the quantity $N(G_{s}|G\oplus L)$.

Then we denote $\mathcal{M}_{q}(E)\coloneq \{G\in \mathcal{M}(E)\colon |G|=q\}$ as the collection of multisets with $q$ elements to have the following estimation
\allowdisplaybreaks[4]
\begin{align*}\label{prop_4_2}
&\sum_{X_{1}\in \mathcal{I}^{E}_{q_{1}}(X_{0})}\sum_{X_{2}\in \mathcal{I}^{E}_{q_{2}}(X_{1})}...\sum_{X_{l}\in \mathcal{I}^{E}_{q_{l}}(X_{l-1})}
n_{X_{1
}}\prod_{s_{1}=1}^{q_{1}}N(G_{s_{1}}^{(1)}|X_{1}\oplus X_{0} \oplus X_{2})\phi(G^{(1)}_{s_{1}})
\times  n_{X_{2}}\prod_{s_{2}=1}^{q_{2}}N(G_{s_{2}}^{(2)}|X_{2}\oplus X_{1}\oplus X_{3})\phi(G^{(2)}_{s_{2}})... 
\\&\quad \times n_{X_{l}}\prod_{s_{l}=1}^{q_{l}}N(G_{s_{l}}^{(l)}|X_{l}\oplus X_{l-1})\phi(G^{(l)}_{s_{l}})
\\&\leq \sup_{Y_{2}\in \mathcal{M}_{q_{2}}(E)}\sum_{X_{1}\in I^{E}_{q_{1}}(X_{0}\oplus Y_{2})}n_{X_{1
}}\prod_{s_{1}=1}^{q_{1}}N(G_{s_{1}}^{(1)}|X_{1}\oplus X_{0} \oplus Y_{2})\phi(G^{(1)}_{s_{1}})
\sup_{\substack{Y_{1}\in \mathcal{M}_{q_{1}}(E),\\Y_{3}\in \mathcal{M}_{q_{3}}(E)}}\sum_{X_{2}\in \mathcal{I}^{E}_{q_{2}}(X_{1}\oplus X_{3})}n_{X_{2}}
\\&\quad \times\prod_{s_{2}=1}^{q_{2}}N(G_{s_{2}}^{(2)}|X_{2}\oplus Y_{1}\oplus Y_{3})\phi(G^{(2)}_{s_{2}})...
\times \sup_{Y_{l-1}\in \mathcal{M}_{q_{l-1}}(E)}\sum_{X_{l}\in \mathcal{I}^{E}_{q_{l}}(X_{l-1})}n_{X_{l}}\prod_{s_{l}=1}^{q_{l}}N(G_{s_{l}}^{(l)}|X_{l}\oplus Y_{l-1})\phi(G^{(l)}_{s_{l}})
\\&\leq 2^{5}(4\gamma k)^{q_{1}}q_{1}^{q_{1}}c^{|L|/k+1+q_{2}}\times 2^{5}(4\gamma k)^{q_{2}}q_{2}^{q_{2}}c^{q_{1}+q_{3}}...\times 2^{5}(4\gamma k)^{q_{l-1}}q_{l-1}^{q_{l-1}}c^{q_{l-2}+q_{l}}\times 2^{5}(4\gamma k)^{q_{l}}q_{l}^{q_{l}}c^{q_{l-1}}
\\&\leq 32^{l}c^{|L|/k+1}(4\gamma k)^{q_{1}+q_{2}+...+q_{l}}q_{1}^{q_{1}}q_{2}^{q_{2}}...q_{l}^{q_{l}}c^{q_{1}+2(q_{2}+q_{3}+...q_{l-1})+q_{l}}.
\refstepcounter{equation}\tag{\theequation}
\end{align*}
Here, we have introduced Lemma \ref{lemma_cluster_sum} and let $c\coloneq e^{405}$. We also used $|L|/k+1$ elements in $E$ to cover the vertex subset $L$. Then by putting Eq.\,\eqref{prop_4_2} into Eq.\,\eqref{propo_4} we obtain the following 
\allowdisplaybreaks[4]
\begin{align*}\label{}
&\sum_{G\in \mathcal{G}_{m}^{E}(L)}n_{G}\prod_{s=1}^{m}N(G_{s}|G\oplus L)\phi(G_{s})
\\\leq&
m!\sum_{l=1}^{m}\sum_{\substack{q_{1},q_{2}...q_{l}\geq 1:\\
q_{1}+q_{2}...q_{l}=m}}(q_{1}!q_{2}!...q_{l}!)^{-1} 32^{l}c^{|L|/k+1}(4\gamma k)^{q_{1}+q_{2}+...+q_{l}}q_{1}^{q_{1}}q_{2}^{q_{2}}...q_{l}^{q_{l}}c^{q_{1}+2(q_{2}+q_{3}+...q_{l-1})+q_{l}}
\\\leq&
m!\sum_{l=1}^{m}\sum_{\substack{q_{1},q_{2}...q_{l}\geq 1:\\
q_{1}+q_{2}...q_{l}=m}} 32^{l}c^{|L|/k+1}(4\gamma k)^{q_{1}+q_{2}+...+q_{l}}c^{q_{1}+2(q_{2}+q_{3}+...q_{l-1})+q_{l}}e^{q_{1}+q_{2}+...+q_{l}}
\\\leq&
m!c^{|L|/k+1}\sum_{l=1}^{m}32^{l}\sum_{\substack{q_{1},q_{2}...q_{l}\geq 1:\\
q_{1}+q_{2}...q_{l}=m}} (4\gamma k e c^{2})^{m}
\\\leq&
m!c^{|L|/k+1}\sum_{l=1}^{m}\binom{m}{l} 32^{l}(4\gamma k e c^{2})^{m}= m! c^{|L|/k}\cdot c \cdot 33^{m} \cdot (4\gamma k e c^{2})^{m}\leq m! e^{405|L|/k} (e^{1321} \gamma k)^{m}
\refstepcounter{equation}\tag{\theequation}
\end{align*}
where we used $\sum_{\substack{q_{1},q_{2}...q_{l}\geq 1:\\
q_{1}+q_{2}...q_{l}=m}}1=\binom{m-1}{l-1}\leq \binom{m}{l}$.
Thus, we finish the proof for Proposition~\ref{pro_cluster_sum}.
\begin{figure}[h]
\centering                                                     
\includegraphics[width=\linewidth]{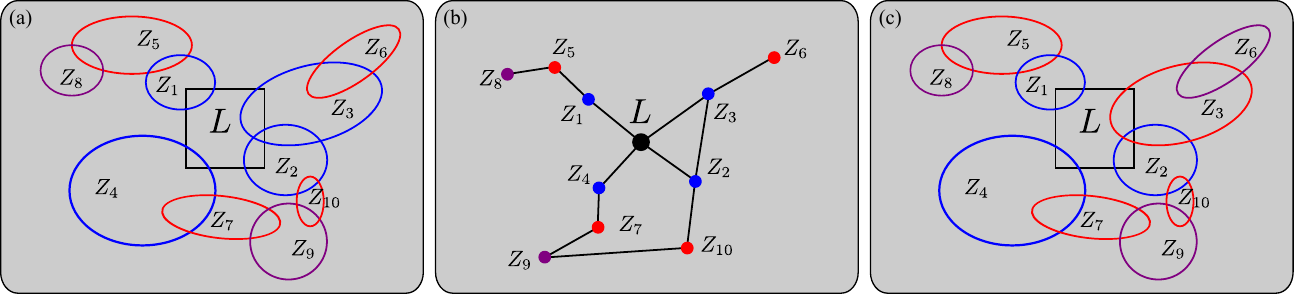}
\caption{Illustration of the decomposition of a multiset into shells.}
\label{layers_example}
\end{figure}
\end{proof}
{~}

\hrulefill{\bf [ End of Proof of Proposition~\ref{pro_cluster_sum}]}

{~}

\begin{lemma}\label{lemma_estimate_m_x}
For a multiset $G\in \mathcal{M}(E)$ and a sequence $w\in \mathcal{S}(G)$ equipped with a fixed vector $\vec{s}$, denote the number of totol bosonic operators at site $x$ offer by it as $m_{x}(w,\vec{s})\coloneq \mathsf{N}_{x}(h^{(s_{1})}_{w_{1}}h^{(s_{2})}_{w_{2}}...h^{(s_{m})}_{w_{m}})$. If for $i=1,2,...,m$ the inequalities $|w_{i}|\leq k$ and $\mathsf{N}_{x}(h^{(s_{i})}_{w_{i}})\leq \mathfrak{g}$ holds, then
\begin{equation}\label{}
\begin{aligned}[b]
\prod_{x\in V_{w}}m_{x}(w,\vec{s})^{\mu_{x}(w)}\leq \prod_{l=1}^{m} \mathfrak{g}^{k}N(w_{l}|w)^{k}.
\end{aligned}
\end{equation} 
Here, $\mu_{x}(w)\coloneq |\{Z\in w| Z\ni x\}|$ denotes the number of regions in the sequence $w$ that contains the site $x$ and $N(w_{l}|w)\coloneq |\{Z\in w| Z\cap w_{l}\neq \emptyset\}|$ is the number regions in the sequence $w$ that overlaps with a fixed region $w_{l}$.
\end{lemma}
\begin{proof}
For a fixed region $w_{i}\in w$, we have for a site $x\in w_{i}$ that
\begin{equation}\label{m_x_w_s}
\begin{aligned}[b]
m_{x}(w,\vec{s})=\sum_{l=1}^{m}\mathsf{N}_{x}(h_{w_{l}}^{(s_{l})})=\sum_{Z\in w: Z\in x}\mathsf{N}_{x}(h_{Z}^{(s_{Z})})\leq \sum_{Z\in w: Z\in x} \mathfrak{g}= \mathfrak{g}\mu_{x}(w)\leq  \mathfrak{g}N(w_{i}|w).
\end{aligned}
\end{equation}
Here, for a region $Z\in w$, the quantity $s_{Z}$ denotes the corresponding value of in the vector $\vec{s}$. Then we have
\begin{equation}\label{sum_x_m_x}
\begin{aligned}[b]
&\sum_{x\in w_{i}}m_{x}(w,\vec{s})\leq \sum_{x\in w_{i}} \mathfrak{g}N(w_{i}|w)\leq  \mathfrak{g}k N(w_{i}|w)
\\\implies & \qty[\sum_{x\in w_{i}}m_{x}(w,\vec{s})]^{k}\leq \qty( \mathfrak{g}k N(w_{i}|w))^{k}
\implies  k^{k}\prod_{x\in w_{i}}m_{x}(w,\vec{s})\leq \qty( \mathfrak{g}k N(w_{i}|w))^{k} \implies \prod_{x\in w_{i}}m_{x}(w,\vec{s}) \leq  \mathfrak{g}^{k}N(w_{i}|w)^{k}.
\end{aligned}
\end{equation}
In Eq.\,\eqref{sum_x_m_x}, we used $(a_{1}+a_{2}+...a_{n})^{n}\geq n^{n}a_{1}a_{2}...a_{n}$. Note that $|w_{i}|\leq k$, we set $a_{l}=m_{x_{l}}(w,\vec{s})$ for $l=1,2,...,|w_{i}|$ and $a_{l}=1$ for $l=|w_{i}|,|w_{i}|+1,...,k$ to utilize this inequality. Equation~\eqref{sum_x_m_x} directly leads us to
\begin{equation}\label{}
\begin{aligned}[b]
\prod_{l=1}^{m}   \mathfrak{g}^{k}N_{w_{l}|w}^{k} \geq \prod_{l=1}^{m}\prod_{x\in w_{l}}m_{x}(w,\vec{s})= \prod_{x\in V_{w}}m_{x}(w,\vec{s})^{\mu_{x}(w)},
\end{aligned}
\end{equation} 
which finishes the proof.
\end{proof}
{~}

\hrulefill{\bf [ End of Proof of Lemma~\ref{lemma_estimate_m_x}]}

{~}
\begin{remark}\label{remark_cluster_convergence}
Given the bound $m_{x}(w,\vec{s})\leq \mathfrak{g} \mu_{x}(w)$, one generally deduces that
\begin{equation}\label{}
\begin{aligned}[b]
\prod_{x\in V_{w}}m_{x}(w,\vec{s})^{m_{x}(w,\vec{s})}\leq \qty[\prod_{x\in V_{w}}m_{x}(w,\vec{s})^{\mu_{x}(w)}]^{ \mathfrak{g}}\leq \prod_{l=1}^{m} \mathfrak{g}^{k \mathfrak{g}}N(w_{l}|w)^{k \mathfrak{g}}.
\end{aligned}
\end{equation}
In Eq.\,\eqref{sum_t_2}, to bound the term $\prod_{x\in V_{w}}\qty[m_{x}(w,\vec{s})^{m_{x}(w,\vec{s})}]^{1/2}$ using Proposition \ref{pro_cluster_sum}, we must impose the condition
\begin{equation}\label{condition_cluster_convergence}
\begin{aligned}[b]
k \mathfrak{g}\leq 2.
\end{aligned}
\end{equation}
To see why this is necessary, consider an extreme case where $L=\{i,j\}$ with $i\neq j$, and a representative candidate multiset in $\mathcal{G}^{E}_{m}(L)$ consisting of $m$ identical copies of $\{i,j\}$. In this scenario, the typical summand 
\begin{equation}\label{}
\begin{aligned}[b]
\sum_{G \in \mathcal{G}^{E}_{m}(L)} |\mathcal{P}(G)| \prod_{s=1}^{m}\mathfrak{g}^{k\mathfrak{g}}N(G_{s}|G\oplus L)^{k\mathfrak{g}/2} \phi(G_{s})
\end{aligned}
\end{equation}
at least scales as $c^{m}m^{k\mathfrak{g}m/2}$ for large $m$, where $c=\mathcal{O}(1)$. Such scaling behavior leads to a divergence in the summation of Eq.\,\eqref{sum_t1_t2}, unless the condition \eqref{condition_cluster_convergence} is satisfied. This condition ensures that the growth factor $m^{k\mathfrak{g}m/2}$ is compensated by the combinatorial factor $m!$ (specifically, $t_{2}!$) in the denominator of Eq.\,\eqref{sum_t_2}.

Notably, for the long-range Bose-Hubbard model, the condition \eqref{condition_cluster_convergence} is naturally satisfied as $(k, \mathfrak{g})=(2,1)$.
\end{remark}

The following result ensures the absolute convergence of the interaction-picture expansion for a long-range bosonic system.
\begin{proposition}\label{pro_abs_boson}
For the long-range Bose-Hubbard model [cf.\,Eq.\,(\ref{slr_bose_hubbard})] defined on the finite lattice $V$.  Define the set $E \coloneqq \{Z \subset V \colon |Z| \leq k\}$ as the collection of all subsets of $V$ with cardinality at most $k$. 
Let $\mathcal{L}$ be a set of sequences whose elements are drawn from $E$, denote $V_{\mathcal{L}}\coloneq \bigcup_{w\in \mathcal{L}}V_{w}$ as the corresponding region.  
Define
\begin{equation}\label{}
\begin{aligned}[b]
F(w)\coloneq  \ii{0}{\beta}{\tau_{1}}\ii{0}{\tau_{1}}{\tau_{2}}...\ii{0}{\tau_{m-1}}{\tau_{m}} h_{w_{1}}(\tau_{1})h_{w_{2}}(\tau_{2})...h_{w_{m}}(\tau_{m}), \quad \bullet(\tau)\coloneq e^{\tau W_{V_{w}}}\bullet e^{-\tau W_{V_{w}}},
\end{aligned}
\end{equation} 
and then we have
\begin{equation}\label{pro_abs_boson_eq_1}
\begin{aligned}[b]
\norm{\sum_{w\in \mathcal{L}}e^{-\beta W_{V_{\mathcal{L}}}}F(w)}_{1}\leq \qty(\frac{c_{1}}{\sqrt{\beta}})^{|V_{\mathcal{L}}|}\frac{e^{c_{2}|V_{\mathcal{L}}|/k}}{1-c_{3}\sqrt{\beta}}< \infty
\end{aligned}
\end{equation}
above a threshold temperature $\beta<\beta_{c}$.
\end{proposition}
\begin{proof}
The proof mainly relies on Lemma~\ref{lem_trace_f_w}, 
\begin{equation}\label{}
\begin{aligned}[b]
\text{LHS of \eq{pro_abs_boson_eq_1}}&\leq \sum_{w\in \mathcal{L}}\norm{e^{-\beta W_{V_{\mathcal{L}}}}F(w)}_{1}=\sum_{w\in \mathcal{L}}\norm{e^{-\beta W_{V_{w}}}F(w)}_{1}\norm{e^{-\beta W_{V_{w}^{\cc}}}}_{1}
\\&=\sum_{m=0}^{\infty} \sum_{w\in \mathcal{L}: |w|=m}\norm{e^{-\beta W_{V_{w}}}F(w)}_{1}\norm{e^{-\beta W_{V_{w}^{\cc}}}}_{1}
\\&\leq \sum_{m=0}^{\infty} \sum_{w\in \mathcal{L}: |w|=m}(C_{0}\beta^{-1/2})^{|V_{w}^{\cc}|}(C_{1}\beta^{-1/2})^{|V_{w}|}\frac{(C_{2}\beta^{1/2})^{m}}{m!}\prod_{i=1}^{m}\qty(\sum_{s_{i}}|J_{w_{i}}^{(s_{i})}|)N(w_{i}|w)
\\&\leq (\max\{C_{0},C_{1}\}\beta^{-1/2})^{|V_{\mathcal{L}}|}\sum_{m=0}^{\infty}\frac{(C_{2}\beta^{1/2})^{m}}{m!}\sum_{G\in \mathcal{G}_{m}^{E}(V_{\mathcal{L}})}|\mathcal{P}(G)|\prod_{l=1}^{m}\qty(\sum_{s_{l}}|J_{w_{l}}^{(s_{l})}|)N(w_{l}|w)
\\&\leq (\max\{C_{0},C_{1}\}\beta^{-1/2})^{|V_{\mathcal{L}}|}\sum_{m=0}^{\infty}\frac{(C_{2}\beta^{1/2})^{m}}{m!}e^{c_{2}|V_{\mathcal{L}}|/k}(C_{3}\gamma k)^{m}m!,
\end{aligned}
\end{equation}
which clearly finishes the proof at sufficient high temperature $\beta<\beta_{c}$.
\end{proof}

{~}

\hrulefill{\bf [ End of Proof of Proposition~\ref{pro_abs_boson}]}

{~}

\begin{remark}
The finiteness for $\norm{\sum_{w\in \mathcal{L}}e^{-\beta W^{(+)}_{V_{\mathcal{L}}}}\widetilde{F}(w)}_{1}$ can be in the same manner under the help of Lemma \ref{lemma_trace_Ftil_w} if we set $O_{X}=\bigotimes_{i\in X}\id_{i}$ and $O_{Y}=\bigotimes_{i\in Y}\id_{i}$.
\end{remark}

\begin{lemma}[Lemma 1 of Ref.\,\cite{mann2024algorithmic}]\label{lemma_inclusion_exclusion}
Let $f$ be a function defined on the subsets of a finite set $E$, then
\begin{equation}\label{}
\begin{aligned}[b]
f(E) = \sum_{S \subseteq E} (-1)^{|S|} \sum_{T \subseteq S} (-1)^{|T|} f(T).
\end{aligned}
\end{equation}
\end{lemma}

\begin{lemma}\label{lemma_inclusion_exclusion_2}
Let $E$ be a finite set and let $S \subseteq E$, then
\begin{equation}\label{}
\begin{aligned}[b]
\sum_{T: E\supseteq T \supseteq S}(-1)^{|T|}
=
\begin{cases}
\begin{aligned}
&(-1)^{|E|}\, && S=E 
\\     
&0\, && S\subsetneq E
\end{aligned}
\end{cases}
\end{aligned}
\end{equation}  
\end{lemma}
\begin{proof}
We consider
\begin{equation}\label{in_ex_lemma_first}
\begin{aligned}[b]
\sum_{T: E\supseteq T \supseteq S}(-1)^{|T|}=\sum_{T: E\supseteq T \supseteq S}(-1)^{|T\backslash S|}=(-1)^{|S|}\sum_{T: E\supseteq T \supseteq S}(-1)^{|T\backslash S|}=(-1)^{|S|}\sum_{T' \subseteq E\backslash S}(-1)^{|T'|}.
\end{aligned}
\end{equation}
Note that we always have, 
\begin{equation}\label{}
\begin{aligned}[b]
\sum_{A\subseteq B}a^{|A|}b^{|B|-|A|}=\sum_{m=0}^{|B|}\binom{|B|}{m}a^{m}b^{n-m}=(a+b)^{|B|},
\end{aligned}
\end{equation}
which leads us to the following if we set $(a,b)=(-1,1)$:
\begin{equation}\label{sum_A_subseteq_B}
\begin{aligned}[b]
\sum_{A\subseteq B}(-1)^{|A|}=\sum_{A\subseteq B}(-1)^{|A|}1^{|B|-|A|}=(1-1)^{|B|}=0.
\end{aligned}
\end{equation}
By using Eq.\,\eqref{sum_A_subseteq_B} we notice that
\begin{equation}\label{}
\begin{aligned}[b]
\sum_{T' \subseteq E\backslash S}(-1)^{|T'|} = \begin{cases}
\begin{aligned}
&0\, && S\subset E 
\\     
&1\, && S=E ,
\end{aligned}
\end{cases}
\end{aligned}
\end{equation}
which together with Eq.\,\eqref{in_ex_lemma_first} finishes the proof.
\end{proof}
{~}

\hrulefill{\bf [ End of Proof of Lemma~\ref{lemma_inclusion_exclusion}]}

{~}
\begin{remark}
Some related discussion can be found in Ref.\,\cite{molnar2015approximating}.
\end{remark}

\begin{lemma}\label{lem_series_sum_estimate}
Let $0<a\leq a_{0}$, $b\in \mathbb{R}$ and $q,p\in \mathbb{N}_{\geq 1}$, then the following inequalities hold
\begin{subequations}\label{}
\begin{align}
& \sum_{n=0}^{q}\exp[\max_{m\in \{n-p..n+p\}}(-am^{2}+bm)]\bigg/\sum_{n=0}^{q}e^{-an^{2}+bn}\leq c_{1}^{p}\label{lem_series_ineq_1}\\
& \sum_{n=0}^{q}(n+p)^{p}\exp[\max_{m\in \{n-p..n+p\}}(-am^{2}+bm)]\bigg/\sum_{n=0}^{q}e^{-an^{2}+bn}\leq c_{2}^{p}a^{-p/2}p! \label{lem_series_ineq_2}
\end{align}
\end{subequations}
with $\{n_{1}..n_{2}\}$ denoting the set of integers between $n_{1},n_{2}\in \mathbb{N}$. Here, the constants are given by 
\begin{subequations}\label{}
\begin{align*}
&c_{1}\coloneq \frac{4}{\sqrt{\pi}}e^{3b^{2}/4a} \qty[3\max\qty{\frac{\sqrt{2a_{0}}}{\erf(\sqrt{2a_{0}})},1}+ 2\sqrt{\pi}] \\
& c_{2}\coloneq \frac{4}{\sqrt{\pi}}e^{3b^{2}/4a}\qty[\qty(6ea_{0}^{1/2}+4e)\max\qty{\frac{\sqrt{2a_{0}}}{\erf(\sqrt{2a_{0}})},1}+(4\sqrt{2}+1)]   .
\end{align*}
\end{subequations}
\end{lemma}
\begin{proof}
For the first inequality, we first note that $-an^{2}+bn\geq -2an^{2}-b^{2}/4a$ and therefore
\begin{equation}\label{lem_series_ineq1_deno}
\begin{aligned}[b]
\sum_{n=0}^{q}e^{-an^{2}+bn}&\geq e^{-b^{2}/4a}\sum_{n=0}^{q}e^{-2an^{2}}\geq e^{-b^{2}/4a}\ii{0}{q+1}{x}e^{-2ax^{2}}=\frac{1}{2}e^{-b^{2}/4a}\sqrt{\frac{\pi}{2a}}\erf(\sqrt{2a}(q+1))
\\&\geq \frac{\sqrt{\pi}}{2}e^{-b^{2}/4a}(2a)^{-1/2}\erf(\sqrt{2a}),
\end{aligned}
\end{equation}
where we note that $q\geq 1$ and the error function $\erf(x)$ is increasing over $\mathbb{R}^{+}$.
Then for the numerator, we first have $-am^{2}+bm\leq -am^{2}/2+b^{2}/2a$ and obtain
\begin{equation}\label{lem_ser_ineq_1_num}
\begin{aligned}[b]
\sum_{n=0}^{q}\exp[\max_{m\in \{n-p..n+p\}}(-am^{2}+bm)]\leq e^{b^{2}/2a}\sum_{n=0}^{q}\exp[\max_{m\in \{n-p..n+p\}}(-am^{2}/2)].
\end{aligned}
\end{equation}
(\emph{i}) If $q\leq p$, then we always have $n-p\leq  0$ and therefore we we upper bound RHS of Eq.\,\eqref{lem_ser_ineq_1_num} by $e^{b^{2}/2a}\sum_{n=0}^{q}1=e^{b^{2}/2a}(q+1)\leq e^{b^{2}/2a}(p+1)$. In such a case, we have
\begin{equation}\label{lem_ser_ineq_1_num_case2}
\begin{aligned}[b]
\text{RHS of Eq.\,\eqref{lem_series_ineq_1}}\leq e^{3b^{2}/4a}\frac{2}{\sqrt{\pi}}\frac{\sqrt{2a}(p+1)}{\erf(\sqrt{2a})}\leq e^{3b^{2}/4a}\frac{2}{\sqrt{\pi}}\frac{\sqrt{2a_{0}}(p+1)}{\erf(\sqrt{2a_{0}})}\leq C_{1}^{p},\quad \quad q\leq p
\end{aligned}
\end{equation}
with the constant $C_{1}\coloneq 4\pi^{-1/2}e^{3b^{2}/4a}\max\{1,\sqrt{2a_{0}}/\erf(\sqrt{2a_{0}})\}$. (\emph{ii}) If $q>p$, we arrive at
\begin{equation}\label{}
\begin{aligned}[b]
\text{RHS of Eq.\,\eqref{lem_ser_ineq_1_num}}&\leq e^{b^{2}/2a}\qty(\sum_{n=0}^{p}1+\sum_{n=p}^{q}\exp[\max_{m\in \{n-p..n+p\}}(-am^{2}/2)])
\\&=e^{b^{2}/2a}\qty(p+1+\sum_{k=0}^{q-p}e^{-ak^{2}/2})\leq e^{b^{2}/2a}\qty[p+2+\sqrt{\frac{2\pi}{a}}\erf\qty(\sqrt{\frac{a}{2}}(q-p))],
\end{aligned}
\end{equation}
where we have used $\sum_{n=0}^{q}e^{-an^{2}}\leq 1+\ii{0}{q}{x}e^{-ax^{2}}$. In this case, by combining Eqs.\,\eqref{lem_series_ineq1_deno} and \eqref{lem_ser_ineq_1_num_case2} we obtain
\begin{equation}\label{lem_ser_ineql_proof_case2}
\begin{aligned}[b]
\text{RHS of Eq.\,\eqref{lem_series_ineq_1}}&\leq \frac{4}{\sqrt{\pi}}e^{3b^{2}/4a}\frac{p+2+\sqrt{\frac{2\pi}{a}}\erf(a(q-p)/2)}{(2a)^{-1/2}\erf(\sqrt{2a}(q+1))}\leq \frac{4}{\sqrt{\pi}}e^{3b^{2}/4a} \qty[\frac{\sqrt{2a}}{\erf(\sqrt{2a})}(p+2)+ 2\sqrt{\pi} \frac{\erf(\sqrt{a/2}(q-p))}{\erf(\sqrt{2a}(q+1))}]
\\&\leq \frac{4}{\sqrt{\pi}}e^{3b^{2}/4a} \qty[\frac{\sqrt{2a_{0}}}{\erf(\sqrt{2a_{0}})}3^{p}+ 2\sqrt{\pi}]\leq C_{2}^{p},
\end{aligned}
\end{equation}
where we note that $\sqrt{a/2}(q-p)\leq \sqrt{2a}(q+1)$ and leverage the monotonic increasing property of error function. We have also used the fact that  the function $x/\erf(x)$ is increasing [see Lemma \ref{lem_increase_function}] over $\mathbb{R}^{+}$ and defined $C_{2}\coloneq 4\pi^{-1/2}e^{3b^{2}/4a}\qty[3\max\{1,\sqrt{2a_{0}}/\erf(\sqrt{2a_{0}})\}+2\pi^{1/2}]$. By choosing $C_{3}\coloneq \max\{C_{1},C_{2}\}$ and combining Eqs.\,\eqref{lem_ser_ineq_1_num_case2} and \eqref{lem_ser_ineql_proof_case2} we finish the proof for the first inequality.

For the second inequality, the only task is the estimation of the numerator. By a similar analysis as above, we arrive at 
\begin{equation}\label{lem_ser_ineq_2_num}
\begin{aligned}[b]
\sum_{n=0}^{q}(n+p)^{p}\exp[\max_{m\in \{n-p..n+p\}}(-am^{2}+bm)]\leq e^{b^{2}/2a}\sum_{n=0}^{q}(n+p)^{p}\exp[\max_{m\in \{n-p..n+p\}}(-am^{2}/2)].
\end{aligned}
\end{equation}
(\emph{i}) If $q\leq p$, we upper bound RHS of Eq.\,\eqref{lem_ser_ineq_2_num} by $e^{b^{2}/2a}\sum_{n=0}^{q}(n+p)^{p}\cdot 1 \leq e^{b^{2}/2a}(2p)^{p}(p+1)\leq e^{b^{2}/2a}(4e)^{p}p!$, where we have used $p^{p}\leq e^{p}p!$. In this case, we simply know from Eq.\,\eqref{lem_series_ineq1_deno} to obtain 
\begin{equation}\label{}
\begin{aligned}[b]
\text{RHS of \eq{lem_series_ineq_2}}\leq   \frac{e^{b^{2}/2a}(4e)^{p}p!}{\sqrt{\pi/2}e^{-b^{2}/4a}(2a)^{-1/2}\erf(\sqrt{2a})}\leq C_{3}^{p}p!,\quad \quad q\leq p
\end{aligned}
\end{equation}
with the constant $C_{3}\coloneq 4\sqrt{2/\pi}e^{3b^{2}/4a+1}\max\{1,\sqrt{2a_{0}}/\erf(\sqrt{2a_{0}})\}$. (\emph{ii}) If $q>p$, we have
\begin{equation}\label{}
\begin{aligned}[b]
\text{RHS of \eq{lem_ser_ineq_2_num}}&\leq e^{b^{2}/2a}\qty[\sum_{n=0}^{p}(n+p)^{p}\cdot 1 +\sum_{n=p}^{q}(n+p)^{p}e^{-a(n-p)^{2}/2}]
\leq e^{b^{2}/2a}\qty[(2p)^{p}(p+1) +\sum_{n=0}^{q-p}(n+2p)^{p}e^{-an^{2}/2}].
\end{aligned}
\end{equation}
Next, we assume that $q\geq 3p$ and we will later see that the resulting upper bound is still applicable for $p<q<3p$. By assumption, we note that
\begin{equation}\label{}
\begin{aligned}[b]
\sum_{n=0}^{q-p}(n+2p)^{p}e^{-an^{2}/2}=\sum_{n=0}^{2p}(n+2p)^{p}e^{-an^{2}/2}+\sum_{n=2p}^{q-p}(n+2p)^{p}e^{-an^{2}/2} \leq (4p)^{p}\cdot 2p +2^{p}\sum_{n=2p}^{q-p}n^{p}e^{-an^{2}/2}.
\end{aligned}
\end{equation}
Note that the function $f(x)\coloneq x^{p}e^{-ax^{2}/2}$ is increasing in $[0, x_{0}]$ and decreasing
in $[x_{0}, 0]$ with $x_{0}=(p/a)^{1/2}$, we have 
\begin{equation}\label{n_sum_p_q}
\begin{aligned}[b]
\sum_{n=2p}^{q-p}n^{p}e^{-an^{2}/2}&\leq 
\begin{cases}
\begin{aligned}
& \ii{0}{q-p}{x}x^{p}e^{-ax^{2}/2}+f(x_{0})\, && q-p\leq x_{0}
\\     
&\ii{0}{q-p}{x}x^{p}e^{-ax^{2}/2}+2f(x_{0})\, && q-p\geq x_{0}
\end{aligned}
\end{cases}
\leq \ii{0}{q-p}{x}x^{p}e^{-ax^{2}/2}+2f(x_{0})
\\&\leq \frac{2^{(p-1)/2}}{a^{(p+1)/2}}\gamma\qty(\frac{p+1}{2},\frac{a}{2}(q-p)^{2})+2\qty(\frac{p}{a})^{p/2}.
\end{aligned}
\end{equation}
Here, we introduced the lower incomplete gamma function $\gamma(s,x)\coloneq \ii{0}{x}{t}t^{s-1}e^{-t}$. Since the RHS of Eq.\,\eqref{n_sum_p_q} is strictly positive, the previously established inequality bounding the sum $\sum_{n=0}^{q-p}(n+2p)^{p}e^{-an^{2}/2}$ remains valid in the regime $p<q<3p$. From Eq.\,\eqref{lem_series_ineq1_deno} we know that 
\begin{equation}\label{}
\begin{aligned}[b]
\text{RHS of \eq{lem_series_ineq_2}}
\leq& 2e^{3b^{2}/4a}\frac{(2p)^{p}(p+1)+(4p)^{p}\cdot 2p+2^{p}\qty[\frac{2^{(p-1)/2}}{a^{(p+1)/2}}\gamma\qty(\frac{p+1}{2},\frac{a}{2}(q-p)^{2})+2\qty(\frac{p}{a})^{p/2}] }{\sqrt{\frac{\pi}{2a}}\erf(\sqrt{2a}(q+1))}
\\=& 2e^{3b^{2}/4a}\qty[\frac{(2p)^{p}(p+1)+(4p)^{p}\cdot 2p+2^{p+1}(\frac{p}{a})^{p/2}}{\sqrt{\frac{\pi}{2a}}\erf(\sqrt{2a}(q+1))}+\frac{\frac{2^{(3p-1)/2}}{a^{(p+1)/2}}\gamma\qty(\frac{p+1}{2},\frac{a}{2}(q-p)^{2})}{\sqrt{\frac{\pi}{2a}}\erf(\sqrt{2a}(q+1))} ]
\\\leq & 2e^{3b^{2}/4a}\qty[\frac{\qty(12ea_{0}^{1/2}+8e)^{p}a^{-p/2}p!}{\sqrt{\frac{\pi}{2a_{0}}}\erf(\sqrt{2a_{0}})}+\frac{2^{3p/2}\Gamma\qty(\frac{p+1}{2})}{\sqrt{\pi}a^{p/2}}\frac{P\qty(\frac{p+1}{2},\frac{a}{2}(q-p)^{2})}{P\qty(\frac{1}{2},2a(q+1)^{2})} ],\quad \quad q> p.
\end{aligned}
\end{equation}
Here in the last line, for the first term in the bracket we again used $\erf(\sqrt{2a}(q+1))\geq \erf(\sqrt{2a})$ and the first claim of Lemma \ref{lem_increase_function}. For the second term, we introduce the regularized gamma function $P(s,x)\coloneq \gamma(s,x)/\Gamma(s)$. Then we know from the second claim of Lemma \ref{lem_increase_function} that $P\qty(\frac{p+1}{2},\frac{a}{2}(q-p)^{2})\leq P\qty(\frac{1}{2},\frac{a}{2}(q-p)^{2})\leq P\qty(\frac{1}{2},2a(q+1)^{2})$, which together with the inequality $\Gamma\qty(\frac{p+1}{2})\leq (4+\sqrt{2})^{p}\sqrt{p!}$ (see Lemma 3 in Ref.\,\cite{tong2024boson}) enables us to move further as follows
\begin{equation}\label{}
\begin{aligned}[b]
\text{RHS of \eq{lem_series_ineq_2}}\leq 2e^{3b^{2}/4a}\qty(C_{4}a^{-p/2}p!+C_{5}a^{-p/2}\sqrt{p!})\leq C_{6}a^{-p/2}p! 
\end{aligned}
\end{equation}
Here, we have defined $C_{4}\coloneq \pi^{-1/2}(12ea_{0}^{1/2}+8e)\max\{1,\sqrt{2a_{0}}/\erf(\sqrt{2a_{0}})\}$, $C_{5}\coloneq 2^{3/2}\pi^{-1/2}(4+\sqrt{2})$ and $C_{6}\coloneq 2e^{3b^{2}/4a}(C_{4}+C_{5})$ to finish the proof.

\end{proof}
{~}

\hrulefill{\bf [ End of Proof of Lemma~\ref{lem_series_sum_estimate}]}

{~}

\begin{lemma}\label{lem_increase_function}
(\emph{i}) The function $x/\erf(x)$ is non-decreasing over $\mathbb{R}^{+}$
(\emph{ii}) Let $\gamma(s,x)\coloneq \ii{0}{x}{t}t^{s-1}e^{-t}$ be the lower incomplete gamma function for $s, x \in \mathbb{R}^{+}$. The regularized gamma function $P(s,x)\coloneq \gamma(s,x)/\Gamma(s)$ is increasing w.r.t $x\in \mathbb{R}^{+}$ for any fixed $s\in \mathbb{R}^{+}$ and decreasing  w.r.t $s\in \mathbb{R}^{+}$ for any fixed $x\in \mathbb{R}^{+}$.
\end{lemma}
\begin{proof}
For the first claim, we denote $a(x)\coloneq x/\erf(x)$ and calculate $a'(x)=[\erf(x)]^{-2}\qty[\erf(x)-2\pi^{-1/2}xe^{-x^{2}}]$. We denote $a_{1}(x)\coloneq \erf(x)-2\pi^{-1/2}xe^{-x^{2}}$ and find $a_{1}'(x)=4x\pi^{-1/2}e^{-x^{2}}\geq 0$, which together with $a_{1}(0)=0$ implies $a'(x)\geq 0$. 

For the second claim, the first part is obvious since $t^{s-1}e^{-t}>0$. For the second part, we denote $f(t)\coloneq t^{s-1}e^{-t}$ and calculate
\begin{equation}\label{}
\begin{aligned}[b]
&\frac{\partial}{\partial s}P(s,x)
\\=&\frac{1}{\Gamma(s)^{2}}\qty[\frac{\partial}{\partial s}\gamma(s,x)\cdot \Gamma(s)-\gamma(s,x) \frac{\partial}{\partial s} \Gamma(s)]=\frac{1}{\Gamma(s)^{2}}\qty[\ii{0}{x}{t}\ln t \cdot f(t)\ii{0}{\infty}{t}f(t)-\ii{0}{x}{t}f(t)\ii{0}{\infty}{t}\ln t \cdot  f(t)]
\\=&\frac{1}{\Gamma(s)^{2}}\qty[\ii{0}{x}{u}\ln u \cdot f(u)\ii{0}{\infty}{v}f(v)-\ii{0}{x}{u}f(u)\ii{0}{\infty}{v}\ln v \cdot  f(v)]
\\=&\frac{1}{\Gamma(s)^{2}}\ii{0}{x}{u}\ii{x}{\infty}{v}\qty[f(u)f(v)\qty(\ln u -\ln v)]< 0,
\end{aligned}
\end{equation}
which completes the proof.

\end{proof}

{~}

\hrulefill{\bf [ End of Proof of Lemma~\ref{lem_increase_function}]}

{~}
\begin{lemma}\label{lem_trace_f_w_q}
Let $w$ be a sequence, denote $V_{w}=\bigcup_{Z\in w}Z$ and $m=|w|$ as the support and length of $w$, respectively. Define
\begin{equation}\label{}
\begin{aligned}[b]
F(w)\coloneq  \ii{0}{\beta}{\tau_{1}}\ii{0}{\tau_{1}}{\tau_{2}}...\ii{0}{\tau_{m-1}}{\tau_{m}} h_{w_{1}}(\tau_{1})h_{w_{2}}(\tau_{2})...h_{w_{m}}(\tau_{m}), \quad \bullet(\tau)\coloneq e^{\tau W_{V_{w}}}\bullet e^{-\tau W_{V_{w}}}.
\end{aligned}
\end{equation} 
If $\{h_{Z}\}_{Z\in w}$ comes from the hopping and squeezing terms in the long-range Bose-Hubbard model [cf.\,Eq.\,(\ref{slr_bose_hubbard})] and write it as $h_{Z}=\sum_{s}J_{Z}^{(s)}h^{(s)}_{Z}$, then for any positive integer $q$ the following inequality holds with $C=\mathcal{O}(1)$ for high temperatures $\beta\leq \beta_{c}$:
\begin{equation}\label{re_lemma_trace_f_w_q}
\begin{aligned}[b]
\qty|\frac{\operatorname{Tr}_{V_{w}}\qty[e^{-\beta W_{V_{w}}}F(w)\Pi_{V_{w},q}]}{\operatorname{Tr}_{V_{w}}\qty(e^{-\beta W_{V_{w}}}\Pi_{V_{w},q})}|\leq \frac{(C \sqrt{\beta})^{m}}{m!}\sum_{s_{1}s_{2}...s_{m}}\prod_{i=1}^{m}|J_{w_{i}}^{(s_{i})}|\cdot  \prod_{x\in V_{w}}[m_{x}(w,\vec{s})!]^{1/2}.
\end{aligned}
\end{equation} 
Here, we introduced the vector $\vec{s}\coloneq (s_{1},s_{2},...,s_{m})$ and denote $\mathsf{N}_{x}(h^{(s_{1})}_{w_{1}}h^{(s_{2})}_{w_{2}}...h^{(s_{m})}_{w_{m}})$ as $m_{x}(w,\vec{s})$. 
\end{lemma}
\begin{proof}
The proof largely follows the proof for Lemma \ref{lem_trace_f_w} and derivations in Appendix E\,2 of Ref.\,\cite{tong2024boson}.

Similar to Eq.\,\eqref{nbarx_l}, we have
\begin{equation}\label{}
\begin{aligned}[b]
e^{-\beta W_{V_{w}}}F(w)\Pi_{V_{w},q}=\int [\tau_{1}\tau_{2}...\tau_{m}] \sum_{s_{1},s_{2},...,s_{m}}|J_{w_{1}}^{(s_{1})}||J_{w_{2}}^{(s_{2})}|...|J_{w_{m}}^{(s_{m})}|\cdot \prod_{x\in V_{w}} \qty[e^{-\beta W_{x}}h_{w_{1},x}^{(s_{1})}(\tau_{1})h_{w_{2},x}^{(s_{2})}(\tau_{2})...h_{w_{m},x}^{(s_{m})}(\tau_{m})\Pi_{x,q}],
\end{aligned}
\end{equation}
which motivates us to write
\begin{equation}\label{re_fw_trace_q_2}
\begin{aligned}[b]
\text{RHS of Eq.\,\eqref{re_lemma_trace_f_w_q}}&\leq \int [\tau_{1}\tau_{2}...\tau_{m}] \sum_{s_{1},s_{2},...,s_{m}}\prod_{i=1}^{m}|J_{w_{i}}^{(s_{i})}|\cdot \prod_{x\in V_{w}}\frac{\|e^{-\beta W_{x}}h_{w_{1},x}^{(s_{1})}(\tau_{1})h_{w_{2},x}^{(s_{2})}(\tau_{2})...h_{w_{m},x}^{(s_{m})}(\tau_{m})\Pi_{x,q}\|_{1}}{\operatorname{Tr}_{x}(e^{-\beta W_{x}}\Pi_{x,q})}
\\&=\int [\tau_{1}\tau_{2}...\tau_{m}] \sum_{s_{1},s_{2},...,s_{m}}\prod_{i=1}^{m}|J_{w_{i}}^{(s_{i})}|\cdot \prod_{x\in V_{w}}\frac{\|h_{x}(w,\vec{s})e^{O_{x}(n_{x})}\Pi_{x,q}\|_{1}}{\operatorname{Tr}_{x}(e^{-\beta W_{x}}\Pi_{x,q})},
\end{aligned}
\end{equation}
where $h_{x}(w,\vec{s})$ and $O_{x}(n_{x})$ have already been defined in Eqs.\,\eqref{e_beta_W_x_h} and \eqref{O_x_n_x}.
Then we focus on the factor on the RHS of Eq.\,\eqref{re_fw_trace_q_2},
\begin{equation}\label{}
\begin{aligned}[b]
\frac{\|h(w,\vec{s})e^{O_{x}(n_{x})}\Pi_{x,q}\|_{1}}{\operatorname{Tr}_{x}(e^{-\beta W_{x}}\Pi_{x,q})}&\leq \frac{\sqrt{\operatorname{Tr}_{x}\qty[h_{x}(w,\vec{s})^{\dagger}h_{x}(w,\vec{s})e^{O_{x}(n_{x})}\Pi_{x,q}] \operatorname{Tr}_{x}[e^{O_{x}(n_{x})}\Pi_{x,q}]}}{\operatorname{Tr}_{x}(e^{-\beta W_{x}}\Pi_{x,q})}
\\&=\sqrt{\frac{\operatorname{Tr}_{x}\qty[h_{x}(w,\vec{s})^{\dagger}h_{x}(w,\vec{s})e^{O_{x}(n_{x})}\Pi_{x,q}]}{\operatorname{Tr}_{x}(e^{-\beta W_{x}}\Pi_{x,q})}} \cdot \sqrt{\frac{\operatorname{Tr}_{x}[e^{O_{x}(n_{x})}\Pi_{x,q}]}{\operatorname{Tr}_{x}(e^{-\beta W_{x}}\Pi_{x,q})}}
\end{aligned}
\end{equation}
We try to deal with these two factors separately. First, similar to the estimate over the contribution of $h_{x}(w,\vec{s})^{\dagger}h_{x}(w,\vec{s})$ in Appdendix E\,2 of Ref.\,\cite{tong2024boson}, we have 
\begin{equation}\label{}
\begin{aligned}[b]
\operatorname{Tr}_{x}\qty[h_{x}(w,\vec{s})^{\dagger}h_{x}(w,\vec{s})e^{O_{x}(n_{x})}\Pi_{x,q}]\leq \operatorname{Tr}_{x}\qty{[n_{x}+m_{x}(w,\vec{s})]^{m_{x}(w,\vec{s})}e^{O_{x}(n_{x})}\Pi_{x,q}}=\sum_{n=0}^{q}[n+m_{x}(w,\vec{s})]^{m_{x}(w,\vec{s})}e^{O_{x}(n)}.
\end{aligned}
\end{equation}
By denoting $\Delta\tau_{1}\coloneq \beta-\tau_{1}$, $\Delta\tau_{m+1}\coloneq \tau_{m}$ and $\Delta\tau_{2}=\tau_{i}-\tau_{i+1}$ for $i=2,3,..,m$. We will also use the following upper bound,
\allowdisplaybreaks[4]
\begin{align*}\label{}
O_{x}(n)
=&-\sum_{i=1}^{m}\Delta\tau_{i}W_{x}(n+S_{i,x})-\Delta\tau_{m+1}W_{x}(n)
\leq-\sum_{i=1}^{m}\Delta\tau_{i}\cdot\min_{i}W(n+S_{i,x})- \Delta\tau_{m+1}W_{x}(n)
\\\leq& -\sum_{i=1}^{m+1}\Delta\tau_{k}\cdot \min\{\min_{i}W_{x}(n+S_{i,x}),W(n)\}
\leq -\beta \min_{y\in \{-m_{x}(w,\vec{s})..m_{x}(w,\vec{s})\}}W_{x}(n+y) 
\\\leq& -\beta \min_{y\in \{n-m_{x}(w,\vec{s})..n+m_{x}(w,\vec{s})\}}W_{x}(y)
\refstepcounter{equation}\tag{\theequation}
\end{align*}
with $\{n_{1}..n_{2}\}$ denoting the set of integers between $n_{1},n_{2}\in y$.
Here, we have used fact that $|S_{i,x}|\leq m_{x}(w,\vec{s})$ for all $i=2,3,...,m$ from the definition of $S_{i,x}$. To proceed, we use Lemma \ref{lem_series_sum_estimate} with choosing $a_{0}=\beta_{c}U_{\max}/2$ and using $b^{2}/a\leq 2\beta_{c}\mu/U_{\min}$ to obtain 
\begin{subequations}\label{Tr_q_O_x_n_x_hdag_h}
\begin{align}
& \sqrt{\frac{\operatorname{Tr}_{x}[e^{O_{x}(n_{x})}\Pi_{x,q}]}{\operatorname{Tr}_{x}(e^{-\beta W_{x}}\Pi_{x,q})}}\leq C_{1}^{m_{x}(w,\vec{s})/2}\\
&\sqrt{\frac{\operatorname{Tr}_{x}\qty[h_{x}(w,\vec{s})^{\dagger}h_{x}(w,\vec{s})e^{O_{x}(n_{x})}\Pi_{x,q}]}{\operatorname{Tr}_{x}(e^{-\beta W_{x}}\Pi_{x,q})}} \leq \qty(\frac{C_{2}}{\sqrt{\beta}})^{m_{x}(w,\vec{s})/2} \sqrt{m_{x}(w,\vec{s})!}
\end{align}
\end{subequations}
with $C_{1},C_{2}$ being $\mathcal{O}(1)$ constants for all $x\in V$. By putting Eq.\,\eqref{Tr_q_O_x_n_x_hdag_h} into Eq.\,\eqref{re_fw_trace_q_2} we obtain
\begin{equation}\label{}
\begin{aligned}[b]
\text{RHS of Eq.\,\eqref{re_lemma_trace_f_w_q}}&\leq\int [\tau_{1}\tau_{2}...\tau_{m}] \sum_{s_{1},s_{2},...,s_{m}}\prod_{i=1}^{m}|J_{w_{i}}^{(s_{i})}|\cdot \prod_{x\in V_{w}}\qty(\frac{C}{\sqrt{\beta}})^{m_{x}(w,\vec{s})/2}\sqrt{m_{x}(w,\vec{s})!}
\\&=\frac{(C \sqrt{\beta})^{m}}{m!}\sum_{s_{1}s_{2}...s_{m}}\prod_{i=1}^{m}|J_{w_{i}}^{(s_{i})}|\cdot  \prod_{x\in V_{w}}[m_{x}(w,\vec{s})!]^{1/2}.
\end{aligned}
\end{equation}
Here, in the last line, we evaluate the integral to finish the proof.
We have also used $\sum_{x\in V_{w}}m_{x}(w,\vec{s})=2|w|=2m$ for the Bose-Hubbard model, where every off-site term is $2$-local and offers exactly one operator to its support.
\end{proof}

{~}

\hrulefill{\bf [ End of Proof of Lemma~\ref{lem_trace_f_w_q}]}

{~}

\section{Discussion on Finite-Range Bosonic Systems}
For finite-range systems, the corresponding condition to Eq.\,\eqref{slr_condition} becomes
\begin{equation}\label{ssr_condition}
\begin{aligned}[b]
\sum_{Z\in E: Z\ni \{i,j\}}|J_{Z}| \leq g\Theta(d_{c}-d_{i,j}), \quad g=\mathcal{O}(1),
\end{aligned}
\end{equation}
where $d_{c}=\mathcal{O}(1)$ is a threshold distance beyond which all interactions vanish, as in \cite{kliesch2014locality,tong2024boson}. As an example, for the finite-range Bose-Hubbard model, by condition \eqref{ssr_condition}, we actually require the hopping amplitude $J_{ij}$ to adhere to
\begin{equation}\label{ssr_condition_bose}
\begin{aligned}[b]
|J_{i,j}|\leq g\Theta(d_{c}-d_{i,j}).
\end{aligned}
\end{equation}
Evidently, for the nearest-neighbor case, we have $d_{c}=2$.

For such finite-range systems, it typically suffices to introduce the interaction hypergraph $(V,E)$ where $V$ and $E$ are the corresponding vertex and (hyper)edge sets, respectively. Here, the (hyper)edge corresponds to the support of each local term in the total Hamiltonian. In this hypergraph, the maximum degree $\mathfrak{d}\coloneqq \sup_{i\in V}|\{Z\in E \colon Z\ni i\}|$, i.e., the maximum number of (hyper)edges adjacent to a given site, remains independent of the system size $|V|$. However, this approach is not viable for general long-range systems, resulting in substantially different analytical treatments for finite-range versus long-range systems. In both finite- and long-range interacting systems, we require a uniformly upper bound for the coefficients of the on-site terms, i.e., condition \eqref{on_site_unbound_unifo}. 

The clustering theorem for finite-range bosonic systems was established in Ref.\,\cite{tong2024boson}. We note that the approach presented here differs significantly from the one in the aforementioned reference. Nevertheless, by re-establishing a bound for the function $L_{i,j}\leq g\Theta(d_{c}-d_{i,j})$ and substituting the result into Eq.\,\eqref{sum_t_1}, we can recover the exponential decay of the correlation function. The argument proceeds as follows.
We first use $\Theta(d_{c}-d_{i,j})\leq ee^{-d_{i,j}/d_{c}}$ and have
\allowdisplaybreaks[4]
\begin{align*}\label{}
\sum_{j\in L}L_{i,j}L_{j,k}&=g^{2}\sum_{j\in \Lambda}\Theta(d_{c}-d_{i,j})\Theta(d_{c}-d_{j,k})
\leq g^{2}\sum_{j\in V:d_{i,j},d_{j,k}
\leq d_{c}}\Theta(d_{c}-d_{i,j})\Theta(d_{c}-d_{j,k})
\\& \leq (eg)^{2}\sum_{j\in V:d_{i,j},d_{j,k}\leq d_{c}} e^{-(d_{i,j}+d_{j,k})/d_{c}}
\leq (eg)^{2}e^{-d_{i,k}/d_{c}}|\{j\in V|d_{i,j},d_{j,k}\leq d_{c}|
\\&\leq (eg)^{2}e^{-d_{i,k}/d_{c}}|\{j\in V|d_{i,j}\leq d_{c}\}|
\leq (eg)^{2}e^{-d_{i,k}/d_{c}}\mathcal{N}_{D}(d_{c}),
\refstepcounter{equation}\tag{\theequation}
\end{align*}
where $\mathcal{N}_{D}(d_{c}) \coloneq \sup_{i\in V}|\{j\in V \mid d_{i,j}\leq d_{c}\}|$ denotes the maximum number of sites within a distance $d_c$ of any given site $i \in V$. This quantity, which depends on the threshold distance $d_c$ and the spatial dimension $D$ of the lattice $V$, is assumed to be a finite constant.
Then similarly, we have
\allowdisplaybreaks[4]
\begin{align*}\label{}
\sum_{j\in V}L_{i,j}[\bm{L}^{2}]_{j,k}&\leq \sum_{j\in V}g\Theta(d_{c}-d_{i,j})(eg)^{2}\mathcal{N}_{D}(d_{c})e^{-d_{j,k}/d_{c}}= \sum_{j\in V: d_{i,j}\leq d}g\Theta(d_{c}-d_{i,j})(eg)^{2}\mathcal{N}_{D}(d_{c})e^{-d_{j,k}/d_{c}}
\\&\leq (eg)^{2}\mathcal{N}_{D}(d_{c})\sum_{j\in V: d_{i,j}\leq d}(eg)e^{-d_{i,j}/d}e^{-d_{j,k}/d_{c}}
\\&\leq (eg)^{2}\mathcal{N}_{D}(d_{c})\sum_{j\in V: d_{i,j}\leq d}(eg)e^{-d_{i,k}/d_{c}}
\\&= (eg)^{3}\mathcal{N}_{D}(d_{c})e^{-d_{i,k}/d_{c}}|\{j\in V| d_{i,j}\leq d_{c}\}| \leq  (eg)^{3}\mathcal{N}_{D}(d_{c})^{2}e^{-d_{i,k}/d_{c}}.
\refstepcounter{equation}\tag{\theequation}
\end{align*}
By iteratively repeating this procedure, we obtain
\begin{equation}\label{}
\begin{aligned}[b]
[\bm{L}^{\ell}]_{i,j}\leq (eg)^{\ell}\mathcal{N}_{D}(d_{c})^{\ell-1}e^{-d_{i,j}/d}.
\end{aligned}
\end{equation}
Substituting this result in Eq.\,\eqref{sum_t_1}, we can clearly see that the original power-law decay is replaced by the exponential decay now.

For the finite-range Bose-Hubbard model, the approximation scheme for the partition function presented above remains applicable. Due to the boundedness of the maximum degree, the runtime to enumerate the elements in $\{a\in \mathscr{G}_{\mathrm{c}}: |a|\leq m\}$ is at most $\exp(\mathcal{O}(m))\times N$, as established in Lemma~6 of Ref.\,\cite{mann2024algorithmic}. Consequently, the total runtime to compute $\log \mathcal{Z}_{W}^{(q)}+T_{m}$ is given by
\begin{equation}\label{eq:runtime_finite_range}
\mathcal{O}(N\log N) + e^{\mathcal{O}(m)} \times \mathcal{O}\left((\log N)^{3km}\right).
\end{equation}
By selecting $m=\mathcal{O}(\log(N/\epsilon))$, this runtime becomes almost polynomial in the system size $N$ and is given by $\exp\left(\mathcal{O}\left(\log(N/\epsilon)\log\log(N/\epsilon)\right)\right)$, since $\log \log N(N/\epsilon)$ is almost a constant due to its very slow growth in $N$. We summarize this result as follows 
\begin{theorem}
For the canonical finite-range Bose-Hubbard model defined by Eqs.~\eqref{slr_bose_hubbard} and \eqref{ssr_condition_bose} with $\widetilde{J}_{i,j}=0$ on a finite lattice $V$. Denote $N=|V|$, then there exists a classical algorithm that, for any given precision $\epsilon>0$ and index $\theta>0$, computes an approximation $f_{\beta}$ satisfying
$|\log \mathcal{Z}_{\beta}-f_{\beta}|\leq N^{-\theta}+\epsilon.$
at any fixed temperature above a threshold, $\beta\leq \beta_{c}=\mathcal{O}(1)$. The runtime of the algorithm is almost polynomial in the system size, given by $\exp(\mathcal{O}(\log((N/\epsilon)\log\log (N/\epsilon)))).$
\end{theorem}
\begin{remark}
The implicit constants in the $\mathcal{O}$-notation for the runtime depend on the chosen temperature $\beta$, threshold temperature $\beta_{c}$, the flexible index $\theta$ and model parameters ($d_{c}, D, g,\mu,U_{\min}$ and $U_{\max}$).
\end{remark}

As in Ref.\,\cite{tong2024boson}, our method can be generalized to accommodate on-site potentials of higher order than quadratic. The details of this extension are straightforward yet tedious, and are thus omitted for brevity.

\hspace{1cm}

\end{widetext}

\end{document}